\documentclass{article}
\usepackage{graphicx} 
\usepackage{amsmath}
\usepackage{amsfonts}
\usepackage{amssymb}  
\usepackage{amsthm}
\usepackage{xcolor}
\usepackage{booktabs}
\usepackage{algorithm}
\usepackage{natbib}
\usepackage{algpseudocode}
\usepackage{longtable}
\usepackage{hyperref}
\usepackage[margin=1in]{geometry}

\newtheorem{theorem}{Theorem}
\newtheorem{proposition}[theorem]{Proposition}
\newtheorem{lemma}{Lemma}

\newtheorem{corollary}{Corollary}

\newtheorem*{assumption*}{Assumption}
\DeclareMathOperator{\Var}{Var}

\DeclareMathOperator{\Cov}{Cov}
\DeclareMathOperator{\Corr}{Corr}

\title{Multivariate incremental effects for continuous treatments: Studying the health effects of environmental mixtures}

\author{%
  Zhuochao Huang\thanks{Department of Statistics, University of Florida. 
    Corresponding author: Zhuochao Huang, \texttt{zhuochao.huang@ufl.edu}}%
  \and Kejin Dong\footnotemark[1]%
  \and Tuo Lin\thanks{Department of Biostatistics, University of Florida.}%
  \and Joseph Antonelli\footnotemark[1]%
}
\date{}
\begin{document}
\maketitle

\begin{abstract}
Evaluating the causal health effects of multivariate, continuous exposures, such as air pollution mixtures, is a critical public health challenge. A primary obstacle is the frequent violation of the positivity assumption, which renders the effects of standard deterministic interventions unidentified or heavily reliant on unreliable model extrapolation. In this paper, we develop a novel causal inference framework to address this challenge. We extend exponential tilting to multivariate exposures and address the critical question of how to compare different intervention directions fairly. This establishes a systematic framework for defining and evaluating various policy-relevant causal estimands, allowing researchers to address diverse scientific questions. We develop numerous methodological advancements, including efficient one-step estimation strategies, a Riemannian BFGS algorithm to solve a constrained manifold optimization problem, semiparametric efficiency bounds for causal estimands, minimax rates for estimators, and asymptotic normality results. We demonstrate our framework's utility by applying it to a nationwide environmental health dataset to identify the optimal strategy for reducing adverse health outcomes associated with a PM$_{2.5}$ chemical mixture.
\end{abstract}

\noindent%

\section{Introduction}

Understanding the causal effects of complex air pollution mixtures is crucial for effective public health policy, but traditional methods face significant challenges. A key challenge is that individuals are exposed to a complex mixture of correlated substances. Single-pollutant analyses often fall short, as they fail to account for the confounding effect of pollutants and cannot capture interaction effects, leading to misleading policy recommendations \citep{bobb2015bayesian, antonelli2024causal}. A primary obstacle in evaluating the causal effects of multivariate continuous exposures is the positivity assumption. This assumption requires that for any given set of covariates, the exposure levels of interest have nonzero density. In environmental research, this assumption is frequently violated as certain combinations of pollutants may be physically implausible or absent from observational data \citep{zigler2021invited}. Consequently, estimating the effect of deterministic interventions, such as setting a pollutant mixture to a specific level, requires unreliable extrapolation \citep{antonelli2024causal, rudolph2025everything}. In the presence of positivity violations, there are (at least) two potential solutions. One approach is to modify the estimation method to be robust to extrapolation, such as ``extrapolation-aware'' methods that bound effects when moving outside the data support \citep{pfister2024extrapolation}. The more common approach, and the one we focus on throughout, is to modify the scientific question by redefining the estimand. 

One useful strategy for modifying the estimand with continuous treatments involves identifying and estimating the derivative of the exposure-response function, then integrating this derivative to recover the full curve. By focusing on local effects first and developing novel bias-corrected estimators for the derivative function, this ``differentiate-then-integrate'' approach can circumvent the need for a global positivity assumption and has been a subject of significant recent research \citep{zhang2024nonparametric, zhang2025doubly}. Another alternative is to leverage instrumental variables, where recent advances have extended these methods to continuous treatments without relying on positivity to identify local effects \citep{rakshit2024local, zeng2025nonparametric}.

While these methods are powerful, a distinct set of approaches that are flexible and policy-relevant rely on the use of stochastic interventions. These interventions generally define a policy that modifies the distribution of exposures rather than replacing it with a fixed value. Early examples of this included interventions that truncate the exposure distribution, such as enforcing pollution levels below a certain cutoff \citep{taubman2009intervening}. This idea was later formalized to provide a general framework for evaluating the population causal effect of shifting an entire exposure distribution \citep{munoz2012population}. One such approach is the use of shift interventions, which define the post-intervention exposure as the natural exposure value shifted by a certain amount \citep{haneuse2013estimation}. Related ideas, referred to as modified treatment policies (MTPs), define the counterfactual exposure as a function of the exposure that would have naturally occurred. This concept was pioneered in the context of dynamic treatment regimes \citep{ezzati2004comparative} and later formalized for single time points \citep{richardson2013single, young2014identification}, and longitudinal scenarios with continuous treatments \citep{diaz2023nonparametric}. There is ongoing research exploring the properties of generalized policies that depend on an individual's natural value of treatment, for example using optimal transport to derive tighter bounds under unmeasured confounding \citep{levis2024stochastic}. This class of estimands generally relies on weaker positivity conditions than most deterministic estimands; however, these estimands still rely on positivity holding for certain exposure values.

When positivity violations are a big concern, one alternative that inherently respects the data's support and thereby does not rely on positivity is the incremental intervention \citep{rothenhausler2019incremental}. This approach was initially developed for binary and longitudinal treatments through incremental shifts in propensity scores, and defines a new interventional distribution by tilting the observed exposure distribution in a specified direction \citep{kennedy2019nonparametric}. This concept has subsequently been generalized to univariate, continuous treatments as well \citep{diaz2020causal, schindl2024incremental}. Some recent work has criticized exponentially tilted estimands such as these \citep{schindl2025causal} due to their less intuitive parameterization, asymmetric reallocation of probability mass, and less favorable asymptotic properties compared to certain alternatives. Regardless of their form, stochastic interventions typically present their own set of challenges. One key concern is that estimands identified under positivity violations may correspond to interventions that are not directly implementable, revealing a fundamental ``interpretability-implementability tradeoff'' \citep{mcclean2025propensity}. 

Nearly all work to date has focused on univariate treatments, which is insufficient for addressing problems raised in the analysis of environmental mixtures. We build on the existing literature by using exponentially tilted estimands within the multivariate treatment context, which is a non-trivial extension that introduces several new, complex questions. First, the intervention parameter becomes a vector, creating an infinite space of possible intervention directions. This raises a variety of questions about how to shift all exposures at once and which direction is best. Second, a fair basis for comparing different shifts is needed; interventions must be constrained to be of the same size for meaningful comparisons. We propose solutions to these issues and explore a variety of estimands that target policy-relevant questions of interest in environmental epidemiology. We show that these different estimands vary in terms of how efficiently they can be estimated from the data, and we provide algorithms for finding optimal policies that lead to the biggest reduction in adverse health outcomes. We provide theoretical support in the form of minimax rates that show how well the estimands can be estimated and how the difficulty of estimation intrinsically depends on the conditional covariance matrix of the exposures. We also develop efficient influence function based estimators that can achieve root-$n$ convergence and asymptotic normality using complex, machine learning estimators for nuisance function estimation.

\section{Exposure Shifts under Exponential Tilts}
\label{sec:Estimands}

Let our observed data consist of $n$ independent and identically distributed samples $\{\boldsymbol{Z}_i\}_{i=1}^n$, drawn from some underlying distribution $P_0$. Each observation $\boldsymbol{Z}_i = (\boldsymbol{X}_i, \boldsymbol{W}_i, Y_i)$ is composed of a $p$-dimensional vector of covariates $\boldsymbol{X}_i \in \mathcal{X} \subseteq \mathbb{R}^p$, a $q$-dimensional vector representing the environmental exposures or treatments\footnote{Note that we use the terms treatment and exposure interchangeably throughout the manuscript. Also note that environmental mixture simply refers to a vector of environmental exposures.} $\boldsymbol{W}_i = (W_{i1}, \dots, W_{iq}) \in \mathcal{W} \subseteq \mathbb{R}^q$, and a scalar outcome of interest $Y_i \in \mathbb{R}$. We denote the conditional density of the exposure mixture given the covariates as $f(\boldsymbol{w} \mid \boldsymbol{x})$. To formally define causal effects, we operate within the potential outcomes framework. For each individual $i$ and any exposure vector $\boldsymbol{w}$, we let $Y_i(\boldsymbol{w})$ denote the potential outcome that would have been observed had individual $i$ received exposure level $\boldsymbol{w}$. This notation relies on the Stable Unit Treatment Value Assumption (SUTVA, \citep{rubin1980randomization}), which has two key principles. It first assumes the potential outcome for one individual is unaffected by the exposure assignments of other individuals (no interference), and also assumes that there are not multiple versions of treatment in the sense that there are not two distinct treatments that lead to the same value of $\boldsymbol{W}_i$. Importantly, this assumption ensures that an individual's observed outcome corresponds to their potential outcome under their observed exposure, meaning that $Y_i = Y_i(\boldsymbol{W}_i)$. 

\subsection{Estimands using Exponential Tilts}

In this work, we extend exponential tilting incremental causal effects to the multivariate treatment setting. This type of stochastic treatment was first proposed for binary treatments \citep{kennedy2019nonparametric} and later generalized to single continuous treatments \citep{diaz2020causal, schindl2024incremental}. We adapt this formulation to define interventions on the entire $q$-dimensional exposure vector, $\boldsymbol{W}$. Given the conditional density of the exposure mixture, $f(\boldsymbol{w} \mid \boldsymbol{x})$, we define an exponentially tilted interventional density, $g_{\boldsymbol{\delta}}(\boldsymbol{w} \mid \boldsymbol{x})$, indexed by a user-specified vector $\boldsymbol{\delta} \in \mathbb{R}^q$:
\begin{equation}
    g_{\boldsymbol{\delta}}(\boldsymbol{w} \mid \boldsymbol{x}) = \frac{\exp(\boldsymbol{\delta}^\top \boldsymbol{w}) f(\boldsymbol{w} \mid \boldsymbol{x})}{\int_{\mathcal{W}} \exp(\boldsymbol{\delta}^\top \boldsymbol{v}) f(\boldsymbol{v} \mid \boldsymbol{x}) \, d\boldsymbol{v}}
    \label{eq:tilt}
\end{equation}
Here, the denominator is a normalizing constant that ensures $g_{\boldsymbol{\delta}}$ integrates to one, and $\boldsymbol{\delta}$ is a vector determining both the direction and magnitude of how the natural density is shifted.

Our causal estimand of interest, the incremental effect, is the expected potential outcome under the stochastic intervention defined by the tilted density $g_{\boldsymbol{\delta}}$. We denote this estimand as $\psi(\boldsymbol{\delta})$:
$$
\psi(\boldsymbol{\delta}) = \mathbb{E}[Y^{g_{\boldsymbol{\delta}}}]
$$
This represents the population average outcome if, for covariates $\boldsymbol{X}=\boldsymbol{x}$, each individual's exposure were a random draw from the shifted distribution $g_{\boldsymbol{\delta}}(\boldsymbol{w} \mid \boldsymbol{x})$. To identify this quantity from the observed data, we assume no unmeasured confounding: for all $\boldsymbol{w}\in\mathcal{W}$,
$$
Y(\boldsymbol{w})\perp \boldsymbol{W}\mid \boldsymbol{X}.
$$
We develop sensitivity analysis approaches to assess violations of this assumption in Section~\ref{sec:SensitivityFinal}. Under SUTVA and no unmeasured confounding, this causal quantity is identified from the observed data distribution via:
\begin{equation}
    \psi(\boldsymbol{\delta}) = \int_{\mathcal{X}} \int_{\mathcal{W}} \mathbb{E}[Y \mid \boldsymbol{W}=\boldsymbol{w}, \boldsymbol{X}=\boldsymbol{x}] \, g_{\boldsymbol{\delta}}(\boldsymbol{w} \mid \boldsymbol{x}) \, d\boldsymbol{w} \, dP(\boldsymbol{x}),
    \label{eq:estimand}
\end{equation}
where $P(\boldsymbol{x})$ denotes the marginal probability measure of the covariates $\boldsymbol{X}$. Note that because the support of $g_{\boldsymbol{\delta}}$ is identical to the support of $f$, we do not have to invoke any positivity assumptions and the intervention does not require us to estimate outcomes for exposure and covariate combinations that are never observed in the data. The parameter $\boldsymbol{\delta}$ can be interpreted as the gradient of the log-likelihood ratio between the interventional and observational densities \citep{schindl2024incremental}:
$$
\boldsymbol{\delta} = \frac{\partial}{\partial \boldsymbol{w}} \log \left( \frac{g_{\boldsymbol{\delta}}(\boldsymbol{w} \mid \boldsymbol{x})}{f(\boldsymbol{w} \mid \boldsymbol{x})} \right).
$$
This means that each component, $\delta_j$, quantifies the change in this log density ratio for an infinitesimal increase in the $j$-th exposure component, $w_j$. Intuitively, setting $\boldsymbol{\delta} = (1, 0, \dots, 0)^\top$ defines an intervention that tilts the distribution to favor higher values of the first exposure, $W_1$.

\subsection{Different Exposure Shifts and Efficiency}
\label{subsec:exposure-shifts-efficiency}

Let $\mu(\boldsymbol{x}, \boldsymbol{w}) = \mathbb{E}[Y \mid \boldsymbol{X}=\boldsymbol{x}, \boldsymbol{W}=\boldsymbol{w}]$. Following the derivation in \citep{schindl2024incremental}, but generalized to the multivariate setting, we obtain the following efficient influence function.

\begin{proposition}
\label{prop:eif}
The efficient influence function of $\psi(\boldsymbol{\delta})$ under a nonparametric model is given by $\varphi(\boldsymbol{Z}; \boldsymbol{\delta}) = D_Y + D_{g,\mu} + D_{\psi},$ where
\begin{gather*}
    D_Y = \frac{g_{\boldsymbol{\delta}}(\boldsymbol{W} \mid \boldsymbol{X})}{f(\boldsymbol{W} \mid \boldsymbol{X})} \left( Y - \mu(\boldsymbol{X}, \boldsymbol{W}) \right) \\
    D_{g,\mu} = \frac{g_{\boldsymbol{\delta}}(\boldsymbol{W} \mid \boldsymbol{X})}{f(\boldsymbol{W} \mid \boldsymbol{X})} \left( \mu(\boldsymbol{X}, \boldsymbol{W}) - \mathbb{E}_{g_{\boldsymbol{\delta}}}[\mu(\boldsymbol{X}, \boldsymbol{W}) \mid \boldsymbol{X}] \right) \\
    D_{\psi} = \mathbb{E}_{g_{\boldsymbol{\delta}}}[\mu(\boldsymbol{X}, \boldsymbol{W}) \mid \boldsymbol{X}] - \psi
\end{gather*}
\end{proposition}
Note the subscript $g_{\boldsymbol{\delta}}$ denotes expectations with respect to the tilted exposure density. The asymptotic variance of any regular and asymptotically linear (RAL) estimator for $\psi(\boldsymbol{\delta})$ is equal to the variance of its EIF, $\mathrm{Var}(\varphi(\boldsymbol{Z}; \boldsymbol{\delta}))$. A critical observation from the structure of $\varphi$ is the repeated appearance of the density ratio, $g_{\boldsymbol{\delta}}(\boldsymbol{W} \mid \boldsymbol{X}) / f(\boldsymbol{W} \mid \boldsymbol{X})$. When this ratio exhibits high variability, it inflates the variance of the first two components of the EIF, leading to less precise estimates of the causal effect. Therefore, to improve statistical efficiency for a given intervention strength, we could select an intervention direction $\boldsymbol{\delta}$ that minimizes the variance of this density ratio, $\mathrm{Var}(g_{\boldsymbol{\delta}}/f)$. The third term, on the other hand, represents the deviation between the average post-intervention effect for a specific covariate group and the overall average post-intervention effect. It quantifies how much higher or lower the expected intervention effect for an individual with covariates $\boldsymbol{X}$ is compared to the overall population average. To gain analytical insights into these terms, we can let the exposures follow a multivariate normal distribution, $f(\boldsymbol{w} \mid \boldsymbol{x}) \sim \mathcal{N}(\boldsymbol{\mu}_x, \boldsymbol{\Sigma})$. Under this assumption, we see that the tilted distribution $g_{\boldsymbol{\delta}}(\boldsymbol{w} \mid \boldsymbol{x})$ is also normal, with a shifted mean:
$$
g_{\boldsymbol{\delta}}(\boldsymbol{w} \mid \boldsymbol{x}) \sim \mathcal{N}(\boldsymbol{\mu}_x + \boldsymbol{\Sigma}\boldsymbol{\delta}, \boldsymbol{\Sigma})
$$
Additionally, in this simplified setting, the variance of the density ratio has a simple, closed-form expression:
$$
\mathrm{Var}\left( \frac{g_{\boldsymbol{\delta}}(\boldsymbol{W} \mid \boldsymbol{X})}{f(\boldsymbol{W} \mid \boldsymbol{X})} \right) = \exp(\boldsymbol{\delta}^\top \boldsymbol{\Sigma} \boldsymbol{\delta}) - 1.
$$
Therefore, if we want to find a direction $\boldsymbol{\delta}$ that minimizes the variance of the density ratio, we can simply find one that minimizes $\boldsymbol{\delta}^\top \boldsymbol{\Sigma} \boldsymbol{\delta}$ subject to a constraint on the ``strength'' of the intervention. While various constraints are possible (e.g., fixing $\boldsymbol{\delta}^\top\boldsymbol{\delta}$), a particularly meaningful constraint is to fix the distance between the original and tilted distributions, for which the 2-Wasserstein distance provides a natural metric. We discuss the choice of this metric in the following section, but for the multivariate normal case with the same covariance matrix, this distance has a simple closed-form expression as $d_W^2(f, g_{\boldsymbol{\delta}}) = \|(\boldsymbol{\mu}_x + \boldsymbol{\Sigma}\boldsymbol{\delta}) - \boldsymbol{\mu}_x\|^2 = \|\boldsymbol{\Sigma}\boldsymbol{\delta}\|^2 = \boldsymbol{\delta}^\top \boldsymbol{\Sigma}^2 \boldsymbol{\delta}$.
It is straightforward to show in this setting that the variance of the density ratio is minimized when $\boldsymbol{\delta}$ is chosen to be proportional to the eigenvector of $\boldsymbol{\Sigma}$ corresponding to its largest eigenvalue, $\lambda_{\max}$. This shows that for a fixed intervention strength (as measured by the Wasserstein distance), the most statistically efficient causal effect to estimate, at least in terms of the density ratio, is the one that corresponds to shifting the exposure mixture along its primary axis of variation. Note that this result only holds exactly under a multivariate normal distribution, but we have seen empirically that this choice of $\boldsymbol{\delta}$ typically leads to efficient estimates of $\psi(\boldsymbol{\delta})$ under a variety of exposure distributions. 

\subsection{Fair Exposure Shifts under Fixed Gelbrich Distance}
\label{subsec:fixed-gelbrich-distance}
A primary motivation for developing our framework is to answer policy-relevant questions, such as identifying the optimal way to modify an environmental exposure mixture to achieve the greatest public health benefit. This naturally leads to an optimization problem: finding the intervention vector $\boldsymbol{\delta}$ that maximizes (or minimizes) the causal estimand $\psi(\boldsymbol{\delta})$. However, a naive comparison across all possible $\boldsymbol{\delta}$ vectors is misleading. For any intervention direction that yields a beneficial effect, one could simply increase the intervention's strength, for example, by scaling the magnitude of $\boldsymbol{\delta}$ to produce an arbitrarily larger (or smaller) value of $\psi(\boldsymbol{\delta})$. This would invariably lead to trivial solutions that correspond to extreme, unrealistic shifts in the exposure distribution. A more relevant policy question is ``for a fixed amount of interventional effort, what is the best direction to apply that effort?'' To formalize this, we must first establish a fair basis for comparison, constraining our search to a set of interventions that are of the same size, which captures the actual movement of the exposure values, not just the magnitude of the parameter $\boldsymbol{\delta}$. Note that the notion of fairness here is with respect to the size of the intervention being applied, which differs from other notions of fairness based on whether an estimand preserves the ordinal ranking of effects across all covariate subgroups \citep{mcclean2024fair}. 

In order to establish our notion of fairness between interventions, we can use the 2-Wasserstein distance between the observed distribution and the post-intervention distribution, which is widely studied and serves this purpose \citep{panaretos2019statistical}. Note that the exponential tilt is on the conditional exposure distribution and therefore there is a different Wasserstein distance for each covariate value. To address this issue, we use the distance between the observed marginal distribution of $\boldsymbol W$ and the marginal distribution induced by $g_{\boldsymbol{\delta}}(\cdot\mid\boldsymbol X)$. Intuitively, the Wasserstein distance measures the minimum cost of transporting the probability mass of one distribution to match another. By fixing the Wasserstein distance, we ensure that the total amount of shift between the old and new exposure distributions is fixed. The optimization problem thus becomes a meaningful search for the $\boldsymbol{\delta}$ that minimizes $\psi(\boldsymbol{\delta})$ among all possible intervention directions of a comparable magnitude. This distributional fairness notion is widely used in both the operations and statistics research literatures \citep{mohajerin2018data, blanchet2019quantifying, duchi2021learning, gao2023distributionally}.

For analytical tractability, we rely on a well-established result providing a formula for the squared 2-Wasserstein distance $d_W^2$ based on the means ($\boldsymbol{\mu}_1, \boldsymbol{\mu}_2$) and covariance matrices ($\boldsymbol{\Sigma}_1, \boldsymbol{\Sigma}_2$) of two distributions $(P_1, P_2)$, which is referred to as the Gelbrich formula \citep{gelbrich1990formula}. Crucially, this formula serves as a general lower bound for the true squared 2-Wasserstein distance between any two probability measures on $\mathbb{R}^q$ with finite second moments. Specifically, we have that
$$
d_W^2\left(P_1, P_2\right) \geq\left\|\boldsymbol{\mu}_1-\boldsymbol{\mu}_2\right\|^2+\operatorname{tr}\left(\boldsymbol{\Sigma}_1+\boldsymbol{\Sigma}_2-2\left(\boldsymbol{\Sigma}_1^{1 / 2} \boldsymbol{\Sigma}_2 \boldsymbol{\Sigma}_1^{1 / 2}\right)^{1 / 2}\right):=d_G^2\left(P_1, P_2\right)
$$
Furthermore, this bound becomes an exact equality for any two distributions belonging to the same family of elliptically contoured distributions, a class which notably includes all multivariate normal distributions as well as uniform distributions on ellipsoids. The tractability of the Gelbrich formula has led various researchers to use it as a surrogate for the 2-Wasserstein distance, a method frequently employed to overcome
the computational complexity associated with the Wasserstein
metric \citep{kuhn2019wasserstein, hakobyan2024wasserstein}, since the empirical 2-Wasserstein distance lacks a closed-form formula from data and can only be approached by numerical methods \citep{panaretos2020invitation}. Moreover, the lower bound it provides has been shown to be tight in a fairly general situation against an upper bound derived for the 2-Wasserstein distance \citep{biswas2024bounding, papp2025scalable}, and is therefore generally close to the 2-Wasserstein distance 
\citep{nguyen2021mean, ye2024distributionally}.

This provides a computationally tractable and well-justified measure of intervention size. 
Accordingly, we measure intervention size through the Gelbrich formula applied to the observed and tilted marginal distributions of $\boldsymbol{W}$, and write the resulting quantity as $G(\boldsymbol{\delta})$. Comparing interventions on the level set $G(\boldsymbol{\delta}) = c^2$ puts them on a common scale. Under a multivariate normal model this coincides with the squared 2-Wasserstein distance, while in more general settings it remains a rigorous lower bound. The optimization problem thus becomes a search for the $\boldsymbol{\delta}$ that minimizes $\psi(\boldsymbol{\delta})$ among all possible intervention directions of a comparable size, allowing one to explore which changes to an environmental exposure mixture are most beneficial or harmful.

\subsection{Choice of Estimand}

The framework of incremental effects, when extended to a multivariate setting, moves beyond simply estimating the effect of a single, pre-specified shift or simply examining how the causal effect depends on shift size. Instead, it allows us to explore the entire space of potential interventions to identify those that are most impactful. In this section, we propose a number of different estimands within this framework and discuss their policy-relevance.

\subsubsection{Optimal Shifts}

In environmental contexts, a key objective is to determine the most effective strategy for intervention given limited resources. For instance, how should regulators modify a complex mixture of air pollutants to achieve the greatest improvement in health outcomes? Our framework directly addresses this question by defining an estimand for the optimal policy shift. We define our primary estimand of interest, $\boldsymbol{\delta}^*_c$, as the intervention direction that minimizes the causal effect $\psi(\boldsymbol{\delta})$ over the set of all ``fair shifts'' of a given size $c$:
$$
\boldsymbol{\delta}^*_c = \arg\min_{\boldsymbol{\delta} \in \mathcal{M}_c} \psi(\boldsymbol{\delta}), \quad \text{where} \quad \mathcal{M}_c = \{ \boldsymbol{\delta} : G(\boldsymbol{\delta}) = c^2 \}.
$$
The corresponding value of this optimal policy is $\psi^*_c = \psi(\boldsymbol{\delta}^*_c)$. To provide some intuition for the optimal $\boldsymbol{\delta}$, we can explore it analytically in a simplified setting. If we assume the outcome model is linear ($\mu(\boldsymbol{x},\boldsymbol{w}) = \boldsymbol{\alpha}^\top \boldsymbol{x} + \boldsymbol{\beta}^\top \boldsymbol{w}$) and the exposure distribution is normal ($f \sim \mathcal{N}(\boldsymbol{\mu}_x, \boldsymbol{\Sigma})$), the causal effect is minimized when $\boldsymbol{\delta}$ is proportional to $-\boldsymbol{\Sigma}^{-1}\boldsymbol{\beta}$. This result provides some intuition: the optimal direction is a balance between the direct effect of each pollutant on the outcome (the vector $\boldsymbol{\beta}$) and the correlation structure of the mixture (represented by $\boldsymbol{\Sigma}^{-1}$). This highlights a key insight of our multivariate approach: the covariance of the exposures is critical for determining not only which shifts are easiest to estimate, but also which shifts are most impactful. While the optimal policy shift is one primary goal, our framework is flexible and allows for the definition of other estimands that can be useful in answering other scientific questions of interest in environmental epidemiology. We now detail these in the following sections.

\subsubsection{Single Exposure Shifts}

A common goal in the analysis of environmental mixtures is to identify the components of the mixture that are most harmful. This has led to a wide range of statistical approaches aimed at performing exposure selection \citep{bobb2015bayesian, antonelli2020estimating, ferrari2020identifying, wei2020sparse, samanta2022estimation}. These approaches have inherently disregarded the potential impacts of positivity violations when examining which exposures are most harmful, but we can adapt our approach here to target similar questions without the need for strong positivity assumptions. A natural choice is to let $\boldsymbol{\delta}$ be a vector of zeros with a single non-zero component, e.g., $\boldsymbol{\delta}_j = (0, \dots, t_j, \dots, 0)$. The magnitude $t_j$ would be chosen to satisfy the same fairness constraint, $G(\boldsymbol{\delta}_j) = c^2$. While this estimand would encourage the $j$-th component of the exposures to increase more than others, in the presence of correlated exposures, it will also shift the remaining exposures, and the extent to which this occurs is less clear. An alternative option is to find a $\boldsymbol{\delta}$ such that the means of the exposures in the tilted distribution are the same, except for the $j$-th exposure, thereby isolating the impact of that single exposure. If the exposures follow a multivariate normal distribution, this is straightforward since the mean shift is given by $\boldsymbol{\Sigma} \boldsymbol{\delta}$. If we want to ensure that only the $j$-th exposure's mean is shifted, then we could set $\boldsymbol{\delta}_j = \boldsymbol{\Sigma}^{-1} \boldsymbol{e}_j$ where $\boldsymbol{e}_j = (0, \dots, t_j, \dots, 0)$ and again $t_j$ is chosen to ensure a fair shift. Note that this analytical form only holds when the exposure follows a multivariate normal distribution, which will not hold in general. We can instead use this value of $\boldsymbol{\delta}_j$ as a starting point in a numerical algorithm that searches for the exponential tilt that only shifts the $j$-th component. Note that depending on the size of the shift and the correlation structure of the exposures, such a shift may not be possible as it would inherently lead to positivity violations. In such scenarios it may be more feasible to shift groups of correlated exposures, an approach we consider in Section~\ref{sec:application}.

Once these are obtained, we can calculate $\psi(\boldsymbol{\delta}_j)$ for all $j = 1, \dots, q$ to infer which exposures have the biggest effect on the outcome. Further note that this approach can be naturally extended to explore interactions or combined effects. For example, by setting two components of $\boldsymbol{\delta}$ to be non-zero, one could investigate whether simultaneously increasing two pollutants has a synergistic effect greater than the sum of their individual impacts. We point readers to recent work examining stochastic interventions to identify interactions \citep{mccoy2026semiparametric}, as similar ideas could be applied within our framework, though we focus here on single exposure effects for now.

\subsubsection{Efficient Exposure Shifts}

While the previous two estimands are arguably the most scientifically and policy-relevant in most applications, there are other considerations at play, such as how efficiently we can estimate the chosen estimand. Environmental applications in particular are known to have relatively small effect sizes, and therefore efficiency can be particularly important when sample sizes are not exceedingly large. As we established in Section~\ref{subsec:exposure-shifts-efficiency}, the statistical difficulty of estimating $\psi(\boldsymbol{\delta})$ is heavily driven by the variance of the density ratio, $\mathrm{Var}(g_{\boldsymbol{\delta}}/f)$. For a fixed intervention size, as defined by the Wasserstein distance, we showed that the most efficient intervention direction, in terms of minimizing this variance, is proportional to the first eigenvector of the exposure covariance matrix, $\boldsymbol{\Sigma}$. We call this direction $\boldsymbol{\delta}_{\text{eff}}$, and it is clear that this direction depends only on the correlation structure of the observed exposures. While this direction is only the most efficient one under normality of the exposures, we proceed with this choice in general, as we have found that it leads to efficient estimates in more general settings, and deriving the most efficient direction in general is a difficult task. 

However, the optimal policy direction, $\boldsymbol{\delta}^*_c$, which minimizes the causal effect $\psi(\boldsymbol{\delta})$, depends on both the exposure distribution and the exposure-outcome relationship, $\mu(\boldsymbol{x},\boldsymbol{w})$. In simplified linear models, the direction of steepest ascent for the causal effect is proportional to $\boldsymbol{\Sigma}^{-1}\boldsymbol{\beta}$. In general, there is no reason for the direction of maximal statistical efficiency (related to the eigenvectors of $\boldsymbol{\Sigma}$) to be the same as the direction of the optimal causal effect (related to $\boldsymbol{\Sigma}^{-1}\boldsymbol{\beta}$). To see this, consider an intuitive example with two highly and positively correlated pollutants, where the first principal component is approximately in the $(1, 1)$ direction, which corresponds to the direction that is ``easiest'' to estimate with the observed data. Now, suppose that only the second pollutant has a strong causal effect on the outcome (i.e., $\boldsymbol{\beta} \approx (0, \beta_2)$). The optimal policy would primarily involve shifting the second pollutant. Our framework reveals an inherent tension: the policy we most want to evaluate (shifting the second pollutant alone) is statistically difficult because it moves in a direction against the data's strong correlation structure, leading to increased uncertainty in our estimate of its effect. In general, this presents a trade-off between interpretability and statistical efficiency, and users can decide based on features of their observed data which estimand to target. 

\section{Estimation and Inference}
\label{sec:estimation}

For a fixed intervention vector $\boldsymbol{\delta}$, the target estimand is the expected outcome under the tilted exposure distribution:
$$
\psi(\boldsymbol{\delta}) = \mathbb{E}\left[\int_{\mathcal{W}} \mu(\boldsymbol{X}, \boldsymbol{w}) g_{\boldsymbol{\delta}}(\boldsymbol{w} \mid \boldsymbol{X}) \, d\boldsymbol{w}\right],
$$
where the outer expectation is over the marginal distribution of covariates $\boldsymbol{X}$. Therefore, a direct approach to estimating $\psi(\boldsymbol{\delta})$ is through a plug-in procedure. This method involves replacing each component in the expression above with a corresponding empirical estimate, which leads to the plug-in estimator:
$$
\hat{\psi}_{\text{plugin}}(\boldsymbol{\delta}) = \frac{1}{n} \sum_{i=1}^{n} \left[ \int_{\mathcal{W}} \hat{\mu}(\boldsymbol{X}_i, \boldsymbol{w}) \hat{g}_{\boldsymbol{\delta}}(\boldsymbol{w} \mid \boldsymbol{X}_i) \, d\boldsymbol{w} \right].
$$

However, this estimator faces certain practical and theoretical challenges. One issue is that the estimator's performance is highly dependent on an accurate estimate of the multivariate conditional density $\hat{f}(\boldsymbol{w}|\boldsymbol{x})$, which can be difficult to estimate for multivariate exposures, and the resulting estimator will likely converge more slowly. Furthermore, this approach does not leverage the structure of the efficient influence function and, as a result, is generally not statistically efficient. These limitations motivate alternative estimators with superior statistical properties.

\subsection{One-step Estimation and Cross-fitting}

To overcome the limitations of the plug-in estimator, we employ an approach rooted in semiparametric efficiency theory. This method uses the efficient influence function (EIF) to construct an estimator that is consistent under weaker conditions and achieves the optimal asymptotic variance. The EIF for the estimand $\psi(\boldsymbol{\delta})$ is given in Section~\ref{sec:Estimands} by $\varphi(\boldsymbol{Z}; \boldsymbol{\delta}) = D_Y + D_{g,\mu} + D_{\psi}$. In our notation, substituting the EIF components yields
$$
\varphi(\boldsymbol{Z}; \psi, \mu, f)
= r_{\boldsymbol{\delta}}(\boldsymbol{W},\boldsymbol{X})\Big\{Y-\mathbb{E}_{g_{\boldsymbol{\delta}}}[\mu\mid \boldsymbol{X}]\Big\}
+ \mathbb{E}_{g_{\boldsymbol{\delta}}}[\mu\mid \boldsymbol{X}] - \psi,
$$
where $r_{\boldsymbol{\delta}}(\boldsymbol{w},\boldsymbol{x}) = g_{\boldsymbol{\delta}}(\boldsymbol{w}\mid \boldsymbol{x})/f(\boldsymbol{w}\mid \boldsymbol{x})$ and expectations with subscript $g_{\boldsymbol{\delta}}$ are taken with respect to $g_{\boldsymbol{\delta}}(\cdot\mid \boldsymbol{X})$. The one-step estimation procedure utilizes this EIF to correct an initial parameter estimate by adding the empirical average of the EIF to the initial plug-in estimate, which serves as a bias-correction term:
$$
\hat{\psi}_{\text{onestep}}(\boldsymbol{\delta})
= \hat{\psi}_{\text{plugin}}(\boldsymbol{\delta})
+ \frac{1}{n} \sum_{i=1}^{n}
\varphi\!\big(\boldsymbol{Z}_i; \hat{\psi}_{\text{plugin}}, \hat{\mu}, \hat{f}\big).
$$
This can be re-written to show that the one-step estimator takes the following form:
$$
\begin{aligned}
\hat{\psi}_{\text{onestep}}(\boldsymbol{\delta})
= \frac{1}{n}\sum_{i=1}^n
\hat r(\boldsymbol{W}_i,\boldsymbol{X}_i)\big[Y_i-\mathbb{E}_{\hat g_{\boldsymbol{\delta}}}[\hat\mu\mid \boldsymbol{X}_i]\big]
+ \frac{1}{n}\sum_{i=1}^n \mathbb{E}_{\hat g_{\boldsymbol{\delta}}}[\hat\mu\mid \boldsymbol{X}_i].
\end{aligned}
$$
This estimator has a number of key features that we will describe in subsequent sections when we study the asymptotic properties of this estimator. To summarize, it allows the use of flexible machine learning methods for estimation of each of the nuisance functions, and it is asymptotically efficient given its construction based on the EIF. The practical performance of the estimator is highly dependent on an estimate of the conditional density of the exposures, given by $f$, which can be challenging with a moderate number of exposures. This density shows up in the expectations with respect to $g_{\boldsymbol{\delta}}$, but also in the ratio term, denoted by $r_{\boldsymbol{\delta}}(\boldsymbol{W}, \boldsymbol{X})$. First, we detail how to estimate this quantity using estimates of $f$ and $\mu$, though in Section~\ref{sssec:ClassificationDensity} we propose an approach to directly estimating this quantity using regression techniques that may not have the same theoretical properties, but can have good finite-sample performance when there are a moderate number of exposures. 

For improved asymptotic properties and better finite sample performance, all one-step estimators are implemented using $K$-fold cross-fitting. Specific implementation details and its role in the asymptotic analysis are given in Appendix~\ref{sec:appendix-convergence-normality}.

\subsection{Direct estimation using regression approaches}
\label{sssec:ClassificationDensity}

When there are a moderate number of exposures, estimation of $\psi(\boldsymbol{\delta})$ becomes increasingly challenging, even for the efficient one-step estimators described above due to the inherent difficulty of estimating a multivariate, conditional distribution $f(\boldsymbol{w} \mid \boldsymbol{x})$. For this reason, we also explore an approach first described in \citep{schindl2024incremental} that does not require estimation of the exposure density at all. This can be advantageous in finite samples, particularly when estimation of the density ratio $r_{\boldsymbol{\delta}}(\boldsymbol{W}, \boldsymbol{X})$ is unstable. The first key insight is that the density ratio can be written as
$$
r_{\boldsymbol{\delta}}(\boldsymbol{w}, \boldsymbol{x})
=
\frac{\exp(\boldsymbol{\delta}^\top \boldsymbol{w})}
{\int_{\mathcal{W}} \exp(\boldsymbol{\delta}^\top \boldsymbol{v}) f(\boldsymbol{v} \mid \boldsymbol{x}) \, d\boldsymbol{v}}
=
\frac{\exp(\boldsymbol{\delta}^\top \boldsymbol{w})}
{\nu_{\boldsymbol{\delta}}(\boldsymbol{x})},
$$
which shows the density ratio can be estimated by estimating the conditional expectation $\nu_{\boldsymbol{\delta}}(\boldsymbol{X})$. This does not require the conditional density of the exposures and can be carried out using flexible, univariate regression techniques. Further, the other component of our one-step estimator can be written as
$$
m_{\boldsymbol{\delta}}(\boldsymbol{X})
:=
\int_{\mathcal{W}} \mu(\boldsymbol{X}, \boldsymbol{w}) g_{\boldsymbol{\delta}}(\boldsymbol{w} \mid \boldsymbol{X}) \, d\boldsymbol{w}
=
\frac{\mathbb{E}[ \exp(\boldsymbol{\delta}^\top \boldsymbol{W}) \mu(\boldsymbol{X}, \boldsymbol{W}) \mid \boldsymbol{X}]}{\nu_{\boldsymbol{\delta}}(\boldsymbol{X})}
=
\frac{\eta_{\boldsymbol{\delta}}(\boldsymbol{X})}{\nu_{\boldsymbol{\delta}}(\boldsymbol{X})}.
$$
This shows that this quantity can be estimated by taking the ratio of two quantities, each of which can be estimated using flexible, univariate regression techniques. This estimator was introduced previously \citep{schindl2024incremental}, though it was not implemented as they found the ratio of the estimates for these two conditional expectations to be unstable. They were working, however, in the univariate exposure setting where estimating the conditional density of the exposures is more straightforward. In our setting, with multiple exposures, conditional density estimation can be very difficult, whereas this approach relies on univariate prediction models only. Additionally, note that one can take a third strategy, which is to still estimate the conditional density of the exposures and use it whenever calculating $\int_{\mathcal{W}} \mu(\boldsymbol{X}, \boldsymbol{w}) g_{\boldsymbol{\delta}}(\boldsymbol{w} \mid \boldsymbol{X}) \, d\boldsymbol{w}$, but then use the regression approach described above for estimating the density ratio to improve stability of our estimates. While potentially useful for estimation, these estimators that obviate the need to estimate the exposure density will not inherit the same theoretical properties as the one-step estimator that directly uses estimates of $f$ and $\mu$, which we study theoretically in Section~\ref{sec:Theory}. They may, however, produce better finite-sample performance, which we study in the simulation studies in Section~\ref{sec:Simulations} across a wide range of exposure distributions. 
\subsection{Optimizing over a Manifold}

Our primary estimand, the optimal policy shift $\boldsymbol{\delta}^*_c$, is defined as the solution to
$$
\min_{\boldsymbol{\delta}\in\mathcal{M}_c} \psi(\boldsymbol{\delta}),
$$
where $\mathcal{M}_c := \{\boldsymbol{\delta}\in\mathbb{R}^q : G(\boldsymbol{\delta}) = c^2\}$ and $G$ is the Gelbrich quantity defined in Section~\ref{subsec:fixed-gelbrich-distance}. We assume $c^2$ is a regular value of $G$ (equivalently, $\nabla G(\boldsymbol{\delta})\neq \boldsymbol{0}$ for all $\boldsymbol{\delta}\in\mathcal{M}_c$), so that $\mathcal{M}_c$ is a smooth embedded hypersurface. This regularity condition is verified for the Gelbrich constraint in Appendix~\ref{sec:bfgs-induced-gelbrich}. Standard Euclidean optimization algorithms such as gradient descent are not directly applicable as a gradient step taken from a point on the manifold will likely lead to a point outside of the feasible set. Instead, we optimize $\psi(\boldsymbol{\delta})$ over $\mathcal{M}_c$ with a Riemannian BFGS algorithm. While global convergence relies on strict geometric conditions, our practical implementation leverages computationally efficient projection-type retractions that reliably secure local stationary solutions. In the empirical BFGS search, we also use a small numerical stabilization technique that discourages directions which place excessive mass on a small number of units. This does not change the definition of the optimal-policy estimand and is described in Appendix~\ref{sec:bfgs-induced-gelbrich}. Appendix~\ref{sec:bfgs-induced-gelbrich} gives the algorithmic details and convergence proofs.

\section{Efficiency and Minimax Lower Bounds}
\label{sec:Theory}

This section studies the theoretical limits of estimating
$\psi(\boldsymbol{\delta})$ and
$\theta(\boldsymbol{\delta})
=
\psi(\boldsymbol{\delta})-\psi(\boldsymbol{0})$.
The results characterize how statistical difficulty depends not only on
the sample size $n$ and the intervention vector $\boldsymbol{\delta}$,
but also on the conditional covariance structure of the exposure mixture.
The key matrix is
$$
\boldsymbol{\Sigma}_{W\mid X}
:=
\mathbb{E}\{\operatorname{Var}(\boldsymbol{W}\mid \boldsymbol{X})\},
$$
which summarizes the residual variation and correlation of the exposures
after conditioning on covariates.

\subsection{Minimax Lower Bound}
\label{sec:minimax}

We first establish a nonparametric minimax lower bound for the incremental
effect $\theta(\boldsymbol{\delta})$. The lower bound identifies how the
intrinsic difficulty of estimation changes with the magnitude and direction
of the multivariate tilt.

\begin{theorem}[Minimax lower bound]\label{thm:minimax}
Under the conditions stated in Appendix~\ref{sec:appendix-efficiency-minimax} for the minimax lower bound, let
$\mathcal{P}$ denote the corresponding model class, and let
$\theta_P(\boldsymbol{\delta})
:=
\psi_P(\boldsymbol{\delta})-\psi_P(\boldsymbol{0})$
denote the incremental effect under $P$.
If $\boldsymbol{\Sigma}_{W\mid X}$ is positive definite, then there exists
a constant $C>0$, independent of $n$ and $\boldsymbol{\delta}$, such that
$$
\inf_{\widehat\theta}\ \sup_{P\in\mathcal{P}}
\mathbb{E}_P\!\left[
\{\widehat\theta-\theta_P(\boldsymbol{\delta})\}^2
\right]
\ge
C\,
\frac{
\boldsymbol{\delta}^{\top}
\boldsymbol{\Sigma}_{W\mid X}
\boldsymbol{\delta}
}{n}.
$$
\end{theorem}

Theorem~\ref{thm:minimax} shows that the best possible root mean-squared
error obeys
$$
\operatorname{RMSE}(\widehat\theta)
\gtrsim
\sqrt{
\frac{
\boldsymbol{\delta}^{\top}
\boldsymbol{\Sigma}_{W\mid X}
\boldsymbol{\delta}
}{n}
}.
$$
Since
$$
\boldsymbol{\delta}^{\top}
\boldsymbol{\Sigma}_{W\mid X}
\boldsymbol{\delta}
=
\mathbb{E}\{
\operatorname{Var}(\boldsymbol{\delta}^{\top}\boldsymbol{W}
\mid \boldsymbol{X})
\},
$$
the minimax lower bound is determined by the residual variation in the
tilted exposure contrast. Hence larger tilts, and directions with larger
conditional variance under $\boldsymbol{\Sigma}_{W\mid X}$, are intrinsically
harder to estimate, regardless of the estimation strategy. Conversely,
directions with smaller conditional variance permit sharper estimation. The proof is given in Appendix~\ref{sec:appendix-efficiency-minimax}. This result makes explicit why one must account for both the size and the direction of $\boldsymbol{\delta}$ when assessing estimation difficulty.

\subsection{Efficiency Bound}
\label{sec:eif-variance-bound}

Here we show that the same covariance matrix also controls the nonparametric efficiency
bound. Recall that the asymptotic variance of a regular and asymptotically linear (RAL) estimator is bounded below by the variance of the efficient influence function (EIF), $\operatorname{Var}\{\varphi_{\theta(\boldsymbol{\delta})}(\boldsymbol{Z})\}$. The exact variance of the EIF depends on a complex interplay between the conditional variance of the outcome, $\operatorname{Var}(Y \mid \boldsymbol{X}, \boldsymbol{W})$, and the variation in the regression surface, $\mu(\boldsymbol{X}, \boldsymbol{W})$. However, the following theorem demonstrates that a main driver of this asymptotic variance remains the conditional covariance of the exposures, $\boldsymbol{\Sigma}_{W\mid X}$.

\begin{theorem}[Variance bound of the efficient influence function]
\label{thm:var-iff}
Under \textnormal{(C1)}--\textnormal{(C4)} defined in Appendix~\ref{sec:appendix-efficiency-minimax}, there exist
constants $0<c_{\mathrm{low}}\le c_{\mathrm{up}}<\infty$, not depending
on $\boldsymbol{\delta}$, such that
$$
c_{\mathrm{low}}\,
\boldsymbol{\delta}^{\top}
\boldsymbol{\Sigma}_{W\mid X}
\boldsymbol{\delta}
\le
\operatorname{Var}\{
\varphi_{\theta(\boldsymbol{\delta})}(\boldsymbol{Z})
\}
\le
c_{\mathrm{up}}\,
\boldsymbol{\delta}^{\top}
\boldsymbol{\Sigma}_{W\mid X}
\boldsymbol{\delta}.
$$
\end{theorem}

The proof of Theorem~\ref{thm:var-iff} is given in Appendix~\ref{sec:appendix-efficiency-minimax}.
This result establishes that the nonparametric efficiency bound grows
quadratically with the magnitude of the intervention, with its
direction-specific scaling determined by
$\boldsymbol{\Sigma}_{W\mid X}$. It also builds on the discussion in
Section~\ref{subsec:exposure-shifts-efficiency} on efficient
$\boldsymbol{\delta}$ shifts to arbitrary exposure distributions: when
intervention size is fixed by the Gelbrich/Wasserstein constraint,
directions aligned with the leading eigenvectors of
$\boldsymbol{\Sigma}_{W\mid X}$ achieve the same size of exposure shift, but
with smaller
$\boldsymbol{\delta}^{\top}\boldsymbol{\Sigma}_{W\mid X}\boldsymbol{\delta}$,
and therefore a smaller efficiency bound. Together,
Theorems~\ref{thm:minimax} and~\ref{thm:var-iff} show that
$\boldsymbol{\Sigma}_{W\mid X}$ governs both the minimax lower bound and
the semiparametric efficiency bound. Combined with the convergence result below, this implies that the minimax lower bound is tight up to constants in the first-order sense.

\subsection{Convergence and Normality}

In this section, we establish asymptotic normality of our proposed one-step estimator under suitable conditions for the nuisance functions.

\begin{theorem}[Finite-$\boldsymbol{\delta}$ CLT]\label{thm:finite-delta-CLT-main}
Assume i.i.d.\ observations, bounded $Y$ and $\boldsymbol{W}$, $f(\boldsymbol{w}\mid\boldsymbol{x})$ bounded above and away from zero on its support, and $\widehat f$ nonnegative and bounded above with probability tending to one. Fix $\Delta \in (0,\infty)$ and any tilt $\boldsymbol{\delta}$ with $\|\boldsymbol{\delta}\| \le \Delta$. If the following rate condition holds,
$$
\|\widehat{f} - f\|_2
 \big(\|\widehat{\mu} - \mu\|_2 + \|\widehat{f} - f\|_2\big) = o_P(n^{-1/2})
$$
then
$$
\sqrt{n}\,\{\widehat{\psi}(\boldsymbol{\delta}) - \psi(\boldsymbol{\delta})\}\ \rightsquigarrow\ \mathcal{N}\!\big(0,\ \mathrm{Var}\{\varphi_{\psi(\boldsymbol{\delta})}(\boldsymbol{Z})\}\big),
$$
and
$$
\sqrt{n}\,\{\widehat{\theta}(\boldsymbol{\delta}) - \theta(\boldsymbol{\delta})\}\ \rightsquigarrow\ \mathcal{N}\!\big(0,\ \mathrm{Var}\{\varphi_{\theta(\boldsymbol{\delta})}(\boldsymbol{Z})\}\big),
\qquad
\varphi_{\theta(\boldsymbol{\delta})} := \varphi_{\psi(\boldsymbol{\delta})} - \varphi_{\psi(\boldsymbol{0})}.
$$
\end{theorem}
A proof of this result can be found in Appendix~\ref{sec:appendix-convergence-normality}. Although asymptotic linearity appears to depend on the tilted nuisances $m_{\boldsymbol{\delta}}$ and $r_{\boldsymbol{\delta}}$ because of the EIF structure, this theorem reduces this requirement to $n^{-1/4}$ convergence rates for the outcome regression and exposure density, respectively. Additionally, this result allows us to construct Wald intervals for individual shifts and to use multivariate normal approximations for joint inference across multiple estimands simultaneously. Variance estimators and their theoretical justifications are given in Appendix~\ref{sec:appendix-convergence-normality}.

\section{Sensitivity analysis to unmeasured confounding}
\label{sec:SensitivityFinal}

In this section we develop a sensitivity analysis approach to assess the robustness of our causal estimates to the presence of unmeasured confounders. Throughout, we assume that there exist unmeasured confounders $\boldsymbol{U}$ such that
$$
Y(\boldsymbol{w}) \perp\!\!\!\!\perp \boldsymbol{W} \mid \boldsymbol{X}, \boldsymbol{U}, \quad \text{for all }\boldsymbol{w}.
$$
For notational clarity, let $\boldsymbol{V}:=(\boldsymbol{X},\boldsymbol{U})$ denote the full adjustment set, reserving $\boldsymbol{Z}:=(\boldsymbol{X},\boldsymbol{W},Y)$ for the observed data throughout. We still consider incremental policies defined by exponential tilting of the observed conditional exposure density $f(\boldsymbol{w}\mid \boldsymbol{X})$. The full-data causal estimand is given by
$$
\psi(\boldsymbol{\delta})
=\mathbb{E}\left[\int_{\mathcal{W}} \mu(\boldsymbol{V},\boldsymbol{w}) g_{\boldsymbol{\delta}}(\boldsymbol{w}\mid \boldsymbol{X}) \, d\boldsymbol{w}\right],
$$
where $\mu(\boldsymbol{V},\boldsymbol{w})=\mathbb{E}[Y\mid \boldsymbol{V},\boldsymbol{W}=\boldsymbol{w}]$ is referred to as the long outcome regression. We further denote the estimand obtained by applying the same identification formula while omitting $\boldsymbol{U}$ as
$$
\psi_s(\boldsymbol{\delta})
=\mathbb{E}\left[\int_{\mathcal{W}} \mu_s(\boldsymbol{X},\boldsymbol{w}) g_{\boldsymbol{\delta}}(\boldsymbol{w}\mid \boldsymbol{X}) \, d\boldsymbol{w}\right],
$$
where $\mu_s(\boldsymbol{X},\boldsymbol{w})=\mathbb{E}[Y\mid \boldsymbol{X},\boldsymbol{W}=\boldsymbol{w}]$ is referred to as the short outcome regression. 
When there is unmeasured confounding, the identified short estimand $\psi_s(\boldsymbol{\delta})$, which conditions only on observed covariates $\boldsymbol{X}$, diverges from the causal target $\psi(\boldsymbol{\delta})$ that is defined under conditional exchangeability given $\boldsymbol{V}=(\boldsymbol{X},\boldsymbol{U})$. We follow the approach of \citep{chernozhukov2021long} and leverage the geometry of Riesz representers to derive sharp bias bounds based on $L_2$ norms, yielding interpretable calibration of confounding strength without imposing restrictive structural assumptions on $\mu(\cdot,\cdot)$.

\subsection{Bounding the bias}

Let $f(\boldsymbol{w}\mid \boldsymbol{V})$ denote the true conditional exposure density given the full adjustment set. The parameter $\psi(\boldsymbol{\delta})$ can be written as a linear functional of the long regression:
$$
\psi(\boldsymbol{\delta})
=\mathbb{E}\left[\mu(\boldsymbol{V},\boldsymbol{W})\alpha_{\boldsymbol{\delta}}(\boldsymbol{V},\boldsymbol{W})\right],
\qquad
\alpha_{\boldsymbol{\delta}}(\boldsymbol{V},\boldsymbol{W})
=\frac{g_{\boldsymbol{\delta}}(\boldsymbol{W}\mid \boldsymbol{X})}{f(\boldsymbol{W}\mid \boldsymbol{V})}.
$$
Analogously, the short estimand admits the representation
$$
\psi_s(\boldsymbol{\delta})
=\mathbb{E}\left[\mu_s(\boldsymbol{X},\boldsymbol{W})\alpha_{s,\boldsymbol{\delta}}(\boldsymbol{X},\boldsymbol{W})\right],
\qquad
\alpha_{s,\boldsymbol{\delta}}(\boldsymbol{X},\boldsymbol{W})
=\frac{g_{\boldsymbol{\delta}}(\boldsymbol{W}\mid \boldsymbol{X})}{f(\boldsymbol{W}\mid \boldsymbol{X})}.
$$
Under mild regularity conditions ensuring square-integrability of the relevant Riesz representers, previous work shows that the bias admits the exact representation \citep{chernozhukov2021long}:
$$
\psi(\boldsymbol{\delta})-\psi_s(\boldsymbol{\delta})
=\mathbb{E}\left[\Delta_{\mu}(\boldsymbol{V},\boldsymbol{W})\Delta_{\alpha}(\boldsymbol{V},\boldsymbol{W})\right],
$$
where
$$
\Delta_{\mu}(\boldsymbol{V},\boldsymbol{W})
=\mu(\boldsymbol{V},\boldsymbol{W})-\mu_s(\boldsymbol{X},\boldsymbol{W}),
\qquad
\Delta_{\alpha}(\boldsymbol{V},\boldsymbol{W})
=\alpha_{\boldsymbol{\delta}}(\boldsymbol{V},\boldsymbol{W})-\alpha_{s,\boldsymbol{\delta}}(\boldsymbol{X},\boldsymbol{W}).
$$
This identity isolates two necessary sources of bias: the additional outcome variation explained by $\boldsymbol{U}$ beyond $(\boldsymbol{X},\boldsymbol{W})$ through $\Delta_{\mu}$, and the additional information about the exposure distribution provided by $\boldsymbol{U}$ beyond $\boldsymbol{X}$ through $\Delta_{\alpha}$. By Cauchy--Schwarz,
$$
|\psi_s(\boldsymbol{\delta})-\psi(\boldsymbol{\delta})|
\le \sqrt{\mathbb{E}[\Delta_{\mu}(\boldsymbol{V},\boldsymbol{W})^2]}\sqrt{\mathbb{E}[\Delta_{\alpha}(\boldsymbol{V},\boldsymbol{W})^2]}.
$$

\subsection{Sensitivity Parameters}

To express the bound in terms of identifiable scale components and partial $R^2$-type sensitivity parameters, define the following identifiable parameters
$$
\sigma_s^2
=\mathbb{E}\left[(Y-\mu_s(\boldsymbol{X},\boldsymbol{W}))^2\right],
\qquad
A_s^2(\boldsymbol{\delta})
=\mathbb{E}\left[\alpha_{s,\boldsymbol{\delta}}(\boldsymbol{X},\boldsymbol{W})^2\right],
\qquad
S(\boldsymbol{\delta})=\sigma_s A_s(\boldsymbol{\delta}).
$$
We parameterize the outcome component by the nonparametric partial $R^2$ of $\boldsymbol{U}$ with $Y$ given $(\boldsymbol{X},\boldsymbol{W})$,
$$
R^2_{Y\sim \boldsymbol{U}\mid \boldsymbol{X},\boldsymbol{W}}
=\frac{\mathbb{E}\left[\Delta_{\mu}(\boldsymbol{V},\boldsymbol{W})^2\right]}{\sigma_s^2},
\qquad
C_Y=\sqrt{R^2_{Y\sim \boldsymbol{U}\mid \boldsymbol{X},\boldsymbol{W}}}.
$$
For the treatment component, we use the relative increase in $L_2$ variation of the Riesz representer induced by conditioning on $\boldsymbol{U}$:
$$
C_D^2(\boldsymbol{\delta})
=\frac{\mathbb{E}\left[\Delta_{\alpha}(\boldsymbol{V},\boldsymbol{W})^2\right]}{A_s^2(\boldsymbol{\delta})}.
$$
Combining these definitions yields the sensitivity bound
$$
|\psi_s(\boldsymbol{\delta})-\psi(\boldsymbol{\delta})|
\le S(\boldsymbol{\delta})\cdot C_Y\cdot C_D(\boldsymbol{\delta}).
$$

For interpretability in the empirical analysis, we reparameterize the strength of unmeasured confounding as
$$
\eta_Y^2:=C_Y^2,
\qquad
\eta_\alpha^2(\boldsymbol{\delta}):=1-R_\alpha^2(\boldsymbol{\delta}),
\qquad
C_D^2(\boldsymbol{\delta})
=\frac{\eta_\alpha^2(\boldsymbol{\delta})}{1-\eta_\alpha^2(\boldsymbol{\delta})}.
$$

Here $R_\alpha^2(\boldsymbol{\delta})$ is the squared correlation between the long and short Riesz representer contrasts. The parameters $\eta_Y^2$ and $\eta_\alpha^2(\boldsymbol{\delta})$ quantify the proportions of residual outcome variation and Riesz representer variation, respectively, that are attributable to the unmeasured confounder $\boldsymbol{U}$. Together with the scale factor $S(\boldsymbol{\delta})$, these parameters define the maximal confounding bias, which allows us to assess sensitivity of results given a certain strength of unmeasured confounding. Note that the same bias formula applies to the incremental effect $\theta(\boldsymbol{\delta})=\psi(\boldsymbol{\delta})-\psi(\boldsymbol{0})$ after replacing the short Riesz representer with the contrast $\alpha_{s,\theta,\boldsymbol{\delta}}(\boldsymbol{X},\boldsymbol{W})=\alpha_{s,\boldsymbol{\delta}}(\boldsymbol{X},\boldsymbol{W})-\alpha_{s,\boldsymbol{0}}(\boldsymbol{X},\boldsymbol{W})$, and the remaining sensitivity analysis proceeds in the same manner. Note that one can use formal benchmarking with observed covariates to reason about plausible values of these sensitivity parameters as in \citep{cinelli2020making}. We implement this in Section~\ref{sec:application}, though full details are left to Appendix~\ref{sec:sensitivity}.

\section{Simulation studies}
\label{sec:Simulations}

The simulation study evaluates finite-sample behavior along two design axes: the shape of the conditional exposure distribution and the complexity of the outcome regression. We consider three exposure distributions: Gaussian, skewed, and a truncated mixture of Gaussians. We also consider two outcome regressions: a linear model and a nonlinear interaction model motivated by environmental mixture challenges including synergistic toxicity and effect modification. Throughout, we use $n=5,000$ observations, $p=10$ baseline covariates, $q=6$ exposures, and 500 Monte Carlo repetitions. Specific details of the data-generating processes, and additional information on the different estimation strategies can be found in Appendix~\ref{sec:appendix-simulation-details}. 

We compare seven procedures for estimating the tilted mean $\psi(\boldsymbol{\delta})$, though in all cases we estimate the outcome regression using gradient boosting. We do, however, vary the manner in which the exposure density is estimated. We model the exposure distribution using a semiparametric location-shift decomposition, $W_j = \mu_j(\boldsymbol{X}) + \sigma_j \varepsilon_j$, where the conditional mean $\mu_j(\boldsymbol{X})$ is estimated via gradient boosting, and the joint distribution of the exposures is obtained using a Gaussian copula. We consider three approaches for estimating the distribution of $\varepsilon_j$: 1) assuming it is Gaussian, 2) assuming it follows a $t$-distribution, and 3) using a smoothed version of the empirical distribution of the residuals. Within each of the three approaches to conditional density estimation, we explore 1) estimating the density ratio using the estimate of $f$, and 2) directly estimating the density ratio by estimating $\nu_{\boldsymbol{\delta}}(\boldsymbol{x})$. This leads to six estimators, though we consider a seventh estimator that involves directly estimating both $\nu_{\boldsymbol{\delta}}(\boldsymbol{x})$ and $\eta_{\boldsymbol{\delta}}(\boldsymbol{x})$ without ever needing to estimate the exposure density. Whenever we are estimating either $\nu_{\boldsymbol{\delta}}(\boldsymbol{x})$ or $\eta_{\boldsymbol{\delta}}(\boldsymbol{x})$ directly, we utilize SoftBART \citep{linero2018bayesian} for estimation.

\subsection{Simulation Results}
\label{subsec:sim_results}

Figure~\ref{fig:simulation_results} displays the estimated tilted mean, $\psi(\boldsymbol{\delta})$, across 500 Monte Carlo repetitions for the seven procedures under the six data-generating designs. The figure illustrates how nuisance estimation choices affect finite-sample performance, while corresponding numerical summaries of bias and root mean-squared error, aggregated across all six designs, are detailed in Appendix~\ref{sec:appendix-simulation-results-table}.

\begin{figure}[htbp]
    \centering
    \includegraphics[width=0.96\textwidth]{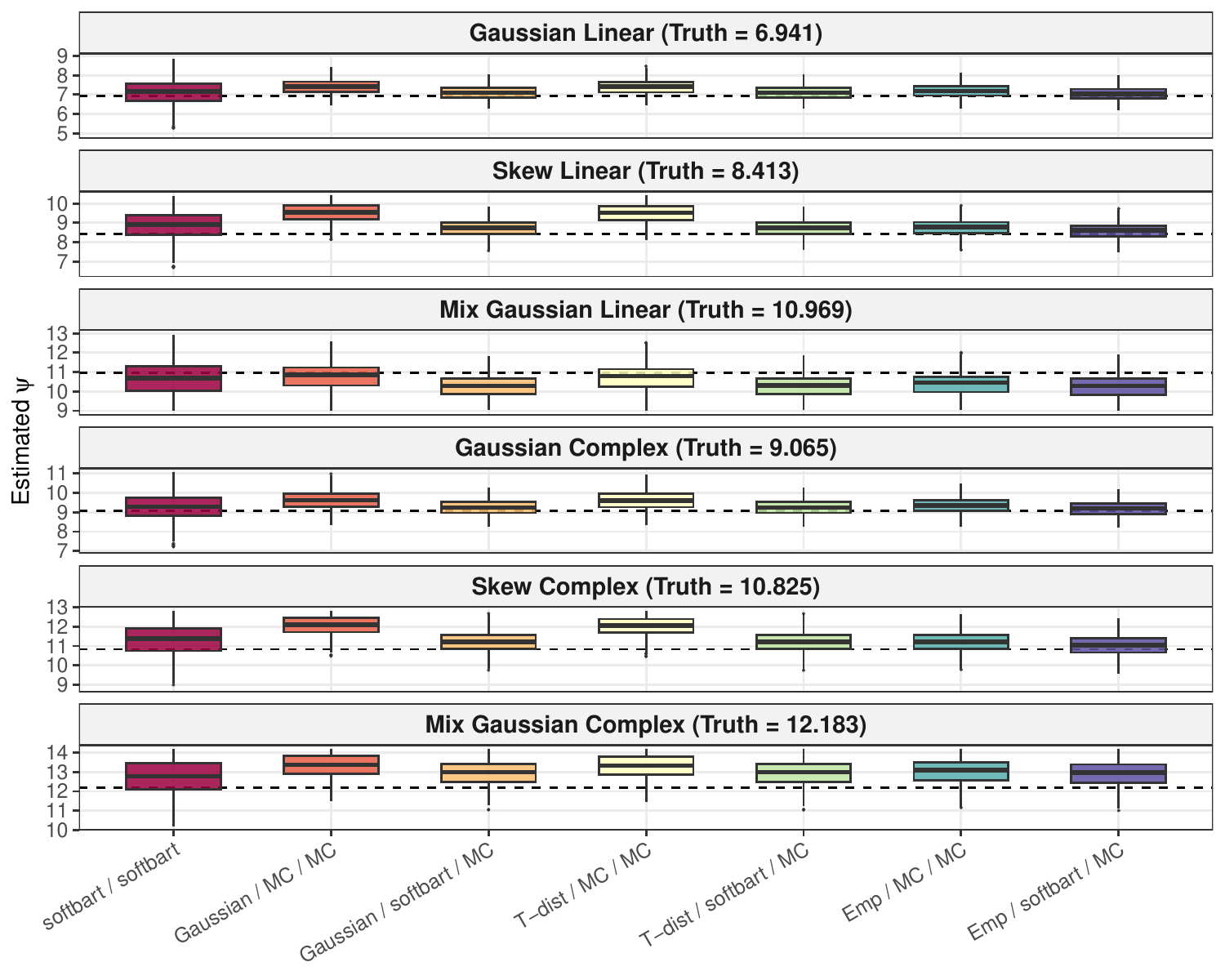}
    \caption{Boxplots of $\psi(\boldsymbol{\delta})$ across 500 repetitions under six data-generating designs. The dashed horizontal line marks the true value of the estimand. The labels on the horizontal axis identify the residual family, the method used for $r_{\boldsymbol{\delta}}(\boldsymbol{w}, \boldsymbol{x})$, and the method used for estimating $m_{\boldsymbol{\delta}}(\boldsymbol{x})$. Throughout, MC represents simply plugging in estimates of $f$ or $\mu$ directly, while SoftBART represents direct estimation as described in Section~\ref{sssec:ClassificationDensity}.}
    \label{fig:simulation_results}
\end{figure}

The main pattern is similar across the six designs. Within the semiparametric exposure models, directly estimating the density ratio $r_{\boldsymbol{\delta}}(\boldsymbol{w}, \boldsymbol{x})$ using regression reduces bias and variability relative to forming the density ratio solely by plugging in the estimated exposure density, regardless of the distributional family used. Among all the approaches considered, using the empirical distribution of the residuals combined with SoftBART estimation of $r_{\boldsymbol{\delta}}(\boldsymbol{w}, \boldsymbol{x})$ consistently performs well across all data-generating processes. The direct regression procedure that completely avoids density estimation is more variable, reflecting the difficulty of estimating both tilted nuisance functions without the structure imposed by an exposure model. These results also indicate that finite-sample performance of the one-step estimator is driven primarily by estimation of the density ratio $r_{\boldsymbol{\delta}}(\boldsymbol{w}, \boldsymbol{x})$. Once this component is estimated, the choice among the three residual families has a smaller impact on the resulting performance.

\section{Assessing Health Impacts of \texorpdfstring{PM$_{2.5}$}{PM2.5} Component Mixtures}
\label{sec:application}
We now evaluate an important public health question regarding the health 
impacts of long-term exposure to fine particulate matter (PM$_{2.5}$) and its complex 
chemical mixture. We constructed a county-level dataset across the 
United States for the year 2019 to estimate the exposure-response relationship 
between PM$_{2.5}$ constituents and age-adjusted hospitalization rates per 10,000 people for Chronic Obstructive Pulmonary Disease (COPD). We obtained health records from the CDC Environmental Public Health Tracking 
Network (EPHTN) and CDC WONDER. High-resolution estimates of 
PM$_{2.5}$ mass and its chemical constituents were obtained from the Atmospheric Composition Analysis Group 
\citep{van2019regional}. To adjust for potential confounding, we integrated a 
broad set of sociodemographic, behavioral, and clinical covariates compiled 
from the U.S. Census Bureau and CDC surveillance systems. Our analysis focuses on a $q=5$-dimensional PM$_{2.5}$ component mixture: black carbon (BC), nitrates (NO$_3$), organic matter (OM), sulfates (SO$_4$), and ammonium (NH$_4$). All reported curves use the cross-fitted one-step estimator from Section~\ref{sec:estimation} together with the empirical residual model shown to perform best in the simulation studies. Motivated by the simulation finding that finite-sample performance is driven primarily by estimation of the density ratio, we directly estimate $r_{\boldsymbol{\delta}}(\boldsymbol{w},\boldsymbol{x})$ using SoftBART and estimate $m_{\boldsymbol{\delta}}(\boldsymbol{x})$ by plugging in the fitted outcome regression for all non-BFGS estimands. For the Riemannian BFGS search, we instead use the estimates of $f$ and $\mu$ directly, which is required for providing a smooth objective for derivative-based optimization.

\begin{figure}[h]
\centering
\includegraphics[width=\textwidth]{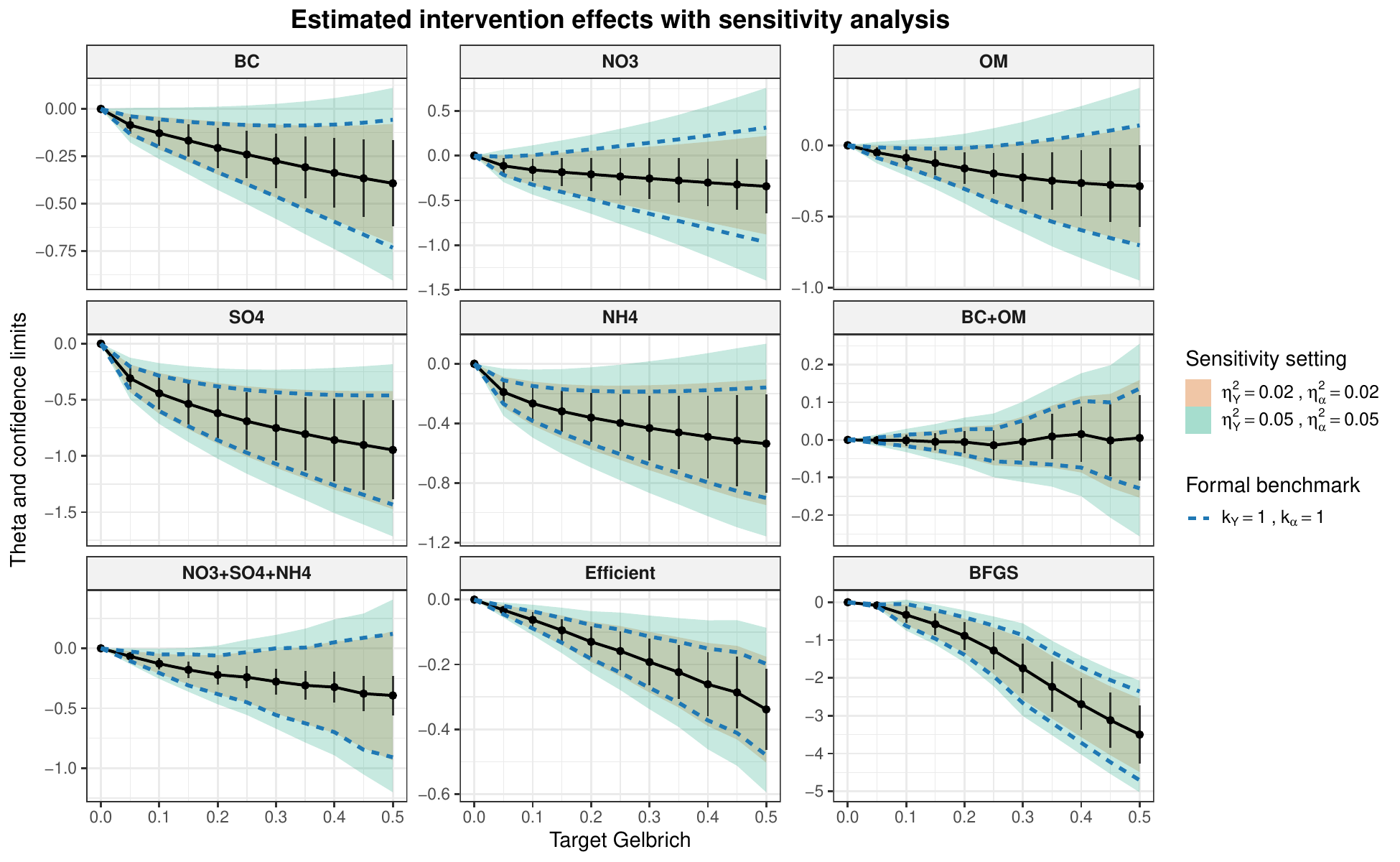}
\caption{Estimated incremental effects for all estimands considered. The black points and vertical bars show $\widehat{\theta}(\boldsymbol{\delta})$ and its $95\%$ confidence interval. Colored ribbons show pointwise Wald intervals after fixing two sensitivity-parameter settings. The dashed blue curves show sensitivity analysis using sensitivity parameters found from formal benchmarking.}
\label{fig:application_main}
\end{figure}

We study four types of intervention paths corresponding to single exposure shifts, shifting groups of exposures, the efficient shift, and finding the optimal shift in terms of reducing overall hospitalization rates. For the single exposure shifts, we utilize shifts of the form $\boldsymbol{\delta}_j = (0, \dots, t_j, \dots, 0)$ when studying exposure $j$. The magnitude $t_j$ is chosen to satisfy the corresponding Gelbrich constraint. These paths are useful for comparing feasible directions on the intervention manifold, though they do not necessarily isolate single-pollutant causal effects because the tilted distribution remains multivariate and continues to alter the entire mixture distribution, not just exposure $j$. We take a different approach for studying groups of exposures, where we consider two groups: 1) BC+OM, and 2) NO$_3$+SO$_4$+NH$_4$. We use a numerical algorithm to find the $\boldsymbol{\delta}$ value that shifts the means of each of the exposures within a group, while holding the means of the other exposures constant. We do not apply this idea to the single-pollutant shifts, because finding such a shift is not feasible in many cases due to the high correlations among the exposures. The efficient shift in this setting corresponds to $\boldsymbol{\delta}$ values that shift all exposures simultaneously, reflecting the strong dependence among all exposures in our data. Lastly, for the optimal policy we run the Riemannian BFGS algorithm, though we do so from 100 different starting values since the optimization problem is non-convex and global optimality is not guaranteed.

Figure~\ref{fig:application_main} shows the main empirical results. All five single-pollutant paths are negative, with the steepest declines for SO$_4$ and NH$_4$, followed by NO$_3$ and OM, and a smaller negative curve for BC. Because these interventions shift the entire mixture rather than isolating a single pollutant while holding the other components fixed, we interpret them as rankings of feasible intervention directions rather than single-pollutant dose-response effects. The grouped paths provide a complementary comparison: the BC+OM path remains close to zero, whereas the NO$_3$+SO$_4$+NH$_4$ path shows a significant negative effect, suggesting that the estimated harmful effects are concentrated primarily in the nitrate, sulfate, and ammonium components. This provides useful information for regulatory agencies trying to determine how best to reduce particulate air pollution, because organic matter and black carbon are primary, combustion-based pollutants, while the remaining three exposures are secondary pollutants formed in the atmosphere. The BFGS path yields the largest estimated reduction in hospitalizations at each Gelbrich level, corresponding to the estimated optimal policy shift. To provide more interpretation to the BFGS results, Figure~\ref{fig:bfgs_delta_paths} summarizes the $\boldsymbol{\delta}$ values found by the BFGS algorithm for each Gelbrich constraint across all 100 starting values. Although the nonconvex search produces some heterogeneity across starting values, the overall story remains consistent across starting values. The most negative values are concentrated on SO$_4$ and NO$_3$, showing that these are the most important components to reduce for lowering hospitalizations.

\begin{figure}[htbp]
\centering
\includegraphics[width=0.9\textwidth]{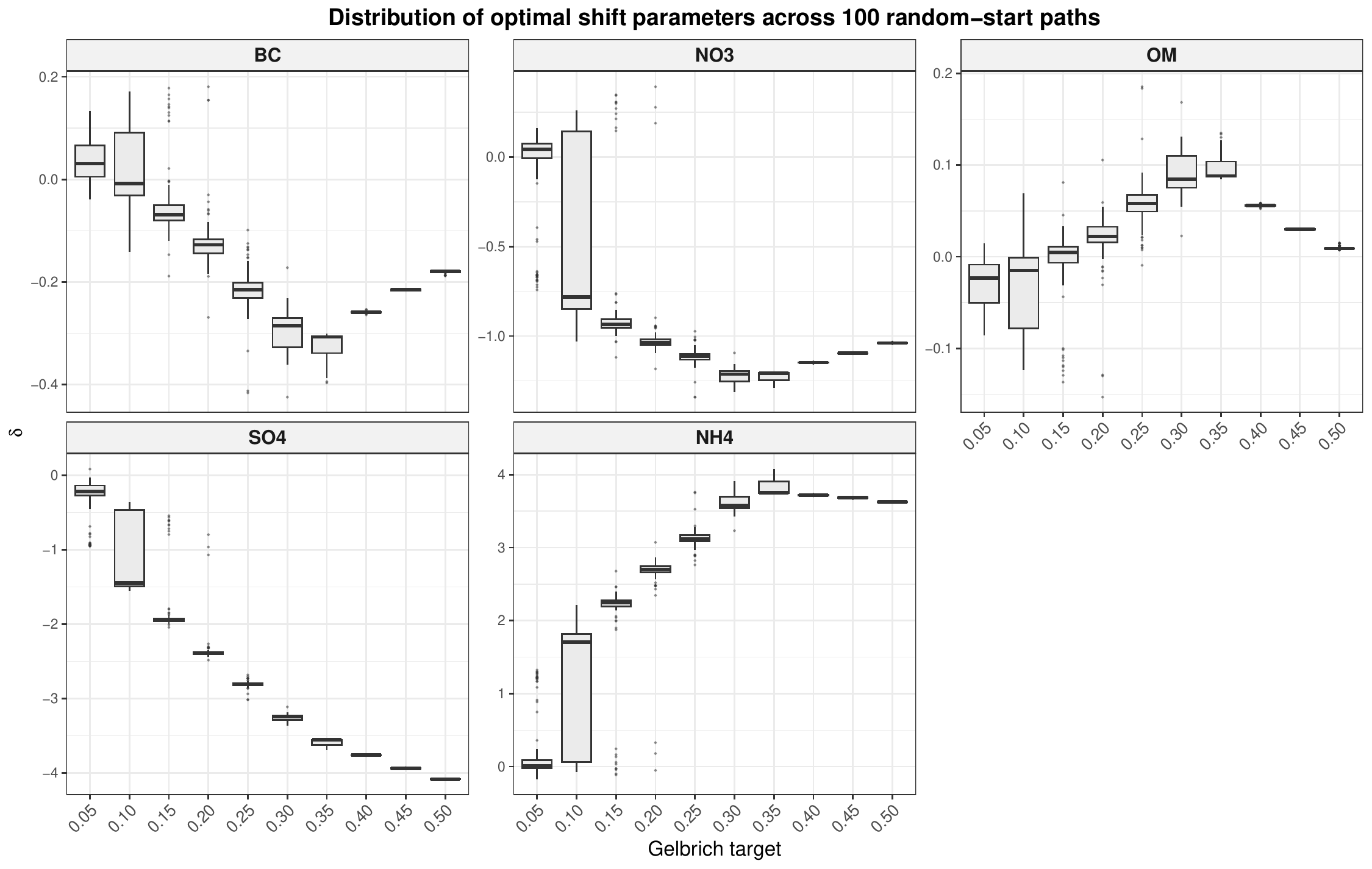}
\caption{Distribution of the BFGS $\boldsymbol{\delta}$ values for each exposure across 100 random starting values for each Gelbrich target.}
\label{fig:bfgs_delta_paths}
\end{figure}

Figure~\ref{fig:application_main} also summarizes sensitivity to unmeasured confounding. The shaded ribbons give two fixed sensitivity settings, and the dashed blue curves show results based on formal benchmarking in Appendix~\ref{sec:benchmark_formal}. The formal benchmarking seen here uses values of $k_Y$ and $k_D$ both set to 1, which means that we assume the unmeasured confounder is as strongly related to the exposures and outcome as the strongest of the observed covariates. Under these benchmark values, the BC, SO$_4$, NH$_4$, efficient, and BFGS curves remain below zero across the displayed Gelbrich range, whereas OM and the NO$_3$+SO$_4$+NH$_4$ path cross zero at larger Gelbrich targets and NO$_3$ remains below zero only at the smallest target. Notably, sensitivity analyses indicate that the estimated effects for OM, NO$_3$, and NO$_3$+SO$_4$+NH$_4$ are the most vulnerable to unmeasured confounding, whereas the SO$_4$, efficient, and BFGS intervention paths exhibit the greatest robustness. Supplementary diagnostics showing how large the sensitivity parameters must become to remove any effects seen here (Appendix~\ref{sec:appendix-application-results}) reinforce this qualitative ordering, confirming that the SO$_4$, efficient, and BFGS interventions require substantially larger sensitivity parameters to nullify the estimated health benefits.

\section{Discussion}

In this manuscript we developed methodology for estimating the health effects of multiple air pollutants simultaneously in a way that is robust to the presence of severe positivity violations. By examining stochastic interventions with tilted exposure distributions, we can study which exposures are most harmful without relying on model-based extrapolation. One critical issue in the multivariate setting is how to define a fair shift that corresponds to similar shifts in the exposure distribution, which we do via the 2-Wasserstein distance. We provide asymptotic theory and minimax estimation rates for our proposed estimands, and show in a national study of the health effects of air pollution that there are detrimental effects of the air pollution mixture, but that these are largely driven by nitrates and sulfates.

There are a number of directions for future work that could expand upon, and improve, the methodology seen here. For one, our estimators are applicable to any stochastic shift estimand, and future research could target different shifts other than the exponentially tilted ones seen here, which maintain public health relevance and could be potentially more interpretable for practitioners. Additionally, one could expand on the sensitivity analyses developed here by incorporating recent results on sensitivity analysis for multiple exposures \citep{zheng2021bayesian}. These incorporate moderate parametric assumptions in the multiple exposure setting and allow one to produce partial identification regions that could be tighter than those seen here, and allow one to incorporate additional assumptions or sources of information, such as negative control variables. Overall, we believe the proposed framework provides analysts, particularly those involved in the analysis of air pollution mixtures, robust approaches to estimating causal effects of multivariate, continuous exposures.

\bibliographystyle{apalike}
\bibliography{References} 

\appendix

\section{Neyman-orthogonality and robustness to misspecified outcome model}

\paragraph{Setup.}
Let $\mu(\boldsymbol{x},\boldsymbol{w}) = \mathbb{E}[Y \mid \boldsymbol{X}=\boldsymbol{x}, \boldsymbol{W}=\boldsymbol{w}]$ denote the outcome regression, let $f(\boldsymbol{w} \mid \boldsymbol{x})$ be the observed exposure density, and let $g(\boldsymbol{w} \mid \boldsymbol{x})$ be the tilted density with density ratio $r_{\boldsymbol{\delta}}(\boldsymbol{w},\boldsymbol{x}) = g(\boldsymbol{w} \mid \boldsymbol{x}) / f(\boldsymbol{w} \mid \boldsymbol{x})$. We write $\mathbb{E}_f[\cdot \mid \boldsymbol{X}]$ and $\mathbb{E}_g[\cdot \mid \boldsymbol{X}]$ for conditional expectations with respect to $f(\cdot \mid \boldsymbol{X})$ and $g(\cdot \mid \boldsymbol{X})$, respectively. Assume standard regularity conditions. Consider the efficient influence function for a scalar parameter $\psi$:
$$
\varphi(\boldsymbol{Z}; \psi, \mu, r)
:= r_{\boldsymbol{\delta}}(\boldsymbol{W},\boldsymbol{X})\Big\{Y - \mathbb{E}_g[\mu(\boldsymbol{X},\boldsymbol{W}) \mid \boldsymbol{X}]\Big\}
+ \mathbb{E}_g[\mu(\boldsymbol{X},\boldsymbol{W}) \mid \boldsymbol{X}] - \psi.
$$
We utilize the identities:
$$
\mathbb{E}[r_{\boldsymbol{\delta}}(\boldsymbol{W},\boldsymbol{X}) \mid \boldsymbol{X}] = 1,\qquad
\mathbb{E}[r_{\boldsymbol{\delta}}(\boldsymbol{W},\boldsymbol{X})\mu(\boldsymbol{X},\boldsymbol{W}) \mid \boldsymbol{X}] = \mathbb{E}_g[\mu(\boldsymbol{X},\boldsymbol{W}) \mid \boldsymbol{X}].
$$
\paragraph{(a) Orthogonality with respect to $\mu$.}
Fix $r$ and perturb $\mu$ along a path $\mu_\varepsilon = \mu + \varepsilon h$. Since $\mathbb{E}[rY]$ and $\psi$ do not depend on $\varepsilon$, we have:
\begin{align*}
\frac{d}{d\varepsilon}\,\mathbb{E}\{\varphi(\boldsymbol{Z}; \psi, \mu_\varepsilon, r)\}\Big|_{\varepsilon=0}
&= \mathbb{E}\!\left[-r_{\boldsymbol{\delta}}(\boldsymbol{W},\boldsymbol{X})\,\mathbb{E}_g[h(\boldsymbol{X},\boldsymbol{W}) \mid \boldsymbol{X}]
+ \mathbb{E}_g[h(\boldsymbol{X},\boldsymbol{W}) \mid \boldsymbol{X}]\right] \\
&= \mathbb{E}\!\left[\big\{1 - \mathbb{E}[r_{\boldsymbol{\delta}}(\boldsymbol{W},\boldsymbol{X}) \mid \boldsymbol{X}]\big\}\,
\mathbb{E}_g[h(\boldsymbol{X},\boldsymbol{W}) \mid \boldsymbol{X}]\right] \\
&= 0.
\end{align*}
Hence, the influence function is Neyman-orthogonal with respect to $\mu$.

\paragraph{(b) Sensitivity in $r$.}
Fix $\mu$ and perturb $r$ along a normalized path $r_\varepsilon = r(1 + \varepsilon v)$, where $v(\boldsymbol{w},\boldsymbol{x})$ is a measurable function satisfying the constraint $\mathbb{E}[r_\varepsilon(\boldsymbol{W},\boldsymbol{X}) \mid \boldsymbol{X}] = 1$ for all $\varepsilon$. Differentiating this constraint at $\varepsilon=0$ yields $\mathbb{E}_f[r_{\boldsymbol{\delta}}(\boldsymbol{W},\boldsymbol{X})v(\boldsymbol{W},\boldsymbol{X}) \mid \boldsymbol{X}] = 0$, which is equivalent to $\mathbb{E}_g[v(\boldsymbol{W},\boldsymbol{X}) \mid \boldsymbol{X}] = 0$.

Let $m_\varepsilon(\boldsymbol{X}) := \mathbb{E}_{g_\varepsilon}[\mu(\boldsymbol{X},\boldsymbol{W}) \mid \boldsymbol{X}]$, where $g_\varepsilon$ is the tilted density corresponding to $r_\varepsilon$. Since $g_\varepsilon(\boldsymbol{w} \mid \boldsymbol{x}) = r_\varepsilon(\boldsymbol{w},\boldsymbol{x})f(\boldsymbol{w} \mid \boldsymbol{x})$, we have:
$$
\frac{d}{d\varepsilon}m_\varepsilon(\boldsymbol{X})\Big|_{\varepsilon=0}
= \mathbb{E}_g[v(\boldsymbol{W},\boldsymbol{X})\mu(\boldsymbol{X},\boldsymbol{W}) \mid \boldsymbol{X}].
$$
Let $m'(\boldsymbol{X}) := \mathbb{E}_g[v\mu \mid \boldsymbol{X}]$. The influence function along the path is:
$$
\varphi(\boldsymbol{Z}; \psi, \mu, r_\varepsilon)
= r_\varepsilon(\boldsymbol{W},\boldsymbol{X})\{Y - m_\varepsilon(\boldsymbol{X})\} + m_\varepsilon(\boldsymbol{X}) - \psi.
$$
Differentiating the expected influence function yields:
\begin{align*}
\frac{d}{d\varepsilon}\,\mathbb{E}\{\varphi(\boldsymbol{Z}; \psi, \mu, r_\varepsilon)\}\Big|_{\varepsilon=0}
&= \mathbb{E}\big[r_{\boldsymbol{\delta}}(\boldsymbol{W},\boldsymbol{X})v(\boldsymbol{W},\boldsymbol{X})\{Y - m_0(\boldsymbol{X})\}\big]
+ \mathbb{E}\big[(1 - r_{\boldsymbol{\delta}}(\boldsymbol{W},\boldsymbol{X}))\,m'(\boldsymbol{X})\big] \\
&= \mathbb{E}\big[r_{\boldsymbol{\delta}}(\boldsymbol{W},\boldsymbol{X})v(\boldsymbol{W},\boldsymbol{X})\mu(\boldsymbol{X},\boldsymbol{W})\big] \\
&\quad - \mathbb{E}\big[m_0(\boldsymbol{X})\mathbb{E}[r_{\boldsymbol{\delta}}(\boldsymbol{W},\boldsymbol{X})v(\boldsymbol{W},\boldsymbol{X}) \mid \boldsymbol{X}]\big] + \mathbb{E}\big[(1 - r_{\boldsymbol{\delta}}(\boldsymbol{W},\boldsymbol{X}))\,m'(\boldsymbol{X})\big] \\
&= \mathbb{E}\big[r_{\boldsymbol{\delta}}(\boldsymbol{W},\boldsymbol{X})v(\boldsymbol{W},\boldsymbol{X})\mu(\boldsymbol{X},\boldsymbol{W})\big] - \underbrace{\mathbb{E}[m_0(\boldsymbol{X}) \cdot 0]}_{=\ 0} + \underbrace{\mathbb{E}\big[(1 - r_{\boldsymbol{\delta}}(\boldsymbol{W},\boldsymbol{X}))\,m'(\boldsymbol{X})\big]}_{=\ 0} \\
&= \mathbb{E}\{\mathbb{E}_g[v(\boldsymbol{W},\boldsymbol{X})\mu(\boldsymbol{X},\boldsymbol{W}) \mid \boldsymbol{X}]\}.
\end{align*}

In general, $\mathbb{E}\{\mathbb{E}_g[v\mu \mid \boldsymbol{X}]\} \neq 0$. Therefore, the derivative does not vanish, implying that the influence function is not Neyman-orthogonal in the $r$ direction unless additional restrictive constraints are imposed on the perturbation $v$. This non-orthogonality with respect to $r$ extends directly to the exposure density $f$. For a fixed target density $g$, any perturbation in $f$ induces a corresponding change in the density ratio $r = g/f$. Consider a perturbation path $f_\varepsilon = f(1 + \varepsilon s_f)$, where $s_f(\boldsymbol{w} \mid \boldsymbol{x})$ is a standard score function satisfying $\mathbb{E}_f[s_f(\boldsymbol{W} \mid \boldsymbol{X}) \mid \boldsymbol{X}] = 0$. Specifically:
$$
r_\varepsilon(\boldsymbol{w},\boldsymbol{x}) 
= \frac{g(\boldsymbol{w} \mid \boldsymbol{x})}{f_\varepsilon(\boldsymbol{w} \mid \boldsymbol{x})} 
= \frac{g(\boldsymbol{w} \mid \boldsymbol{x})}{f(\boldsymbol{w} \mid \boldsymbol{x})(1 + \varepsilon s_f(\boldsymbol{w} \mid \boldsymbol{x}))} 
= r_{\boldsymbol{\delta}}(\boldsymbol{w},\boldsymbol{x})(1 + \varepsilon s_f(\boldsymbol{w} \mid \boldsymbol{x}))^{-1}.
$$
Differentiating with respect to $\varepsilon$ at $\varepsilon = 0$:
$$
\frac{d}{d\varepsilon} r_\varepsilon \Big|_{\varepsilon=0} = -r_{\boldsymbol{\delta}}(\boldsymbol{w},\boldsymbol{x})s_f(\boldsymbol{w} \mid \boldsymbol{x}).
$$
This corresponds to the perturbation $v = -s_f$ in the previous derivation for $r$. Substituting this into the derivative obtained in part (b), we get:
\begin{align*}
\frac{d}{d\varepsilon}\,\mathbb{E}\{\varphi(\boldsymbol{Z}; \psi, \mu, r_\varepsilon)\}\Big|_{\varepsilon=0}
&= \mathbb{E}\big\{\mathbb{E}_g[-s_f(\boldsymbol{W} \mid \boldsymbol{X})\mu(\boldsymbol{X},\boldsymbol{W}) \mid \boldsymbol{X}]\big\} \\
&= -\mathbb{E}\big\{\mathbb{E}_g[s_f(\boldsymbol{W} \mid \boldsymbol{X})\mu(\boldsymbol{X},\boldsymbol{W}) \mid \boldsymbol{X}]\big\}.
\end{align*}
Since $\mathbb{E}_g[s_f\mu \mid \boldsymbol{X}]$ is generally non-zero (unless $\mu$ is constant or $s_f$ is orthogonal to $\mu$ under $g$), the derivative does not vanish. Thus, the influence function is not Neyman-orthogonal with respect to the exposure density $f$.

\paragraph{Robustness to misspecified outcome model.}
Let $\mu_0$ be the true outcome regression and $r$ be the true density ratio. Define the target parameter:
$$
\psi_0
:= \mathbb{E}\big[r_{\boldsymbol{\delta}}(\boldsymbol{W},\boldsymbol{X})\,\mu_0(\boldsymbol{X},\boldsymbol{W})\big]
= \mathbb{E}\Big[\mathbb{E}_f\{r_{\boldsymbol{\delta}}(\boldsymbol{W},\boldsymbol{X})\mu_0(\boldsymbol{X},\boldsymbol{W}) \mid \boldsymbol{X}\}\Big],
$$
which represents the average outcome under the tilted distribution $g$. We now show that at $(\psi_0, r)$, the influence function is globally robust to $\mu$. That is, for any measurable $\tilde{\mu}$:
$$
\mathbb{E}\big[\varphi(\boldsymbol{Z}; \psi_0, \tilde{\mu}, r)\big] = 0.
$$
Consequently, estimators based on the efficient influence function $\varphi$, such as the one-step estimator employed in this work, remain consistent for $\psi_0$ even under misspecification of the outcome regression model.

\paragraph{Proof.}
Fix any measurable function $\tilde{\mu}$ and define:
$$
\tilde{m}(\boldsymbol{X}) := \mathbb{E}_f\big[r_{\boldsymbol{\delta}}(\boldsymbol{W},\boldsymbol{X})\tilde{\mu}(\boldsymbol{X},\boldsymbol{W}) \mid \boldsymbol{X}\big]
= \mathbb{E}_g[\tilde{\mu}(\boldsymbol{X},\boldsymbol{W}) \mid \boldsymbol{X}].
$$
Since $\tilde{m}$ depends only on $\boldsymbol{X}$ and $\mathbb{E}[r_{\boldsymbol{\delta}}(\boldsymbol{W},\boldsymbol{X}) \mid \boldsymbol{X}] = 1$, we have:
$$
\mathbb{E}\big[r_{\boldsymbol{\delta}}(\boldsymbol{W},\boldsymbol{X})\,\tilde{m}(\boldsymbol{X})\big]
= \mathbb{E}\Big[\tilde{m}(\boldsymbol{X})\,\mathbb{E}[r_{\boldsymbol{\delta}}(\boldsymbol{W},\boldsymbol{X}) \mid \boldsymbol{X}]\Big]
= \mathbb{E}\big[\tilde{m}(\boldsymbol{X})\big].
$$
Thus, at $(\psi_0, r)$:
\begin{align*}
\mathbb{E}\big[\varphi(\boldsymbol{Z}; \psi_0, \tilde{\mu}, r)\big]
&= \mathbb{E}\big[r_{\boldsymbol{\delta}}(\boldsymbol{W},\boldsymbol{X})\{Y - \tilde{m}(\boldsymbol{X})\}\big] + \mathbb{E}[\tilde{m}(\boldsymbol{X})] - \psi_0 \\
&= \mathbb{E}[r_{\boldsymbol{\delta}}(\boldsymbol{W},\boldsymbol{X})Y] - \psi_0.
\end{align*}
By the law of iterated expectations:
$$
\mathbb{E}[r_{\boldsymbol{\delta}}(\boldsymbol{W},\boldsymbol{X})Y]
= \mathbb{E}\big[r_{\boldsymbol{\delta}}(\boldsymbol{W},\boldsymbol{X})\,\mu_0(\boldsymbol{X},\boldsymbol{W})\big]
= \psi_0.
$$
Therefore, $\mathbb{E}[\varphi(\boldsymbol{Z}; \psi_0, \tilde{\mu}, r)] = 0$ for all $\tilde{\mu}$. This confirms that consistent estimation of $\psi_0$ relies primarily on the consistency of $\widehat{r}_{\boldsymbol{\delta}}$, rendering the estimator robust to misspecification of $\mu$. \qed

\section{Convergence of Riemannian BFGS under the Gelbrich constraint}
\label{sec:bfgs-induced-gelbrich}

This section provides the algorithmic details and convergence proof for the manifold search used to solve the optimal policy shift problem. We use the conditional tilt $g_{\boldsymbol{\delta}}$, normalizing constant $\nu_{\boldsymbol{\delta}}$, density ratio $r_{\boldsymbol{\delta}}$, and objective $\psi(\boldsymbol{\delta})$ defined in Sections~\ref{sec:Estimands} and~\ref{sec:estimation}. All expectations are under the observed law $P_0$ unless otherwise indicated.

\subsection{Gelbrich Level-Set Search}

Optimizing the policy shift direction $\boldsymbol{\delta}$ subject to the Gelbrich constraint requires optimizing over the level set $\mathcal{M}_c=\{\boldsymbol{\delta}:G(\boldsymbol{\delta})=c^2\}$. We use a Riemannian BFGS method. 

\begin{enumerate}
    \item Tangent space. For any smooth curve $\gamma:(-\varepsilon,\varepsilon)\to\mathcal{M}_c$ satisfying $\gamma(0)=\boldsymbol{\delta}$, the velocity $\gamma'(0)$ lies in the tangent space $T_{\boldsymbol{\delta}}\mathcal{M}_c$. For the level set $\mathcal{M}_c=\{\boldsymbol{\delta}:G(\boldsymbol{\delta})=c^2\}$, this is equivalently
    $$
    T_{\boldsymbol{\delta}}\mathcal{M}_c=\operatorname{Null}\{\nabla G(\boldsymbol{\delta})^\top\}.
    $$
    \item Riemannian gradient. Let $\nabla G(\boldsymbol{\delta})$ denote the Euclidean normal to the level set, and let $\boldsymbol{n}(\boldsymbol{\delta})=\nabla G(\boldsymbol{\delta})/\|\nabla G(\boldsymbol{\delta})\|$. The orthogonal tangent projection is $\mathsf{P}_{\boldsymbol{\delta}}=I-\boldsymbol{n}(\boldsymbol{\delta})\boldsymbol{n}(\boldsymbol{\delta})^\top$, so the Riemannian gradient under the induced Euclidean metric is
    $$
    \operatorname{grad}\psi(\boldsymbol{\delta})=\mathsf{P}_{\boldsymbol{\delta}}\nabla\psi(\boldsymbol{\delta}).
    $$
    \item Retraction. Given a tangent increment $\boldsymbol{\xi}\in T_{\boldsymbol{\delta}}\mathcal{M}_c$, a retraction maps $\boldsymbol{\delta}+\boldsymbol{\xi}$ back to the level set. For $\mathcal{M}_c$, the exact projection retraction is $R_{\boldsymbol{\delta}}(\boldsymbol{\xi})=\Pi(\boldsymbol{\delta}+\boldsymbol{\xi})$, where $\Pi$ is the orthogonal projection onto $\mathcal{M}_c$. In numerical implementation, $\Pi$ is approximated by a single normal-direction correction, which preserves the local first-order accuracy required of a retraction while reducing computational cost.

    \item Line search and accepted tangent increment. At iteration $k$, the current inverse-Hessian approximation $\boldsymbol{B}_k$ is used with $\operatorname{grad}\psi(\boldsymbol{\delta}_k)$ to determine a descent direction in $T_{\boldsymbol{\delta}_k}\mathcal M_c$. A weak Wolfe line search is performed along the corresponding retraction curve, using only trial increments for which the retraction is defined. Let $\boldsymbol{\xi}_k\in T_{\boldsymbol{\delta}_k}\mathcal M_c$ denote the accepted tangent increment. Then
    $$
    \boldsymbol{\delta}_{k+1}
    =
    R_{\boldsymbol{\delta}_k}(\boldsymbol{\xi}_k).
    $$
    \item Vector transport. Tangent vectors at successive iterates are compared using the projection transport
    $$
    \widetilde{\mathcal{T}}_{\boldsymbol{\delta}\to \boldsymbol{\delta}_+}(\boldsymbol{\zeta}) := \mathsf{P}_{\boldsymbol{\delta}_+}\boldsymbol{\zeta}.
    $$
    This transport is generally not isometric, but it is computationally inexpensive and is used to move tangent-space quantities between iterates when forming the RBFGS update.

    \item RBFGS update. Using the accepted tangent increment $\boldsymbol{\xi}_k$ and the projection transport defined above, set
    $$
    \boldsymbol a_k
    =
    \widetilde{\mathcal T}_{\boldsymbol{\delta}_k\to\boldsymbol{\delta}_{k+1}}
    (\boldsymbol{\xi}_k),
    \qquad
    \boldsymbol b_k
    =
    \operatorname{grad}\psi(\boldsymbol{\delta}_{k+1})
    -
    \widetilde{\mathcal T}_{\boldsymbol{\delta}_k\to\boldsymbol{\delta}_{k+1}}
    \operatorname{grad}\psi(\boldsymbol{\delta}_k).
    $$
    When $\boldsymbol b_k^\top\boldsymbol a_k>0$, the inverse-BFGS matrix is updated by
    $$
    \boldsymbol B_{k+1}
    =
    (I-\rho_k\boldsymbol a_k\boldsymbol b_k^\top)
    \boldsymbol B_k
    (I-\rho_k\boldsymbol b_k\boldsymbol a_k^\top)
    +\rho_k\boldsymbol a_k\boldsymbol a_k^\top,
    \qquad
    \rho_k=(\boldsymbol b_k^\top\boldsymbol a_k)^{-1}.
    $$
\end{enumerate}

The resulting Riemannian BFGS iteration is summarized in Algorithm~\ref{alg:rbfgs-appendix}.

\begin{algorithm}[H]
\caption{Riemannian BFGS for Optimal Policy Shift}
\label{alg:rbfgs-appendix}
\begin{algorithmic}
\State Initialize $k=0$ and choose $\boldsymbol{\delta}_0$ on the manifold such that $G(\boldsymbol{\delta}_0)=c^2$.
\State Compute the Euclidean gradient $\nabla\psi(\boldsymbol{\delta}_0)$ and the Riemannian gradient $\operatorname{grad}\psi(\boldsymbol{\delta}_0)=\mathsf{P}_{\boldsymbol{\delta}_0}\nabla\psi(\boldsymbol{\delta}_0)$.
\State Initialize the inverse-Hessian approximation $\boldsymbol{B}_0 = \boldsymbol{I}$.
\While{$\|\operatorname{grad}\psi(\boldsymbol{\delta}_k)\| > \varepsilon_{\mathrm{tol}}$}
  \State Use $\boldsymbol{B}_k$ and $\operatorname{grad}\psi(\boldsymbol{\delta}_k)$ to determine a descent direction.
  \State Perform a weak Wolfe line search along the corresponding retraction curve, using only trial increments for which the retraction is defined; let $\boldsymbol{\xi}_k$ denote the accepted tangent increment.
  \State Update $\boldsymbol{\delta}_{k+1}=R_{\boldsymbol{\delta}_k}(\boldsymbol{\xi}_k)$.
  \State Compute the new Riemannian gradient $\operatorname{grad}\psi(\boldsymbol{\delta}_{k+1})$.
  \State Update $\boldsymbol{B}_{k+1}$ by the RBFGS update using $\boldsymbol{\xi}_k$ and transported gradients.
  \State $k \leftarrow k+1$.
\EndWhile
\State \Return{$\boldsymbol{\delta}_k$}
\end{algorithmic}
\end{algorithm}

In practice, including in the application to air pollution mixtures in the manuscript, we have found that at times this algorithm will find tilt directions $\boldsymbol{\delta}$ whose estimated density-ratio weights are driven by a very small number of observed units, which leads to numerical instability and high variance estimates. To avoid this issue, we propose adding a small penalty to the objective function that we are optimizing in order to avoid such numerical instability. For this discussion, let $\bar r_i(\boldsymbol{\delta})=r_{\boldsymbol{\delta}}(\boldsymbol{W}_i,\boldsymbol{X}_i)/\sum_{j=1}^n r_{\boldsymbol{\delta}}(\boldsymbol{W}_j,\boldsymbol{X}_j)$ and let
$$
S_{\rho_{\mathrm{top}}}(\boldsymbol{\delta})=\sum_{i=1}^{\lceil \rho_{\mathrm{top}} n\rceil}\bar r_{(i)}(\boldsymbol{\delta}),
$$
where $\bar r_{(1)}(\boldsymbol{\delta})\geq\cdots\geq \bar r_{(n)}(\boldsymbol{\delta})$ are the ordered normalized weights. We augment the objective function with
$$
\lambda\left\{\max\left(\frac{S_{\rho_{\mathrm{top}}}(\boldsymbol{\delta})}{\tau_{\mathrm{top}}}-1,0\right)\right\}^2,
$$
where throughout our air pollution application we set $\rho_{\mathrm{top}}=0.01$, $\tau_{\mathrm{top}}=0.15$, and $\lambda=2$. This term is zero unless the largest one percent of normalized weights carries more than 15\% of the total mass. This term penalizes, and therefore discourages, shifts in the exposure distribution that place too much weight on too few observations. We have found that in practice this can lead to much smoother and more stable curves for the BFGS estimand, while still finding directions that lead to nearly optimal policies. 

In the numerical implementation, the gradients used in the projected search, including those of $\widehat G$ and the augmented objective, are evaluated by centered finite differences.

The level-set search above is designed for efficient implementation in the application, using projected finite-difference gradients, projection-based vector transport, and the weight-concentration penalty. The next subsection keeps the same Gelbrich level-set formulation, but studies it with a geometrically exact retraction and vector transport, thereby establishing a rigorous Riemannian BFGS stationarity guarantee for the underlying constrained optimization problem.

\subsection{Geometry of the Gelbrich constraint}

Section~\ref{sec:estimation} formulates Riemannian BFGS iterations for the optimal policy shift over the manifold $\mathcal{M}_c$ defined by the Gelbrich constraint. The stationarity result in Theorem~\ref{thm:rbfgs-stationarity-induced} requires geometric and analytic regularity of the constraint. The following results verify smoothness and generic regularity of the Gelbrich level sets, formalize the projection retraction and vector transport, and establish stationarity of the RBFGS method. All expectations are under the observed law $P_0$ unless otherwise indicated.

\paragraph{Moments of the tilted exposure distribution.}
For a measurable set $A\subseteq\mathbb R^q$, define the post-intervention
marginal law of $\boldsymbol W$ by
\begin{equation}
P_{\boldsymbol{\delta},W}(A)
:=
\mathbb{E}\left[\int_A g_{\boldsymbol{\delta}}(\boldsymbol w\mid \boldsymbol X)\,d\boldsymbol w\right]
=
\mathbb{E}\left[
r_{\boldsymbol{\delta}}(\boldsymbol W,\boldsymbol X)\mathbf 1\{\boldsymbol W\in A\}
\right].
\label{eq:induced-marginal-law}
\end{equation}
Let the baseline marginal mean and covariance be
$$
\boldsymbol{\mu}_W:=\mathbb{E}[\boldsymbol W],
\qquad
\boldsymbol{\Sigma}_W:=\operatorname{Cov}(\boldsymbol W)\in\mathbb S_{++}^q.
$$
For each $\boldsymbol x$, define the conditional cumulant generating function and
its derivatives by
$$
\ell_{\boldsymbol{\delta}}(\boldsymbol x)
:=\log \nu_{\boldsymbol{\delta}}(\boldsymbol x)
=
\log \mathbb{E}\{\exp(\boldsymbol{\delta}^{\top}\boldsymbol W)\mid \boldsymbol X=\boldsymbol x\},
$$
and
\begin{equation}
\boldsymbol{\mu}^W_{\boldsymbol{\delta}}(\boldsymbol x)
:=\nabla_{\boldsymbol{\delta}}\ell_{\boldsymbol{\delta}}(\boldsymbol x),
\qquad
\boldsymbol{\Sigma}^W_{\boldsymbol{\delta}}(\boldsymbol x)
:=\nabla^2_{\boldsymbol{\delta}}\ell_{\boldsymbol{\delta}}(\boldsymbol x).
\label{eq:conditional-cumulants}
\end{equation}
Equivalently, under $g_{\boldsymbol{\delta}}(\cdot\mid \boldsymbol x)$,
$\boldsymbol{\mu}^W_{\boldsymbol{\delta}}(\boldsymbol x)$ and
$\boldsymbol{\Sigma}^W_{\boldsymbol{\delta}}(\boldsymbol x)$ are the conditional
mean and covariance of $\boldsymbol W$. Hence, by the laws of total expectation
and total covariance, the induced marginal moments are
\begin{equation}
\boldsymbol{\mu}_{\boldsymbol{\delta},W}
=
\mathbb{E}\{\boldsymbol{\mu}^W_{\boldsymbol{\delta}}(\boldsymbol X)\},
\qquad
\boldsymbol{\Sigma}_{\boldsymbol{\delta},W}
=
\mathbb{E}[\boldsymbol{\Sigma}^W_{\boldsymbol{\delta}}(\boldsymbol X)]
+\operatorname{Cov}(\boldsymbol{\mu}^W_{\boldsymbol{\delta}}(\boldsymbol X)),
\label{eq:induced-marginal-moments}
\end{equation}
where the centered conditional mean is
$$
\boldsymbol b_{\boldsymbol{\delta}}(\boldsymbol X)
:=
\boldsymbol{\mu}^W_{\boldsymbol{\delta}}(\boldsymbol X)
-\boldsymbol{\mu}_{\boldsymbol{\delta},W},
$$
such that $\operatorname{Cov}(\boldsymbol{\mu}^W_{\boldsymbol{\delta}}(\boldsymbol X)) = \mathbb{E}[\boldsymbol b_{\boldsymbol{\delta}}(\boldsymbol X)\boldsymbol b_{\boldsymbol{\delta}}(\boldsymbol X)^{\top}]$.

The squared Gelbrich constraint defining the constraint manifold is
\begin{equation}
G(\boldsymbol{\delta})
:=
\|\boldsymbol{\mu}_{\boldsymbol{\delta},W}-\boldsymbol{\mu}_W\|_2^2
+
\operatorname{tr}\left[
\boldsymbol{\Sigma}_W
+\boldsymbol{\Sigma}_{\boldsymbol{\delta},W}
-2\left(
\boldsymbol{\Sigma}_W^{1/2}
\boldsymbol{\Sigma}_{\boldsymbol{\delta},W}
\boldsymbol{\Sigma}_W^{1/2}
\right)^{1/2}
\right].
\label{eq:induced-gelbrich-G}
\end{equation}
The expression in \eqref{eq:induced-gelbrich-G} is the Gelbrich moment
functional applied to the baseline marginal law $P_{0,W}$ and the induced marginal
law $P_{\boldsymbol{\delta},W}$.

\paragraph{Regularity assumptions.}
Let $\mathcal D\subseteq\mathbb R^q$ be a nonempty open set containing
$\boldsymbol 0$. The analysis assumes the following conditions.
\begin{description}
\item[(B1)] For every compact $K\subset\mathcal D$ and every integer $r\ge0$,
$$
A_{K,r}(\boldsymbol X)
:=
\sup_{\boldsymbol{\delta}\in K}
\mathbb{E}\left[
\exp(\boldsymbol{\delta}^{\top}\boldsymbol W)(1+\|\boldsymbol W\|^r)
\mid \boldsymbol X
\right]
<\infty
\quad\text{a.s.},
\qquad
\mathbb{E}[A_{K,r}(\boldsymbol X)]<\infty.
$$
In addition, with
$$
D_K(\boldsymbol X):=\inf_{\boldsymbol{\delta}\in K}
\nu_{\boldsymbol{\delta}}(\boldsymbol X),
$$
assume $D_K(\boldsymbol X)>0$ almost surely and
$$
\mathbb{E}\left[
\left\{\frac{A_{K,r}(\boldsymbol X)}{D_K(\boldsymbol X)}\right\}^m
\right]<\infty
\qquad
\text{for all integers }r,m\ge1.
$$
The strengthened envelope condition holds whenever $\boldsymbol W$ has
bounded support and $K$ is compact.
\item[(B2)] The baseline marginal covariance satisfies
$\boldsymbol{\Sigma}_W\in\mathbb S_{++}^q$.
\item[(B3)] For the target $c>0$, the level set
$$
\mathcal M_c:=\{\boldsymbol{\delta}\in\mathcal D:\,G(\boldsymbol{\delta})=c^2\}
$$
is contained in a compact set $K_c\Subset\mathcal D$.
\item[(B4)] The target $c^2$ is a regular value of $G$, meaning that
$\nabla G(\boldsymbol{\delta})\neq \boldsymbol 0$ for all
$\boldsymbol{\delta}\in\mathcal M_c$.
\end{description}

\begin{lemma}[Induced marginal derivatives]
\label{lem:conditional-induced-derivatives}
Under Assumption {\rm(B1)}, there exists a set $\mathcal X_0$ with
$P_0(\boldsymbol X\in\mathcal X_0)=1$ such that, for every
$\boldsymbol x\in\mathcal X_0$, the map
$\boldsymbol{\delta}\mapsto \ell_{\boldsymbol{\delta}}(\boldsymbol x)$ is
real-analytic on $\mathcal D$. Moreover, the maps
$\boldsymbol{\delta}\mapsto \boldsymbol{\mu}_{\boldsymbol{\delta},W}$ and
$\boldsymbol{\delta}\mapsto \boldsymbol{\Sigma}_{\boldsymbol{\delta},W}$ are
$C^\infty$ on $\mathcal D$, with derivatives obtained by dominated
differentiation locally uniformly on every compact $K\Subset\mathcal D$.
For any direction
$\boldsymbol h\in\mathbb R^q$,
\begin{equation}
\mathrm D\boldsymbol{\mu}_{\boldsymbol{\delta},W}[\boldsymbol h]
=
\mathbb{E}\{\boldsymbol{\Sigma}^W_{\boldsymbol{\delta}}(\boldsymbol X)\boldsymbol h\},
\label{eq:D-mubar}
\end{equation}
and
\begin{equation}
\mathrm D\boldsymbol{\Sigma}_{\boldsymbol{\delta},W}[\boldsymbol h]
=
\mathbb{E}\left[
\boldsymbol K_{3,\boldsymbol{\delta}}(\boldsymbol X)[\boldsymbol h]
+\boldsymbol{\Sigma}^W_{\boldsymbol{\delta}}(\boldsymbol X)\boldsymbol h\,
\boldsymbol b_{\boldsymbol{\delta}}(\boldsymbol X)^{\top}
+\boldsymbol b_{\boldsymbol{\delta}}(\boldsymbol X)
\boldsymbol h^{\top}\boldsymbol{\Sigma}^W_{\boldsymbol{\delta}}(\boldsymbol X)
\right],
\label{eq:D-Sigmabar}
\end{equation}
where $\boldsymbol K_{3,\boldsymbol{\delta}}(\boldsymbol X)[\boldsymbol h]$ is defined by
\begin{equation}
\boldsymbol K_{3,\boldsymbol{\delta}}(\boldsymbol x)[\boldsymbol h]
:=
\mathbb{E}_{g_{\boldsymbol{\delta}}}\left[
\{\boldsymbol W-\boldsymbol{\mu}^W_{\boldsymbol{\delta}}(\boldsymbol x)\}
\{\boldsymbol W-\boldsymbol{\mu}^W_{\boldsymbol{\delta}}(\boldsymbol x)\}^{\top}
\boldsymbol h^{\top}
\{\boldsymbol W-\boldsymbol{\mu}^W_{\boldsymbol{\delta}}(\boldsymbol x)\}
\mid \boldsymbol X=\boldsymbol x
\right].
\label{eq:third-cumulant-contraction}
\end{equation}
This is the conditional third-cumulant contraction with respect to the tilted exposure distribution.
\end{lemma}

\begin{proof}
Let $\{K_j:j\ge1\}$ be a compact exhaustion of $\mathcal D$ such that
$K_j\Subset\mathcal D$, $K_j\subset K_{j+1}$,
$\bigcup_{j\ge1}K_j=\mathcal D$, and every compact subset of
$\mathcal D$ is contained in some $K_j$. Assumption (B1), applied over the
countable collection indexed by $(j,r,m)$, yields a set $\mathcal X_0$ with
$P_0(\boldsymbol X\in\mathcal X_0)=1$ such that all conditional moment and
denominator bounds in (B1) hold for every $j,r,m$ and every
$\boldsymbol x\in\mathcal X_0$.

Let $\boldsymbol x\in\mathcal X_0$ and
$\boldsymbol{\delta}_0\in\mathcal D$. There exists $\rho>0$ such that
$$
\boldsymbol{\delta}_0+\rho\boldsymbol s\in\mathcal D
\qquad\text{for every }\boldsymbol s\in\{-1,1\}^q.
$$
There exists a compact set $K_\rho\Subset\mathcal D$ containing the finite set
$\{\boldsymbol{\delta}_0+\rho\boldsymbol s:\boldsymbol s\in\{-1,1\}^q\}$.
For $\|\boldsymbol u\|<\rho$,
$$
\exp\{(\boldsymbol{\delta}_0+\boldsymbol u)^\top\boldsymbol W\}
=
\exp(\boldsymbol{\delta}_0^\top\boldsymbol W)
\sum_{k=0}^{\infty}\frac{(\boldsymbol u^\top\boldsymbol W)^k}{k!}.
$$
The absolute value of the series is bounded by
$$
\exp(\boldsymbol{\delta}_0^\top\boldsymbol W)\exp(\rho\|\boldsymbol W\|)
\le
\exp(\boldsymbol{\delta}_0^\top\boldsymbol W)\exp(\rho\|\boldsymbol W\|_1)
\le
\sum_{\boldsymbol s\in\{-1,1\}^q}
\exp\{(\boldsymbol{\delta}_0+\rho\boldsymbol s)^\top\boldsymbol W\}.
$$
The compact set $K_\rho$ is contained in some $K_j$. The conditional
expectation of the final finite sum, given $\boldsymbol X=\boldsymbol x$, is
finite by (B1) applied to that $K_j$ with $r=0$. Conditional dominated
convergence permits termwise conditional expectation in a
neighborhood of $\boldsymbol{\delta}_0$. Hence
$\nu_{\boldsymbol{\delta}}(\boldsymbol x)$ is real-analytic in a neighborhood
of $\boldsymbol{\delta}_0$. Since $\boldsymbol{\delta}_0\in\mathcal D$ was
arbitrary, $\nu_{\boldsymbol{\delta}}(\boldsymbol x)$ is real-analytic on
$\mathcal D$. It is strictly positive, and the logarithm is real-analytic on
$(0,\infty)$; consequently, $\ell_{\boldsymbol{\delta}}(\boldsymbol x)$ is
real-analytic on $\mathcal D$.

Differentiating the conditional moment generating function gives, for any
$\boldsymbol h,\boldsymbol h_1,\boldsymbol h_2\in\mathbb R^q$,
$$
\mathrm D\nu_{\boldsymbol{\delta}}(\boldsymbol x)[\boldsymbol h]
=
\mathbb{E}\left[
(\boldsymbol h^\top\boldsymbol W)
\exp(\boldsymbol{\delta}^{\top}\boldsymbol W)
\mid \boldsymbol X=\boldsymbol x
\right],
$$
and
$$
\mathrm D^2\nu_{\boldsymbol{\delta}}(\boldsymbol x)
[\boldsymbol h_1,\boldsymbol h_2]
=
\mathbb{E}\left[
(\boldsymbol h_1^\top\boldsymbol W)(\boldsymbol h_2^\top\boldsymbol W)
\exp(\boldsymbol{\delta}^{\top}\boldsymbol W)
\mid \boldsymbol X=\boldsymbol x
\right].
$$
The quotient rule for $\ell_{\boldsymbol{\delta}}=\log\nu_{\boldsymbol{\delta}}$
then yields
$$
\mathrm D\ell_{\boldsymbol{\delta}}(\boldsymbol x)[\boldsymbol h]
=
\boldsymbol h^\top
\mathbb{E}_{g_{\boldsymbol{\delta}}}[\boldsymbol W\mid \boldsymbol X=\boldsymbol x],
$$
and
$$
\mathrm D^2\ell_{\boldsymbol{\delta}}(\boldsymbol x)
[\boldsymbol h_1,\boldsymbol h_2]
=
\boldsymbol h_1^\top
\operatorname{Cov}_{g_{\boldsymbol{\delta}}}(\boldsymbol W\mid \boldsymbol X=\boldsymbol x)
\boldsymbol h_2.
$$
Thus
$$
\boldsymbol{\mu}^W_{\boldsymbol{\delta}}(\boldsymbol x)
=
\mathbb{E}_{g_{\boldsymbol{\delta}}}[\boldsymbol W\mid \boldsymbol X=\boldsymbol x],
\qquad
\boldsymbol{\Sigma}^W_{\boldsymbol{\delta}}(\boldsymbol x)
=
\operatorname{Cov}_{g_{\boldsymbol{\delta}}}(\boldsymbol W\mid \boldsymbol X=\boldsymbol x),
$$
and
$$
\mathrm D\boldsymbol{\mu}^W_{\boldsymbol{\delta}}(\boldsymbol x)[\boldsymbol h]
=
\boldsymbol{\Sigma}^W_{\boldsymbol{\delta}}(\boldsymbol x)\boldsymbol h.
$$
Let $K\Subset\mathcal D$ be compact.
The strengthened envelope condition in (B1) gives integrable envelopes for the conditional
quotients that define the tilted moments; each derivative is bounded by
a finite polynomial in $A_{K,r}(\boldsymbol X)/D_K(\boldsymbol X)$ for sufficiently
large $r$. The same dominated-convergence argument, after integration
over $\boldsymbol X$, justifies differentiating the marginal expectation and
proves \eqref{eq:D-mubar}.

For the derivative of the induced marginal covariance, let
$\boldsymbol Q(\boldsymbol W)$ be any integrable matrix-valued function. Then
$$
\mathrm D
\mathbb{E}_{g_{\boldsymbol{\delta}}}
[\boldsymbol Q(\boldsymbol W)\mid \boldsymbol X=\boldsymbol x][\boldsymbol h]
=
\operatorname{Cov}_{g_{\boldsymbol{\delta}}}
\{\boldsymbol Q(\boldsymbol W),\boldsymbol h^\top\boldsymbol W
\mid \boldsymbol X=\boldsymbol x\},
$$
where the covariance is applied componentwise. Taking
$\boldsymbol Q(\boldsymbol W)=\boldsymbol W\boldsymbol W^\top$ and subtracting
the derivative of
$\boldsymbol{\mu}^W_{\boldsymbol{\delta}}(\boldsymbol{\mu}^W_{\boldsymbol{\delta}})^\top$
gives
$$
\mathrm D\boldsymbol{\Sigma}^W_{\boldsymbol{\delta}}(\boldsymbol x)[\boldsymbol h]
=
\mathbb{E}_{g_{\boldsymbol{\delta}}}\left[
\{\boldsymbol W-\boldsymbol{\mu}^W_{\boldsymbol{\delta}}(\boldsymbol x)\}
\{\boldsymbol W-\boldsymbol{\mu}^W_{\boldsymbol{\delta}}(\boldsymbol x)\}^{\top}
\boldsymbol h^\top
\{\boldsymbol W-\boldsymbol{\mu}^W_{\boldsymbol{\delta}}(\boldsymbol x)\}
\mid \boldsymbol X=\boldsymbol x
\right]
=
\boldsymbol K_{3,\boldsymbol{\delta}}(\boldsymbol x)[\boldsymbol h].
$$
Since
$$
\boldsymbol{\Sigma}_{\boldsymbol{\delta},W}
=
\mathbb{E}[\boldsymbol{\Sigma}^W_{\boldsymbol{\delta}}(\boldsymbol X)]
+
\mathbb{E}[
\boldsymbol b_{\boldsymbol{\delta}}(\boldsymbol X)
\boldsymbol b_{\boldsymbol{\delta}}(\boldsymbol X)^\top],
$$
the derivative satisfies
$$
\mathrm D\boldsymbol b_{\boldsymbol{\delta}}(\boldsymbol X)[\boldsymbol h]
=
\boldsymbol{\Sigma}^W_{\boldsymbol{\delta}}(\boldsymbol X)\boldsymbol h
-
\mathrm D\boldsymbol{\mu}_{\boldsymbol{\delta},W}[\boldsymbol h].
$$
Therefore
$$
\begin{aligned}
\mathrm D\boldsymbol{\Sigma}_{\boldsymbol{\delta},W}[\boldsymbol h]
&=
\mathbb{E}[\boldsymbol K_{3,\boldsymbol{\delta}}(\boldsymbol X)[\boldsymbol h]]
+
\mathbb{E}[
\{\mathrm D\boldsymbol b_{\boldsymbol{\delta}}(\boldsymbol X)[\boldsymbol h]\}
\boldsymbol b_{\boldsymbol{\delta}}(\boldsymbol X)^\top
]\\
&\quad+
\mathbb{E}[
\boldsymbol b_{\boldsymbol{\delta}}(\boldsymbol X)
\{\mathrm D\boldsymbol b_{\boldsymbol{\delta}}(\boldsymbol X)[\boldsymbol h]\}^{\top}
].
\end{aligned}
$$
Because
$\mathbb{E}[\boldsymbol b_{\boldsymbol{\delta}}(\boldsymbol X)]=\boldsymbol0$,
the terms containing
$\mathrm D\boldsymbol{\mu}_{\boldsymbol{\delta},W}[\boldsymbol h]$ vanish
after expectation, yielding \eqref{eq:D-Sigmabar}. Higher derivatives are obtained
by repeated dominated differentiation. Since the compact set
$K\Subset\mathcal D$ was arbitrary, the induced marginal moment maps are
$C^\infty$ on $\mathcal D$.
\end{proof}

\begin{lemma}[Positive definiteness of the induced covariance]
\label{lem:induced-cov-spd}
Under Assumptions {\rm(B1)} and {\rm(B2)},
$\boldsymbol{\Sigma}_{\boldsymbol{\delta},W}\in\mathbb S_{++}^q$ for every
$\boldsymbol{\delta}\in\mathcal D$.
\end{lemma}

\begin{proof}
For every $\boldsymbol{\delta}\in\mathcal D$,
$r_{\boldsymbol{\delta}}(\boldsymbol W,\boldsymbol X)>0$ almost surely and
$\mathbb{E}[r_{\boldsymbol{\delta}}(\boldsymbol W,\boldsymbol X)]=1$. Hence
$P_{\boldsymbol{\delta},W}$ and the baseline marginal law $P_{0,W}$ are equivalent:
if $P_{0,W}(A)=0$, then \eqref{eq:induced-marginal-law} gives
$P_{\boldsymbol{\delta},W}(A)=0$; conversely, if
$P_{\boldsymbol{\delta},W}(A)=0$, then
$\mathbb{E}[r_{\boldsymbol{\delta}}(\boldsymbol W,\boldsymbol X)\mathbf 1\{\boldsymbol W\in A\}]=0$,
and strict positivity of $r_{\boldsymbol{\delta}}$ implies $P_{0,W}(A)=0$.

Let $\boldsymbol u\neq\boldsymbol0$. If
$\boldsymbol u^\top\boldsymbol{\Sigma}_{\boldsymbol{\delta},W}\boldsymbol u=0$,
then $\boldsymbol u^\top\boldsymbol W$ is constant
$P_{\boldsymbol{\delta},W}$-almost surely. By equivalence this is also true
$P_{0,W}$-almost surely, which implies
$\boldsymbol u^\top\boldsymbol{\Sigma}_W\boldsymbol u=0$, contradicting
$\boldsymbol{\Sigma}_W\in\mathbb S_{++}^q$. Thus
$\boldsymbol{\Sigma}_{\boldsymbol{\delta},W}$ is positive definite.
\end{proof}

\begin{lemma}[Fr\'echet derivative of the induced Gelbrich constraint]
\label{lem:DG-induced-gelbrich}
Under Assumptions {\rm(B1)} and {\rm(B2)}, the map $G:\mathcal D\to\mathbb R$
defined in \eqref{eq:induced-gelbrich-G} is $C^\infty$. Its Fr\'echet derivative in
the direction $\boldsymbol h\in\mathbb R^q$ is
\begin{equation}
\mathrm DG(\boldsymbol{\delta})[\boldsymbol h]
=
2(\boldsymbol{\mu}_{\boldsymbol{\delta},W}-\boldsymbol{\mu}_W)^\top
\mathrm D\boldsymbol{\mu}_{\boldsymbol{\delta},W}[\boldsymbol h]
+
\left\langle
I-\boldsymbol L_{\boldsymbol{\delta}},
\mathrm D\boldsymbol{\Sigma}_{\boldsymbol{\delta},W}[\boldsymbol h]
\right\rangle_F,
\label{eq:DG-induced-gelbrich}
\end{equation}
where $\langle A,B\rangle_F=\operatorname{tr}(A^\top B)$ and
\begin{equation}
\boldsymbol L_{\boldsymbol{\delta}}
:=
\boldsymbol{\Sigma}_W^{1/2}
\left(
\boldsymbol{\Sigma}_W^{1/2}
\boldsymbol{\Sigma}_{\boldsymbol{\delta},W}
\boldsymbol{\Sigma}_W^{1/2}
\right)^{-1/2}
\boldsymbol{\Sigma}_W^{1/2}.
\label{eq:L-gelbrich-cov-gradient}
\end{equation}
\end{lemma}

\begin{proof}
By Lemmas~\ref{lem:conditional-induced-derivatives} and
\ref{lem:induced-cov-spd}, both
$\boldsymbol{\mu}_{\boldsymbol{\delta},W}$ and
$\boldsymbol{\Sigma}_{\boldsymbol{\delta},W}$ are $C^\infty$, and
$\boldsymbol{\Sigma}_{\boldsymbol{\delta},W}$ remains in the open cone
$\mathbb S_{++}^q$. It suffices to differentiate the covariance term
$$
\Phi(\boldsymbol\Theta)
:=
\operatorname{tr}\left[
\boldsymbol{\Sigma}_W
+\boldsymbol\Theta
-2(\boldsymbol{\Sigma}_W^{1/2}\boldsymbol\Theta\boldsymbol{\Sigma}_W^{1/2})^{1/2}
\right],
\qquad
\boldsymbol\Theta\in\mathbb S_{++}^q.
$$
The derivative of the matrix square-root map is obtained by inverting the derivative
of the square map through the Lyapunov equation; see
\citep[Eqs.~(12)-(16)]{malago2018wasserstein}. For $\boldsymbol A\in\mathbb S_{++}^q$,
$$
\mathrm D\,\operatorname{tr}(\boldsymbol A^{1/2})[\boldsymbol V]
=
\frac12\operatorname{tr}(\boldsymbol A^{-1/2}\boldsymbol V).
$$
Set
$\boldsymbol A(\boldsymbol\Theta)
=\boldsymbol{\Sigma}_W^{1/2}\boldsymbol\Theta\boldsymbol{\Sigma}_W^{1/2}$.
For a symmetric perturbation $\boldsymbol H$,
$$
\begin{aligned}
\mathrm D\Phi(\boldsymbol\Theta)[\boldsymbol H]
&=
\operatorname{tr}(\boldsymbol H)
-2\cdot
\frac12
\operatorname{tr}\left\{
\boldsymbol A(\boldsymbol\Theta)^{-1/2}
\boldsymbol{\Sigma}_W^{1/2}\boldsymbol H\boldsymbol{\Sigma}_W^{1/2}
\right\}\\
&=
\operatorname{tr}\left[
\left\{
I-
\boldsymbol{\Sigma}_W^{1/2}
\boldsymbol A(\boldsymbol\Theta)^{-1/2}
\boldsymbol{\Sigma}_W^{1/2}
\right\}
\boldsymbol H
\right].
\end{aligned}
$$
Evaluating at
$\boldsymbol\Theta=\boldsymbol{\Sigma}_{\boldsymbol{\delta},W}$ gives
$\boldsymbol L_{\boldsymbol{\delta}}$ in
\eqref{eq:L-gelbrich-cov-gradient}. Combining
\eqref{eq:D-mubar} and \eqref{eq:D-Sigmabar} with the chain rule gives
\eqref{eq:DG-induced-gelbrich}. Smoothness follows because matrix inversion
and the principal square root are smooth on $\mathbb S_{++}^q$.
\end{proof}

\begin{lemma}[Generic regularity of Gelbrich level sets]
\label{lem:generic-regularity-induced}
Under Assumptions {\rm(B1)} and {\rm(B2)}, define the critical set
$$
\mathcal C:=\{\boldsymbol{\delta}\in\mathcal D:\nabla G(\boldsymbol{\delta})=\boldsymbol0\},
\qquad
\mathcal V:=G(\mathcal C)\subset[0,\infty).
$$
Then $\mathcal V$ has Lebesgue measure zero. Consequently, for every
$c>0$ such that $c^2\notin\mathcal V$, the nonempty level set
$$
\mathcal M_c=\{\boldsymbol{\delta}\in\mathcal D:\,G(\boldsymbol{\delta})=c^2\}
$$
is a $C^\infty$ embedded hypersurface of $\mathbb R^q$, and
$$
T_{\boldsymbol{\delta}}\mathcal M_c
=
\{\boldsymbol \xi\in\mathbb R^q:\nabla G(\boldsymbol{\delta})^\top
\boldsymbol \xi=0\}.
$$
\end{lemma}

\begin{proof}
By Lemma~\ref{lem:DG-induced-gelbrich}, $G$ is $C^\infty$ on the open set
$\mathcal D$. Sard's theorem applied to the scalar map
$G:\mathcal D\to\mathbb R$ implies that the set of critical values
$\mathcal V$ has Lebesgue measure zero. If $c^2\notin\mathcal V$, then every
$\boldsymbol{\delta}\in\mathcal M_c$ satisfies
$\nabla G(\boldsymbol{\delta})\neq\boldsymbol0$. The regular level-set theorem
therefore implies that $\mathcal M_c$ is a $C^\infty$ embedded hypersurface.
For a smooth curve $\gamma(t)\in\mathcal M_c$ with
$\gamma(0)=\boldsymbol{\delta}$, differentiating
$G(\gamma(t))=c^2$ at $t=0$ gives
$\nabla G(\boldsymbol{\delta})^\top\gamma'(0)=0$. Conversely, the regular
level-set theorem identifies the kernel of $\mathrm DG(\boldsymbol{\delta})$ with
the tangent space, giving the displayed characterization.
\end{proof}

\begin{lemma}[Local quadratic control near the origin]
\label{lem:local-quadratic-induced}
Assume {\rm(B1)} and {\rm(B2)}, and suppose
$$
\boldsymbol A_0
:=
\mathbb{E}\{\operatorname{Var}(\boldsymbol W\mid \boldsymbol X)\}
\in\mathbb S_{++}^q.
$$
Then there exist constants $\rho>0$ and $0<\lambda_-\le\lambda_+<\infty$
such that
\begin{equation}
\lambda_- \|\boldsymbol{\delta}\|^2
\le
G(\boldsymbol{\delta})
\le
\lambda_+ \|\boldsymbol{\delta}\|^2,
\qquad
\|\boldsymbol{\delta}\|\le\rho.
\label{eq:local-quadratic-induced}
\end{equation}
\end{lemma}

\begin{proof}
At $\boldsymbol{\delta}=\boldsymbol0$,
$$
\boldsymbol{\mu}^W_{\boldsymbol0}(\boldsymbol X)
=\mathbb{E}[\boldsymbol W\mid \boldsymbol X],
\qquad
\boldsymbol{\Sigma}^W_{\boldsymbol0}(\boldsymbol X)
=\operatorname{Var}(\boldsymbol W\mid \boldsymbol X).
$$
Thus
$\boldsymbol{\mu}_{\boldsymbol0,W}=\boldsymbol{\mu}_W$ and, by the law of
total covariance,
$\boldsymbol{\Sigma}_{\boldsymbol0,W}=\boldsymbol{\Sigma}_W$. From
\eqref{eq:D-mubar},
$$
\mathrm D\boldsymbol{\mu}_{\boldsymbol0,W}[\boldsymbol h]
=
\mathbb{E}\{\boldsymbol{\Sigma}^W_{\boldsymbol0}(\boldsymbol X)\boldsymbol h\}
=
\boldsymbol A_0\boldsymbol h.
$$
Since $\boldsymbol{\mu}_{\boldsymbol{\delta},W}$ is $C^2$ in a
neighborhood of $\boldsymbol0$,
$$
\boldsymbol{\mu}_{\boldsymbol{\delta},W}-\boldsymbol{\mu}_W
=
\boldsymbol A_0\boldsymbol{\delta}+O(\|\boldsymbol{\delta}\|^2).
$$
Let $\lambda_{\min}(\boldsymbol A_0)>0$. There is $\rho_1>0$ such that, for
$\|\boldsymbol{\delta}\|\le\rho_1$,
$$
\|\boldsymbol{\mu}_{\boldsymbol{\delta},W}-\boldsymbol{\mu}_W\|
\ge
\frac12\lambda_{\min}(\boldsymbol A_0)\|\boldsymbol{\delta}\|.
$$
The covariance term of $G$ is nonnegative. Hence
$$
G(\boldsymbol{\delta})
\ge
\|\boldsymbol{\mu}_{\boldsymbol{\delta},W}-\boldsymbol{\mu}_W\|^2
\ge
\frac14\lambda_{\min}(\boldsymbol A_0)^2\|\boldsymbol{\delta}\|^2
$$
for $\|\boldsymbol{\delta}\|\le\rho_1$.

For the upper bound, Lemma~\ref{lem:conditional-induced-derivatives} gives
$$
\|\boldsymbol{\mu}_{\boldsymbol{\delta},W}-\boldsymbol{\mu}_W\|
=O(\|\boldsymbol{\delta}\|),
\qquad
\|\boldsymbol{\Sigma}_{\boldsymbol{\delta},W}-\boldsymbol{\Sigma}_W\|_F
=O(\|\boldsymbol{\delta}\|)
$$
as $\boldsymbol{\delta}\to\boldsymbol0$. The covariance map
$\Phi(\boldsymbol\Theta)$ in the proof of Lemma~\ref{lem:DG-induced-gelbrich}
satisfies $\Phi(\boldsymbol{\Sigma}_W)=0$ and
$\mathrm D\Phi(\boldsymbol{\Sigma}_W)=0$, because
$\boldsymbol L_{\boldsymbol0}=I$. Since $\Phi$ is $C^2$ in a neighborhood
of $\boldsymbol{\Sigma}_W$, Taylor's theorem gives
$$
\Phi(\boldsymbol{\Sigma}_{\boldsymbol{\delta},W})
=
O(\|\boldsymbol{\Sigma}_{\boldsymbol{\delta},W}-\boldsymbol{\Sigma}_W\|_F^2)
=
O(\|\boldsymbol{\delta}\|^2).
$$
Combining the mean and covariance bounds yields the upper inequality in
\eqref{eq:local-quadratic-induced} after possibly reducing $\rho_1$. The asserted inequalities hold after setting $\lambda_-$, $\lambda_+$, and $\rho$ accordingly.
\end{proof}

\subsection{Retraction, projection, and vector transport.}
Let $c^2$ be a regular value of $G$, and let $\mathcal M_c$ be nonempty. For
$\boldsymbol{\delta}\in\mathcal M_c$, define the unit normal vector
$$
\boldsymbol n(\boldsymbol{\delta})
:=
\frac{\nabla G(\boldsymbol{\delta})}{\|\nabla G(\boldsymbol{\delta})\|},
$$
and the orthogonal projection onto the tangent space by
\begin{equation}
\mathsf{P}_{\boldsymbol{\delta}}
:=
I-\boldsymbol n(\boldsymbol{\delta})\boldsymbol n(\boldsymbol{\delta})^\top.
\label{eq:tangent-projection-induced}
\end{equation}
The Riemannian gradient of the objective
$\psi$ on $\mathcal M_c$, under the metric inherited from the Euclidean inner
product on $\mathbb R^q$, is
$$
\operatorname{grad}\psi(\boldsymbol{\delta})
=
\mathsf{P}_{\boldsymbol{\delta}}\nabla\psi(\boldsymbol{\delta}).
$$
\begin{lemma}[Projection retraction and isometric transport]
\label{lem:projection-retraction-induced}
Assume {\rm(B1)}-{\rm(B4)}. Then $\mathcal M_c$ admits a smooth projection
retraction in a tubular neighborhood. There exists an open
neighborhood $\mathcal N$ of $\mathcal M_c$ and a smooth nearest-point
projection $\Pi:\mathcal N\to\mathcal M_c$ such that
$$
R_{\boldsymbol{\delta}}(\boldsymbol{\xi})
:=
\Pi(\boldsymbol{\delta}+\boldsymbol{\xi}),
\qquad
\boldsymbol{\xi}\in T_{\boldsymbol{\delta}}\mathcal M_c
$$
is a retraction for all sufficiently small $\boldsymbol{\xi}$. Moreover, parallel
transport along the retraction curve
$$
t\mapsto R_{\boldsymbol{\delta}}(t\boldsymbol{\xi}),\qquad 0\le t\le1,
$$
defines an isometric vector transport
$$
\mathcal T_{\boldsymbol{\delta},\boldsymbol{\xi}}
:T_{\boldsymbol{\delta}}\mathcal M_c
\to
T_{R_{\boldsymbol{\delta}}(\boldsymbol{\xi})}\mathcal M_c
$$
that is continuous in all arguments on its domain.
\end{lemma}

\begin{proof}
By Lemma~\ref{lem:generic-regularity-induced} and Assumption (B4),
$\mathcal M_c$ is a smooth embedded hypersurface. Assumption (B3) implies that it
is contained in a compact subset of $\mathcal D$. Since $\mathcal M_c$ is the
inverse image of the closed set $\{c^2\}$ under the continuous map $G$, it is
closed in that compact set and hence compact. The tubular neighborhood theorem
therefore gives an open neighborhood $\mathcal N$ of $\mathcal M_c$ on which
the nearest-point projection
$\Pi:\mathcal N\to\mathcal M_c$ is smooth.

For $\boldsymbol{\delta}\in\mathcal M_c$,
$$
R_{\boldsymbol{\delta}}(\boldsymbol0)=\Pi(\boldsymbol{\delta})=\boldsymbol{\delta}.
$$
The derivative of the nearest-point projection at a point of the embedded
submanifold is the orthogonal projection onto the tangent space. Therefore
$$
\mathrm D R_{\boldsymbol{\delta}}(\boldsymbol0)[\boldsymbol\xi]
=
\mathsf{P}_{\boldsymbol{\delta}}\boldsymbol\xi
=
\boldsymbol\xi,
\qquad
\boldsymbol\xi\in T_{\boldsymbol{\delta}}\mathcal M_c,
$$
which verifies the retraction property.

The curve $t\mapsto R_{\boldsymbol{\delta}}(t\boldsymbol\xi)$ is smooth for
$\boldsymbol\xi$ sufficiently small. Parallel transport with respect to the
Levi-Civita connection induced by the Euclidean metric maps tangent vectors at
$\boldsymbol{\delta}$ linearly to tangent vectors at
$R_{\boldsymbol{\delta}}(\boldsymbol\xi)$ and preserves the Riemannian norm.
Smooth dependence of solutions to the parallel-transport ordinary differential
equation on the initial condition and on the curve gives continuity in all
arguments. This proves the claim.
\end{proof}

\subsection{Convergence proof}

Let
$$
m_{\boldsymbol{\delta}}(\boldsymbol x)
:=
\mathbb{E}_{g_{\boldsymbol{\delta}}}
[\mu(\boldsymbol X,\boldsymbol W)\mid \boldsymbol X=\boldsymbol x]
=
\frac{
\eta_{\boldsymbol{\delta}}(\boldsymbol x)}
{\nu_{\boldsymbol{\delta}}(\boldsymbol x)},
$$
where
$$
\eta_{\boldsymbol{\delta}}(\boldsymbol x)
:=
\mathbb{E}[
\mu(\boldsymbol x,\boldsymbol W)\exp(\boldsymbol{\delta}^{\top}\boldsymbol W)
\mid \boldsymbol X=\boldsymbol x].
$$
Then the target objective is
$$
\psi(\boldsymbol{\delta})
=
\mathbb{E}\{m_{\boldsymbol{\delta}}(\boldsymbol X)\}.
$$
\begin{lemma}[Smoothness of the objective]
\label{lem:objective-regularity-induced}
Let $K\Subset\mathcal D$ be compact. Suppose that for some constants
$C_\mu,p<\infty$,
$$
|\mu(\boldsymbol x,\boldsymbol w)|
\le
C_\mu(1+\|\boldsymbol w\|^p),
$$
and that Assumption {\rm(B1)} holds for the moment orders required by
the first and second derivatives below. Then $\psi$ is $C^2$ on an open neighborhood of $K$, and
$\nabla\psi$ is Lipschitz on $K$.
\end{lemma}

\begin{proof}
There exists an open set $U$ such that $K\subset U\Subset\mathcal D$.
It suffices to establish the derivative bounds on an arbitrary compact subset of
$U$. For $\boldsymbol{\delta}$ in such a compact set and direction
$\boldsymbol h$,
conditional dominated convergence gives
$$
\mathrm D\eta_{\boldsymbol{\delta}}(\boldsymbol x)[\boldsymbol h]
=
\mathbb{E}[
\mu(\boldsymbol x,\boldsymbol W)
\boldsymbol h^\top\boldsymbol W
\exp(\boldsymbol{\delta}^{\top}\boldsymbol W)
\mid \boldsymbol X=\boldsymbol x],
$$
and for directions $\boldsymbol h_1,\boldsymbol h_2$,
$$
\mathrm D^2\eta_{\boldsymbol{\delta}}(\boldsymbol x)
[\boldsymbol h_1,\boldsymbol h_2]
=
\mathbb{E}[
\mu(\boldsymbol x,\boldsymbol W)
(\boldsymbol h_1^\top\boldsymbol W)
(\boldsymbol h_2^\top\boldsymbol W)
\exp(\boldsymbol{\delta}^{\top}\boldsymbol W)
\mid \boldsymbol X=\boldsymbol x].
$$
The corresponding derivatives of $\nu_{\boldsymbol{\delta}}$ are given in
Lemma~\ref{lem:conditional-induced-derivatives}. Since
$\nu_{\boldsymbol{\delta}}(\boldsymbol x)>0$, the quotient rule gives first and
second derivatives of
$m_{\boldsymbol{\delta}}(\boldsymbol x)
=\eta_{\boldsymbol{\delta}}(\boldsymbol x)/\nu_{\boldsymbol{\delta}}(\boldsymbol x)$.
The polynomial growth of $\mu$, combined with the strengthened exponential-moment
and denominator envelope in (B1), provides an integrable bound for
$$
\sup_{\boldsymbol{\delta}\in K}|m_{\boldsymbol{\delta}}(\boldsymbol X)|,
\qquad
\sup_{\boldsymbol{\delta}\in K}\|\nabla_{\boldsymbol{\delta}}
m_{\boldsymbol{\delta}}(\boldsymbol X)\|,
\qquad
\sup_{\boldsymbol{\delta}\in K}\|\nabla^2_{\boldsymbol{\delta}}
m_{\boldsymbol{\delta}}(\boldsymbol X)\|.
$$
Dominated differentiation yields
$$
\nabla\psi(\boldsymbol{\delta})
=
\mathbb{E}[\nabla_{\boldsymbol{\delta}}
m_{\boldsymbol{\delta}}(\boldsymbol X)],
\qquad
\nabla^2\psi(\boldsymbol{\delta})
=
\mathbb{E}[\nabla^2_{\boldsymbol{\delta}}
m_{\boldsymbol{\delta}}(\boldsymbol X)].
$$
The displayed envelope implies
$\sup_{\boldsymbol{\delta}\in K}\|\nabla^2\psi(\boldsymbol{\delta})\|<\infty$.
The mean-value theorem then gives Lipschitz continuity of $\nabla\psi$ on $K$.
\end{proof}

\begin{corollary}[Compactness and line-search regularity]
\label{cor:compact-linesearch-induced}
Assume {\rm(B1)}-{\rm(B4)} and the conditions of
Lemma~\ref{lem:objective-regularity-induced} on a compact set
$K_c\Subset\mathcal D$ containing $\mathcal M_c$. Then:
\begin{enumerate}
\item[\textnormal{(i)}] $\mathcal M_c$ is compact.
\item[\textnormal{(ii)}] The Riemannian gradient
$\operatorname{grad}\psi(\boldsymbol{\delta})
=\mathsf{P}_{\boldsymbol{\delta}}\nabla\psi(\boldsymbol{\delta})$
is Lipschitz on $\mathcal M_c$.
\item[\textnormal{(iii)}] For every descent direction
$\boldsymbol\eta\in T_{\boldsymbol{\delta}}\mathcal M_c$, if the retraction
curve $\gamma(t)=R_{\boldsymbol{\delta}}(t\boldsymbol\eta)$ satisfies the
one-dimensional domain and boundedness conditions in
\citep[Proposition~1]{ring2012optimization}, then a positive step size satisfying the Wolfe conditions exists
along that curve.
\end{enumerate}
\end{corollary}

\begin{proof}
Part (i) was established in the proof of
Lemma~\ref{lem:projection-retraction-induced}: $\mathcal M_c$ is closed in the
compact set $K_c$.

For (ii), the map
$\boldsymbol{\delta}\mapsto\nabla G(\boldsymbol{\delta})$ is $C^1$, and
$\nabla G$ is bounded away from zero on the compact set $\mathcal M_c$ by
(B4). Hence
$\boldsymbol{\delta}\mapsto\boldsymbol n(\boldsymbol{\delta})$ and
$\boldsymbol{\delta}\mapsto \mathsf{P}_{\boldsymbol{\delta}}$ are Lipschitz on
$\mathcal M_c$. By Lemma~\ref{lem:objective-regularity-induced},
$\nabla\psi$ is Lipschitz on $K_c$ and bounded on $K_c$. For
$\boldsymbol{\delta},\boldsymbol{\delta}'\in\mathcal M_c$,
$$
\begin{aligned}
\|\operatorname{grad}\psi(\boldsymbol{\delta})
-\operatorname{grad}\psi(\boldsymbol{\delta}')\|
&\le
\|\mathsf{P}_{\boldsymbol{\delta}}
\{\nabla\psi(\boldsymbol{\delta})-\nabla\psi(\boldsymbol{\delta}')\}\|\\
&\quad+
\|(\mathsf{P}_{\boldsymbol{\delta}}-\mathsf{P}_{\boldsymbol{\delta}'})
\nabla\psi(\boldsymbol{\delta}')\|\\
&\le L\|\boldsymbol{\delta}-\boldsymbol{\delta}'\|
\end{aligned}
$$
for a finite constant $L$.

For (iii), let
$h(t)=\psi(R_{\boldsymbol{\delta}}(t\boldsymbol\eta))$. The stated domain
condition ensures that $h$ is defined and continuously differentiable on the
one-dimensional interval required by \citep[Proposition~1]{ring2012optimization}.
If $\boldsymbol\eta$ is a descent direction, then
$h'(0)=\langle\operatorname{grad}\psi(\boldsymbol{\delta}),
\boldsymbol\eta\rangle<0$. The stated boundedness condition gives the lower
bound required by that proposition. Hence a positive step size satisfying the Wolfe
conditions exists.
\end{proof}

\begin{lemma}[Nonemptiness and containment]
\label{lem:feasible-containment-induced}
Assume the conditions of Lemma~\ref{lem:local-quadratic-induced}. Let
$r_0\in(0,\rho]$. Then there exists $c_0>0$ such that, for any
$c\in(0,c_0)$, there exists
$\boldsymbol{\delta}_0\in\mathcal M_c\cap\overline B(\boldsymbol0,r_0)$.
If, in addition, Assumptions {\rm(B1)}-{\rm(B4)} hold, then every sequence
$\{\boldsymbol{\delta}_k\}_{k\ge0}\subset\mathcal M_c$ with initial point
$\boldsymbol{\delta}_0$ and
$\psi(\boldsymbol{\delta}_{k+1})\le\psi(\boldsymbol{\delta}_k)$ for all $k$ remains
in the compact sublevel
$$
\mathcal L
:=
\{\boldsymbol{\delta}\in\mathcal M_c:\,
\psi(\boldsymbol{\delta})\le \psi(\boldsymbol{\delta}_0)\}.
$$
\end{lemma}

\begin{proof}
By Lemma~\ref{lem:local-quadratic-induced}, for
$\|\boldsymbol{\delta}\|\le\rho$,
$$
\lambda_-\|\boldsymbol{\delta}\|^2
\le
G(\boldsymbol{\delta})
\le
\lambda_+\|\boldsymbol{\delta}\|^2.
$$
Fix any unit vector $\boldsymbol v$ and define $q(t)=G(t\boldsymbol v)$ on
$[0,r_0]$. Then $q$ is continuous, $q(0)=0$, and
$$
q(r_0)\ge\lambda_- r_0^2.
$$
Let $c_0=\sqrt{\lambda_-}\,r_0$. If $0<c<c_0$, then
$0<c^2<q(r_0)$, and the intermediate value theorem gives
$t^\ast\in(0,r_0)$ such that $G(t^\ast\boldsymbol v)=c^2$. Thus
$\boldsymbol{\delta}_0=t^\ast\boldsymbol v\in
\mathcal M_c\cap\overline B(\boldsymbol0,r_0)$.

For the containment claim, Corollary~\ref{cor:compact-linesearch-induced} gives compactness
of $\mathcal M_c$. Since $\psi$ is continuous on $\mathcal M_c$, the set
$\mathcal L$ is closed in a compact set and hence compact. The monotonicity condition implies
$\psi(\boldsymbol{\delta}_k)\le\psi(\boldsymbol{\delta}_0)$ for every $k$; hence
$\boldsymbol{\delta}_k\in\mathcal L$.
\end{proof}

For the Gelbrich level-set problem equipped with the projection retraction and
isometric vector transport in Lemma~\ref{lem:projection-retraction-induced}, we
use the cautious Riemannian BFGS framework described in
\citep{huang2018riemannian}: the quasi-Newton update is applied only when the
curvature condition is satisfied, and the line search is restricted to accepted
steps for which the projection retraction is defined. This yields the following
stationarity result.

\begin{theorem}[Cautious RBFGS stationarity]
\label{thm:rbfgs-stationarity-induced}
Suppose Assumptions {\rm(B1)}-{\rm(B4)} and the conditions of
Lemmas~\ref{lem:objective-regularity-induced} and
\ref{lem:feasible-containment-induced} hold. Let $\{\boldsymbol{\delta}_k\}$ denote
the sequence of iterates. Then
$$
\liminf_{k\to\infty}
\|\operatorname{grad}\psi(\boldsymbol{\delta}_k)\|
=0.
$$
Moreover, every accumulation point of a subsequence along which the Riemannian
gradient norm tends to zero is a Riemannian stationary point of $\psi$ on
$\mathcal M_c$.
\end{theorem}

\begin{proof}
By Lemma~\ref{lem:feasible-containment-induced}, the iterates remain in the
compact initial sublevel set $\mathcal L$, so the compactness requirement in
\citep[Assumption~4.1]{huang2018riemannian} is satisfied. The projection
retraction is local, and the line search is restricted to accepted steps whose
retraction curves remain in its domain.

It remains to verify Lipschitz continuous differentiability with respect to the
specified vector transport. For an accepted tangent increment $\boldsymbol\xi$,
let $\gamma(t)=R_{\boldsymbol{\delta}}(t\boldsymbol\xi)$. Parallel transport gives
$$
\mathcal T_{\boldsymbol{\delta},\boldsymbol\xi}
\operatorname{grad}\psi(\boldsymbol{\delta})
-
\operatorname{grad}\psi(\gamma(1))
=
-\int_0^1
\mathcal T_{\gamma(t),\boldsymbol\xi}^{1\leftarrow t}
\nabla_{\dot\gamma(t)}
\operatorname{grad}\psi(\gamma(t))\,dt.
$$
On the compact tubular neighborhood containing the accepted retraction curves,
the covariant derivative of $\operatorname{grad}\psi$ is bounded and
$\|\dot\gamma(t)\|\le C\|\boldsymbol\xi\|$. Since parallel transport preserves
norms, there exists $L_T<\infty$ such that
$$
\left\|
\mathcal T_{\boldsymbol{\delta},\boldsymbol\xi}
\operatorname{grad}\psi(\boldsymbol{\delta})
-
\operatorname{grad}\psi(R_{\boldsymbol{\delta}}(\boldsymbol\xi))
\right\|
\le
L_T\|\boldsymbol\xi\|.
$$
Thus the condition in \citep[Assumption~4.2]{huang2018riemannian} holds.

Applying \citep[Theorem~4.2]{huang2018riemannian} gives
$$
\liminf_{k\to\infty}
\|\operatorname{grad}\psi(\boldsymbol{\delta}_k)\|
=0.
$$
Since $\mathcal L$ is compact, the sequence has accumulation points. If
$\boldsymbol{\delta}_{k_j}\to\boldsymbol{\delta}_\star$ and
$\|\operatorname{grad}\psi(\boldsymbol{\delta}_{k_j})\|\to0$, continuity of the
Riemannian gradient implies
$\operatorname{grad}\psi(\boldsymbol{\delta}_\star)=\boldsymbol0$. Thus
$\boldsymbol{\delta}_\star$ is Riemannian stationary.
\end{proof}

\section{Efficiency and Minimax Lower Bounds}
\label{sec:appendix-efficiency-minimax}

\subsection*{Assumptions and notation}

Let $\boldsymbol{Z}=(\boldsymbol{X},\boldsymbol{W},Y)$ with $\boldsymbol{X}\in\mathbb R^p$, $\boldsymbol{W}\in\mathbb R^q$, and $Y\in\mathbb R$.  Write $f(\boldsymbol{w}\mid \boldsymbol{x})$ for the exposure density, and define the exponential tilt
\begin{align*}
g_{\boldsymbol{\delta}}(\boldsymbol{w}\mid \boldsymbol{x})
&:=\frac{\exp\{\boldsymbol{\delta}^\top \boldsymbol{w}\}\,f(\boldsymbol{w}\mid \boldsymbol{x})}{\nu_{\boldsymbol{\delta}}(\boldsymbol{x})},\\
\nu_{\boldsymbol{\delta}}(\boldsymbol{x})
&:=\mathbb{E}_f\big[\exp\{\boldsymbol{\delta}^\top\boldsymbol{W}\}\mid \boldsymbol{X}=\boldsymbol{x}\big],
\end{align*}
and let
$$
r_{\boldsymbol{\delta}}(\boldsymbol{w},\boldsymbol{x})
:=\frac{g_{\boldsymbol{\delta}}(\boldsymbol{w}\mid \boldsymbol{x})}{f(\boldsymbol{w}\mid \boldsymbol{x})}.
$$
Let $\mu(\boldsymbol{x},\boldsymbol{w}):=\mathbb{E}[Y\mid \boldsymbol{X}=\boldsymbol{x},\boldsymbol{W}=\boldsymbol{w}]$.

Assumptions:
\begin{description}
\item[(C1) Bounded outcomes] There exists $M<\infty$ such that $|Y|\le M$ a.s. Hence $|\mu(\boldsymbol{x},\boldsymbol{w})|\le M$.
\item[(C2) Nondegenerate noise] $0<\underline{\sigma}^2\le \Var(Y\mid \boldsymbol{X},\boldsymbol{W})$ a.s.
\item[(C3) Bounded exposures] $\|\boldsymbol{W}\|\le M_s<\infty$ a.s.
\item[(C4) Bounded tilt region] Let $C_\delta>0$ be arbitrarily large but fixed, and consider tilt directions satisfying $\|\boldsymbol{\delta}\|M_s\le C_\delta$.
\item[(C5) Model richness] The class of distributions $\mathcal{P}$ is sufficiently rich. There exist constants $M_\phi,\varepsilon_0>0$ such that for any $P_0\in\mathcal P$, if $\phi\in L_\infty(P_0)$ with $\mathbb{E}_{P_0}[\phi(\boldsymbol Z)]=0$ and $\|\phi\|_\infty\le M_\phi$, then for all $|\varepsilon|\le\varepsilon_0$, the density $p_0(1+\varepsilon\phi)$ corresponds to a distribution $P_1\in\mathcal P$.
\end{description}

Identification and target.
Under the identification conditions stated in the main text,
\begin{align*}
\psi_P(\boldsymbol{\delta})
&=\mathbb{E}_P\Big[\int \mu_P(\boldsymbol{X},\boldsymbol{w})g_{\boldsymbol{\delta},P}(\boldsymbol{w}\mid \boldsymbol{X})\,d\boldsymbol{w}\Big],\\
\theta_P(\boldsymbol{\delta})
&:=\psi_P(\boldsymbol{\delta})-\psi_P(\boldsymbol{0}).
\end{align*}

\subsection{Minimax lower bound}

Define
$$
\boldsymbol{\Sigma}_{W\mid X}
:=
\mathbb{E}\{\Var(\boldsymbol{W}\mid \boldsymbol{X})\}.
$$

\begin{theorem}[Minimax lower bound]
\label{thm:minimax-projected}
Suppose \textnormal{(C1)}--\textnormal{(C5)} hold. If $\boldsymbol{\Sigma}_{W\mid X}$ is positive definite, then there exists a constant $C>0$, not depending on $n$ or $\boldsymbol{\delta}$, such that, for every nonzero $\boldsymbol{\delta}$ with $\|\boldsymbol{\delta}\|\le C_\delta/M_s$,
$$
\inf_{\hat\theta}\ \sup_{P\in\mathcal P}
\mathbb{E}_P\big[(\hat\theta-\theta_P(\boldsymbol{\delta}))^2\big]
\ge
C\,\frac{
\boldsymbol{\delta}^\top\boldsymbol{\Sigma}_{W\mid X}\boldsymbol{\delta}
}{n}.
$$
\end{theorem}

\begin{proof}[Proof of Theorem~\ref{thm:minimax-projected}]
Fix $P_0\in\mathcal P$ and a nonzero $\boldsymbol{\delta}$ with $\|\boldsymbol{\delta}\|\le C_\delta/M_s$. All vector norms are Euclidean. All population quantities are evaluated under the baseline law $P_0$ unless explicitly indicated. Write
\begin{align*}
D_{\boldsymbol{\delta}}(\boldsymbol{Z})
&:=\{r_{\boldsymbol{\delta}}(\boldsymbol{W},\boldsymbol{X})-1\}
\{Y-\mu(\boldsymbol{X},\boldsymbol{W})\},\\
V_{\boldsymbol{\delta}}
&:=\mathbb{E}_{P_0}\{D_{\boldsymbol{\delta}}(\boldsymbol{Z})^2\}\\
&=\mathbb{E}_{P_0}\Big[
\{r_{\boldsymbol{\delta}}(\boldsymbol{W},\boldsymbol{X})-1\}^2
\Var(Y\mid \boldsymbol{X},\boldsymbol{W})
\Big],
\end{align*}
and
$$
Q_{\boldsymbol{\delta}}
:=
\boldsymbol{\delta}^\top\boldsymbol{\Sigma}_{W\mid X}\boldsymbol{\delta}.
$$
By positive definiteness, $Q_{\boldsymbol{\delta}}>0$.

\medskip\noindent
\textit{Ratio bounds and covariance scale.}
Let $T=\boldsymbol{\delta}^\top\boldsymbol{W}$ and $\tau=\|\boldsymbol{\delta}\|M_s\le C_\delta$.  Then
$$
e^{-\tau}\le e^T\le e^\tau,
\qquad
e^{-\tau}\le \nu_{\boldsymbol{\delta}}(\boldsymbol{X})\le e^\tau,
\qquad
e^{-2\tau}\le r_{\boldsymbol{\delta}}(\boldsymbol{W},\boldsymbol{X})\le e^{2\tau}.
$$
Thus
$$
\|r_{\boldsymbol{\delta}}-1\|_\infty
\le e^{2\tau}-1
\le 2e^{2C_\delta}M_s\|\boldsymbol{\delta}\|.
$$
Since $|Y-\mu(\boldsymbol{X},\boldsymbol{W})|\le 2M$ by \textnormal{(C1)},
$$
\|D_{\boldsymbol{\delta}}\|_\infty
\le 4Me^{2C_\delta}M_s\|\boldsymbol{\delta}\|.
$$

Let $T'$ be an independent copy of $T$ conditional on $\boldsymbol{X}$.  Conditional on $\boldsymbol X$,
\begin{align*}
\mathbb{E}_{P_0}\big[\{r_{\boldsymbol{\delta}}(\boldsymbol{W},\boldsymbol{X})-1\}^2\mid \boldsymbol{X}\big]
&=\Var_{P_0}\{r_{\boldsymbol{\delta}}(\boldsymbol{W},\boldsymbol{X})\mid \boldsymbol{X}\}\\
&=\frac{\Var_{P_0}(e^T\mid \boldsymbol{X})}{\nu_{\boldsymbol{\delta}}(\boldsymbol{X})^2}\\
&=\frac{\mathbb{E}_{P_0}\{(e^T-e^{T'})^2\mid \boldsymbol{X}\}}{2\nu_{\boldsymbol{\delta}}(\boldsymbol{X})^2}\\
&\ge e^{-4C_\delta}\Var_{P_0}(T\mid \boldsymbol{X}).
\end{align*}
Indeed, $|T|,|T'|\le C_\delta$ implies $|e^T-e^{T'}|\ge e^{-C_\delta}|T-T'|$, while $\nu_{\boldsymbol{\delta}}(\boldsymbol{X})\le e^{C_\delta}$.  Therefore, by \textnormal{(C2)},
\begin{align*}
V_{\boldsymbol{\delta}}
&\ge \underline{\sigma}^2\,
\mathbb{E}_{P_0}\big[\{r_{\boldsymbol{\delta}}(\boldsymbol{W},\boldsymbol{X})-1\}^2\big]\\
&\ge \underline{\sigma}^2 e^{-4C_\delta}\,
\mathbb{E}_{P_0}\{\Var_{P_0}(\boldsymbol{\delta}^\top\boldsymbol{W}\mid \boldsymbol{X})\}\\
&=\underline{\sigma}^2 e^{-4C_\delta}Q_{\boldsymbol{\delta}}.
\end{align*}
In particular, let $\lambda_{\min}:=\lambda_{\min}(\boldsymbol{\Sigma}_{W\mid X})$. Then
$$
\sqrt{V_{\boldsymbol{\delta}}}
\ge
\underline{\sigma}e^{-2C_\delta}\sqrt{\lambda_{\min}}\,
\|\boldsymbol{\delta}\|.
$$
Combining this lower bound with the preceding sup-norm bound gives
$$
\left\|\frac{D_{\boldsymbol{\delta}}}{\sqrt{V_{\boldsymbol{\delta}}}}\right\|_\infty
\le
\frac{4Me^{4C_\delta}M_s}{\underline{\sigma}\sqrt{\lambda_{\min}}}
=:B_*<\infty,
$$
where $B_*$ does not depend on $\boldsymbol{\delta}$.

\medskip\noindent
\textit{Two-point construction.}
Set
$$
a_*:=\min\{1,M_\phi/B_*\},
\qquad
\phi_{\boldsymbol{\delta}}
:=a_*\frac{D_{\boldsymbol{\delta}}}{\sqrt{V_{\boldsymbol{\delta}}}}.
$$
Then $\mathbb{E}_{P_0}[\phi_{\boldsymbol{\delta}}]=0$ and $\|\phi_{\boldsymbol{\delta}}\|_\infty\le M_\phi$.  For
$$
\kappa_0:=\min\Big\{\varepsilon_0,\frac{1}{2M_\phi},\frac12\Big\},
\qquad
\varepsilon_n:=\frac{\kappa_0}{\sqrt n},
$$
let $P_1$ have density $p_1=p_0(1+\varepsilon_n\phi_{\boldsymbol{\delta}})$.  By \textnormal{(C5)}, $P_1\in\mathcal P$.  Also $\int p_1=1$, and $p_1\ge p_0/2$.

Because $\mathbb{E}_{P_0}[\phi_{\boldsymbol{\delta}}\mid \boldsymbol{X},\boldsymbol{W}]=0$,
$$
p_1(\boldsymbol{x},\boldsymbol{w})
=p_0(\boldsymbol{x},\boldsymbol{w})\,
\mathbb{E}_{P_0}[1+\varepsilon_n\phi_{\boldsymbol{\delta}}(\boldsymbol Z)\mid \boldsymbol{X}=\boldsymbol{x},\boldsymbol{W}=\boldsymbol{w}]
=p_0(\boldsymbol{x},\boldsymbol{w}).
$$
Thus the joint law of $(\boldsymbol{X},\boldsymbol{W})$ is the same under $P_0$ and $P_1$.  Hence the exposure density and $r_{\boldsymbol{\delta}}$ are unchanged along this path.  Moreover,
\begin{align*}
\mu_{P_1}(\boldsymbol{x},\boldsymbol{w})-\mu_{P_0}(\boldsymbol{x},\boldsymbol{w})
&=\varepsilon_n\,\mathbb{E}_{P_0}\big[
Y\phi_{\boldsymbol{\delta}}(\boldsymbol Z)
\mid \boldsymbol{X}=\boldsymbol{x},\boldsymbol{W}=\boldsymbol{w}
\big]\\
&=\varepsilon_n a_*
\frac{r_{\boldsymbol{\delta}}(\boldsymbol{w},\boldsymbol{x})-1}{\sqrt{V_{\boldsymbol{\delta}}}}
\Var_{P_0}(Y\mid \boldsymbol{X}=\boldsymbol{x},\boldsymbol{W}=\boldsymbol{w}).
\end{align*}
Since
$$
\theta_P(\boldsymbol{\delta})
=
\mathbb{E}_{P}\big[\mu_P(\boldsymbol{X},\boldsymbol{W})
\{r_{\boldsymbol{\delta},P}(\boldsymbol{W},\boldsymbol{X})-1\}\big]
$$
whenever $r_{\boldsymbol{\delta},P}$ is formed from the exposure law of $P$. Since $r_{\boldsymbol{\delta},P_1}=r_{\boldsymbol{\delta},P_0}=r_{\boldsymbol{\delta}}$, it follows that
\begin{align*}
\theta_{P_1}(\boldsymbol{\delta})-\theta_{P_0}(\boldsymbol{\delta})
&=\varepsilon_n a_*\frac{1}{\sqrt{V_{\boldsymbol{\delta}}}}
\mathbb{E}_{P_0}\Big[
\{r_{\boldsymbol{\delta}}(\boldsymbol{W},\boldsymbol{X})-1\}^2
\Var(Y\mid \boldsymbol{X},\boldsymbol{W})
\Big]\\
&=\varepsilon_n a_*\sqrt{V_{\boldsymbol{\delta}}}.
\end{align*}
Finally,
$$
\chi^2(P_1,P_0)
=\mathbb{E}_{P_0}\{(\varepsilon_n\phi_{\boldsymbol{\delta}})^2\}
=a_*^2\varepsilon_n^2,
$$
and hence
$$
D_{\mathrm{KL}}(P_1\|P_0)
\le \chi^2(P_1,P_0)
=a_*^2\varepsilon_n^2.
$$

\medskip\noindent
\textit{Le Cam bound.}
The product measures satisfy
$$
D_{\mathrm{KL}}(P_1^{\otimes n}\|P_0^{\otimes n})
=nD_{\mathrm{KL}}(P_1\|P_0)
\le a_*^2\kappa_0^2
\le \kappa_0^2.
$$
Pinsker's inequality gives
$$
1-\mathrm{TV}(P_0^{\otimes n},P_1^{\otimes n})
\ge 1-\frac{\kappa_0}{\sqrt2}
=:c_{\mathrm{TV}}>0.
$$
For any estimator $\hat\theta$, Le Cam's two-point inequality yields
\begin{align*}
\sup_{P\in\{P_0,P_1\}}
\mathbb{E}_{P}\big[(\hat\theta-\theta_P(\boldsymbol{\delta}))^2\big]
&\ge
\frac18
\big\{\theta_{P_1}(\boldsymbol{\delta})-\theta_{P_0}(\boldsymbol{\delta})\big\}^2
\big\{1-\mathrm{TV}(P_0^{\otimes n},P_1^{\otimes n})\big\}\\
&\ge
\frac{c_{\mathrm{TV}}}{8}
\frac{\kappa_0^2a_*^2}{n}
V_{\boldsymbol{\delta}}\\
&\ge
\frac{c_{\mathrm{TV}}}{8}
\frac{\kappa_0^2a_*^2\underline{\sigma}^2e^{-4C_\delta}}{n}
Q_{\boldsymbol{\delta}}.
\end{align*}
Since $\{P_0,P_1\}\subset\mathcal P$, taking $\inf_{\hat\theta}$ and $\sup_{P\in\mathcal P}$ yields the same lower bound, proving the theorem with $C=(c_{\mathrm{TV}}/8)\kappa_0^2a_*^2\underline{\sigma}^2e^{-4C_\delta}$.
\end{proof}

\subsection{Variance bound for the efficient influence function}

Recall
$$
\boldsymbol{\Sigma}_{W\mid X}
:=
\mathbb{E}\{\Var(\boldsymbol{W}\mid \boldsymbol{X})\}.
$$
For the tilt proportional to $\exp(\boldsymbol{\delta}^\top\boldsymbol{W})$, the efficient influence function for $\psi(\boldsymbol{\delta})$ is
$$
\varphi_{\psi(\boldsymbol{\delta})}(\boldsymbol{Z})
=
r_{\boldsymbol{\delta}}(\boldsymbol{W},\boldsymbol{X})
\{Y-m_{\boldsymbol{\delta}}(\boldsymbol{X})\}
+m_{\boldsymbol{\delta}}(\boldsymbol{X})-\psi(\boldsymbol{\delta}),
$$
where
$$
m_{\boldsymbol{\delta}}(\boldsymbol{X})
:=
\mathbb{E}_{g_{\boldsymbol{\delta}}}\{\mu(\boldsymbol{X},\boldsymbol{W})\mid \boldsymbol{X}\}.
$$
Thus
$$
\varphi_{\theta(\boldsymbol{\delta})}
:=
\varphi_{\psi(\boldsymbol{\delta})}
-\varphi_{\psi(\boldsymbol{0})}.
$$

\begin{theorem}[Variance bound of the efficient influence function]
\label{thm:var-iff-appendix}
Suppose \textnormal{(C1)}--\textnormal{(C4)} hold. Then, for all nonzero $\boldsymbol{\delta}$ satisfying $\|\boldsymbol{\delta}\|\le C_\delta/M_s$, there exist constants $0<c_{\mathrm{low}}\le c_{\mathrm{up}}<\infty$, not depending on $\boldsymbol{\delta}$, such that
$$
c_{\mathrm{low}}\,
\boldsymbol{\delta}^\top\boldsymbol{\Sigma}_{W\mid X}\boldsymbol{\delta}
\le
\Var\!\big\{\varphi_{\theta(\boldsymbol{\delta})}(\boldsymbol{Z})\big\}
\le
c_{\mathrm{up}}\,
\boldsymbol{\delta}^\top\boldsymbol{\Sigma}_{W\mid X}\boldsymbol{\delta}.
$$
\end{theorem}

\begin{proof}
Since $\varphi_{\psi(\boldsymbol{0})}(\boldsymbol{Z})=Y-\psi(\boldsymbol{0})$,
\begin{align*}
\varphi_{\theta(\boldsymbol{\delta})}(\boldsymbol{Z})
&=
\{r_{\boldsymbol{\delta}}(\boldsymbol{W},\boldsymbol{X})-1\}
\{Y-\mu(\boldsymbol{X},\boldsymbol{W})\}\\
&\quad+
\{r_{\boldsymbol{\delta}}(\boldsymbol{W},\boldsymbol{X})-1\}
\{\mu(\boldsymbol{X},\boldsymbol{W})-m_{\boldsymbol{\delta}}(\boldsymbol{X})\}
-\theta(\boldsymbol{\delta}).
\end{align*}
The two random terms before centering are uncorrelated because
$$
\mathbb{E}\{Y-\mu(\boldsymbol{X},\boldsymbol{W})\mid \boldsymbol{X},\boldsymbol{W}\}=0.
$$
Therefore,
\begin{align*}
\Var\!\big\{\varphi_{\theta(\boldsymbol{\delta})}(\boldsymbol{Z})\big\}
&=
\mathbb{E}\Big[
\{r_{\boldsymbol{\delta}}(\boldsymbol{W},\boldsymbol{X})-1\}^2
\Var(Y\mid \boldsymbol{X},\boldsymbol{W})
\Big]\\
&\quad+
\Var\!\Big(
\{r_{\boldsymbol{\delta}}(\boldsymbol{W},\boldsymbol{X})-1\}
\{\mu(\boldsymbol{X},\boldsymbol{W})-m_{\boldsymbol{\delta}}(\boldsymbol{X})\}
\Big).
\end{align*}

Let $\boldsymbol{W}'$ be an independent copy of $\boldsymbol{W}$ conditional on $\boldsymbol{X}$.
Since $\mathbb{E}\{r_{\boldsymbol{\delta}}(\boldsymbol{W},\boldsymbol{X})\mid \boldsymbol{X}\}=1$,
\begin{align*}
\mathbb{E}\big[\{r_{\boldsymbol{\delta}}(\boldsymbol{W},\boldsymbol{X})-1\}^2\mid \boldsymbol{X}\big]
&=
\Var\{r_{\boldsymbol{\delta}}(\boldsymbol{W},\boldsymbol{X})\mid \boldsymbol{X}\}\\
&=
\frac{
\Var\{\exp(\boldsymbol{\delta}^\top\boldsymbol{W})\mid \boldsymbol{X}\}
}{
\nu_{\boldsymbol{\delta}}(\boldsymbol{X})^2
}\\
&=
\frac{
\mathbb{E}\big[
\{\exp(\boldsymbol{\delta}^\top\boldsymbol{W})-\exp(\boldsymbol{\delta}^\top\boldsymbol{W}')\}^2
\mid \boldsymbol{X}
\big]
}{
2\nu_{\boldsymbol{\delta}}(\boldsymbol{X})^2
}.
\end{align*}
By \textnormal{(C3)} and \textnormal{(C4)}, $|\boldsymbol{\delta}^\top\boldsymbol{W}|\le C_\delta$ and
$e^{-C_\delta}\le\nu_{\boldsymbol{\delta}}(\boldsymbol{X})\le e^{C_\delta}$.
The mean-value theorem gives
\begin{align*}
e^{-4C_\delta}
\Var(\boldsymbol{\delta}^\top\boldsymbol{W}\mid \boldsymbol{X})
&\le
\mathbb{E}\big[\{r_{\boldsymbol{\delta}}(\boldsymbol{W},\boldsymbol{X})-1\}^2\mid \boldsymbol{X}\big]\\
&\le
e^{4C_\delta}
\Var(\boldsymbol{\delta}^\top\boldsymbol{W}\mid \boldsymbol{X}).
\end{align*}
Taking expectations yields
$$
e^{-4C_\delta}\,
\boldsymbol{\delta}^\top\boldsymbol{\Sigma}_{W\mid X}\boldsymbol{\delta}
\le
\mathbb{E}\big[\{r_{\boldsymbol{\delta}}(\boldsymbol{W},\boldsymbol{X})-1\}^2\big]
\le
e^{4C_\delta}\,
\boldsymbol{\delta}^\top\boldsymbol{\Sigma}_{W\mid X}\boldsymbol{\delta}.
$$

The lower bound follows from \textnormal{(C2)}:
\begin{align*}
\Var\!\big\{\varphi_{\theta(\boldsymbol{\delta})}(\boldsymbol{Z})\big\}
&\ge
\mathbb{E}\Big[
\{r_{\boldsymbol{\delta}}(\boldsymbol{W},\boldsymbol{X})-1\}^2
\Var(Y\mid \boldsymbol{X},\boldsymbol{W})
\Big]\\
&\ge
\underline{\sigma}^2 e^{-4C_\delta}\,
\boldsymbol{\delta}^\top\boldsymbol{\Sigma}_{W\mid X}\boldsymbol{\delta}.
\end{align*}
For the upper bound, \textnormal{(C1)} implies
$\Var(Y\mid \boldsymbol{X},\boldsymbol{W})\le 4M^2$ and
$|\mu(\boldsymbol{X},\boldsymbol{W})-m_{\boldsymbol{\delta}}(\boldsymbol{X})|\le 2M$.
Hence
\begin{align*}
\Var\!\big\{\varphi_{\theta(\boldsymbol{\delta})}(\boldsymbol{Z})\big\}
&\le
4M^2\,\mathbb{E}\big[\{r_{\boldsymbol{\delta}}(\boldsymbol{W},\boldsymbol{X})-1\}^2\big]
+4M^2\,\mathbb{E}\big[\{r_{\boldsymbol{\delta}}(\boldsymbol{W},\boldsymbol{X})-1\}^2\big]\\
&\le
8M^2 e^{4C_\delta}\,
\boldsymbol{\delta}^\top\boldsymbol{\Sigma}_{W\mid X}\boldsymbol{\delta}.
\end{align*}
Thus the result holds with
$c_{\mathrm{low}}=\underline{\sigma}^2e^{-4C_\delta}$ and
$c_{\mathrm{up}}=8M^2e^{4C_\delta}$.
\end{proof}

\section{Convergence and normality}
\label{sec:appendix-convergence-normality}

This auxiliary result uses the assumptions and notation from Section~\ref{sec:appendix-efficiency-minimax} and translates rate conditions on the tilted nuisances $(m_{\boldsymbol{\delta}}, r_{\boldsymbol{\delta}})$ into standard $L_2$ rates for the outcome regression $\mu$ and the exposure density $f$.

\subsection{Tilted Nuisance Stability}

For a fixed finite tilt, the asymptotic linearity of the one-step estimator is governed by the tilt-specific nuisance functions
$$
m_{\boldsymbol{\delta}}(\boldsymbol{x}) \;=\; \mathbb{E}_{g_{\boldsymbol{\delta}}}\!\big[\mu(\boldsymbol{X},\boldsymbol{W}) \mid \boldsymbol{X}=\boldsymbol{x}\big],
\qquad
r_{\boldsymbol{\delta}}(\boldsymbol{w},\boldsymbol{x}) \;=\; \frac{g_{\boldsymbol{\delta}}(\boldsymbol{w} \mid \boldsymbol{x})}{f(\boldsymbol{w} \mid \boldsymbol{x})},
$$
The following reduction expresses the required $L_2$ control of these tilt-specific objects in terms of the observed-data nuisance components.

For any fixed finite tilt $\boldsymbol{\delta}$, it suffices to estimate the outcome regression and exposure density,
$$
\mu(\boldsymbol{x},\boldsymbol{w}) \;=\; \mathbb{E}[Y \mid \boldsymbol{X}=\boldsymbol{x}, \boldsymbol{W}=\boldsymbol{w}],
\qquad
f(\boldsymbol{w} \mid \boldsymbol{x}),
$$
which induce $m_{\boldsymbol{\delta}}$ and $r_{\boldsymbol{\delta}}$ through the formulas in Lemma~\ref{lem:reduction-mu-pi}.

Under bounded support for $\boldsymbol{W}$ and the finite-tilt condition $\|\boldsymbol{\delta}\| \le \Delta$, the maps $(\mu,f) \mapsto r_{\boldsymbol{\delta}}$ and $(\mu,f) \mapsto m_{\boldsymbol{\delta}}$ are Lipschitz in $L_2(P)$: there exist fixed finite constants $C_1(\Delta), C_2(\Delta)$ such that, for any estimators $(\widehat{\mu}, \widehat{f})$,
$$
\|\widehat{r}_{\boldsymbol{\delta}} - r_{\boldsymbol{\delta}}\|_2
\;\le\; C_1(\Delta)\,\|\widehat{f} - f\|_2,
\qquad
\|\widehat{m}_{\boldsymbol{\delta}} - m_{\boldsymbol{\delta}}\|_2
\;\le\; C_2(\Delta)\big(\|\widehat{\mu} - \mu\|_2 + \|\widehat{f} - f\|_2\big),
$$
where $(\widehat{m}_{\boldsymbol{\delta}}, \widehat{r}_{\boldsymbol{\delta}})$ are obtained from $(\widehat{\mu}, \widehat{f})$ by the same formulas as above; see Lemma~\ref{lem:reduction-mu-pi}. Consequently, the product condition
$$
\|\widehat{r}_{\boldsymbol{\delta}} - r_{\boldsymbol{\delta}}\|_2\,\|\widehat{m}_{\boldsymbol{\delta}} - m_{\boldsymbol{\delta}}\|_2 = o_P(n^{-1/2}),
$$
which ensures that the second-order remainder is negligible, is implied by the more interpretable requirement that both $\widehat{\mu}$ and $\widehat{f}$ converge at rate $n^{-1/4}$ in $L_2(P)$; see Corollary~\ref{cor:rates-mu-pi}. These $L_2$ rates for $(\widehat{\mu}, \widehat{f})$, and hence for $(\widehat{m}_{\boldsymbol{\delta}}, \widehat{r}_{\boldsymbol{\delta}})$, can be guaranteed under standard smoothness and complexity conditions by the highly adaptive lasso \citep{van2017generally}, which attains near-optimal convergence rates for a broad nonparametric class.

\begin{lemma}[Lipschitz stability of tilted nuisance functions]
\label{lem:reduction-mu-pi}
Suppose \textnormal{(C1)} and \textnormal{(C3)} hold. Let $f(\boldsymbol{w} \mid \boldsymbol{x})$ denote the conditional density of $\boldsymbol{W}$ given $\boldsymbol{X}=\boldsymbol{x}$, with support $\mathcal{W} \subset \mathbb{R}^q$ of finite Lebesgue measure $|\mathcal{W}|<\infty$. Suppose there exist constants $0 < f_{\min} \le f_{\max} < \infty$ such that
$$
f_{\min} \;\le\; f(\boldsymbol{w} \mid \boldsymbol{x}) \;\le\; f_{\max}
\qquad\text{for all $(\boldsymbol{x}, \boldsymbol{w})$ in the support of $(\boldsymbol{X}, \boldsymbol{W})$.}
$$
Let $\mu(\boldsymbol{x}, \boldsymbol{w}) = \mathbb{E}[Y \mid \boldsymbol{X}=\boldsymbol{x}, \boldsymbol{W}=\boldsymbol{w}]$, and fix $\Delta < \infty$. For any $\boldsymbol{\delta}$ with $\|\boldsymbol{\delta}\| \le \Delta$ define
\begin{align*}
\nu_{\boldsymbol{\delta}}(\boldsymbol{x})
&:=\int_{\mathcal{W}} \exp\{\boldsymbol{\delta}^\top \boldsymbol{w}\}\,f(\boldsymbol{w} \mid \boldsymbol{x})\,d\boldsymbol{w},\\
\eta_{\boldsymbol{\delta}}(\boldsymbol{x})
&:=\int_{\mathcal{W}} \exp\{\boldsymbol{\delta}^\top \boldsymbol{w}\}\,\mu(\boldsymbol{x}, \boldsymbol{w})\,f(\boldsymbol{w} \mid \boldsymbol{x})\,d\boldsymbol{w},
\end{align*}
and
\begin{align*}
m_{\boldsymbol{\delta}}(\boldsymbol{x})
&:= \frac{\eta_{\boldsymbol{\delta}}(\boldsymbol{x})}{\nu_{\boldsymbol{\delta}}(\boldsymbol{x})},\\
r_{\boldsymbol{\delta}}(\boldsymbol{w}, \boldsymbol{x})
&:= \frac{\exp\{\boldsymbol{\delta}^\top \boldsymbol{w}\}}{\nu_{\boldsymbol{\delta}}(\boldsymbol{x})}.
\end{align*}

Let $\widehat{\mu}, \widehat{f}$ be estimators of $\mu, f$, with $\widehat f$
nonnegative and bounded above with probability tending to one, and construct
\begin{align*}
\widehat{\nu}_{\boldsymbol{\delta}}(\boldsymbol{x})
&:= \int_{\mathcal{W}} \exp\{\boldsymbol{\delta}^\top \boldsymbol{w}\}\,\widehat{f}(\boldsymbol{w} \mid \boldsymbol{x})\,d\boldsymbol{w},\\
\widehat{\eta}_{\boldsymbol{\delta}}(\boldsymbol{x})
&:= \int_{\mathcal{W}} \exp\{\boldsymbol{\delta}^\top \boldsymbol{w}\}\,\widehat{\mu}(\boldsymbol{x}, \boldsymbol{w})\,\widehat{f}(\boldsymbol{w} \mid \boldsymbol{x})\,d\boldsymbol{w},
\end{align*}

Define the truncated normalizing estimator and the induced tilted nuisance estimators by
\begin{align*}
\widehat{\nu}_{\boldsymbol{\delta}}^{\dagger}(\boldsymbol{x})
&:=\min\Big\{\max\big(\widehat{\nu}_{\boldsymbol{\delta}}(\boldsymbol{x}),\ e^{-\tau_\Delta}/2\big),\ 2e^{\tau_\Delta}\Big\},\\
\widehat{m}_{\boldsymbol{\delta}}(\boldsymbol{x})
&:= \frac{\widehat{\eta}_{\boldsymbol{\delta}}(\boldsymbol{x})}{\widehat{\nu}_{\boldsymbol{\delta}}^{\dagger}(\boldsymbol{x})},\\
\widehat{r}_{\boldsymbol{\delta}}(\boldsymbol{w}, \boldsymbol{x})
&:= \frac{\exp\{\boldsymbol{\delta}^\top \boldsymbol{w}\}}{\widehat{\nu}_{\boldsymbol{\delta}}^{\dagger}(\boldsymbol{x})}.
\end{align*}
Then, with probability tending to one, there exist finite constants
$C_1(\Delta), C_2(\Delta) < \infty$, depending only on
$(\Delta, f_{\min}, f_{\max}, |\mathcal{W}|, M, M_s)$ and the high-probability upper
bound on $\widehat f$, such that
\begin{align}
\|\widehat{r}_{\boldsymbol{\delta}} - r_{\boldsymbol{\delta}}\|_2
&\le C_1(\Delta)\,\|\widehat{f} - f\|_2,
\label{eq:r-lipschitz}\\[0.3em]
\|\widehat{m}_{\boldsymbol{\delta}} - m_{\boldsymbol{\delta}}\|_2
&\le C_2(\Delta)\,\big(\|\widehat{\mu} - \mu\|_2 + \|\widehat{f} - f\|_2\big),
\label{eq:m-lipschitz}
\end{align}
where $\|\cdot\|_2$ denotes the $L_2(P)$-norm with respect to the law of $(\boldsymbol{X}, \boldsymbol{W})$.
\end{lemma}

\begin{proof}
Since $\|\boldsymbol{\delta}\| \le \Delta$ and $\boldsymbol{W}$ is bounded, there exists $\tau_\Delta < \infty$ such that
$|\exp\{\boldsymbol{\delta}^\top \boldsymbol{W}\}| \le e^{\tau_\Delta}$ almost surely. Hence, for all such $\boldsymbol{\delta}$,
$$
e^{-\tau_\Delta} \;\le\; \nu_{\boldsymbol{\delta}}(\boldsymbol{x}) \;\le\; e^{\tau_\Delta}
\qquad\text{for all $\boldsymbol{x}$.}
$$
In particular, $\nu_{\boldsymbol{\delta}}$ is bounded away from $0$ and $\infty$; the same holds for $\widehat{\nu}_{\boldsymbol{\delta}}^{\dagger}$ by construction.

\medskip\noindent
\textit{Control of $\widehat{\nu}_{\boldsymbol{\delta}} - \nu_{\boldsymbol{\delta}}$.}
By definition,
$$
\widehat{\nu}_{\boldsymbol{\delta}}(\boldsymbol{x}) - \nu_{\boldsymbol{\delta}}(\boldsymbol{x})
= \int_{\mathcal{W}} \exp\{\boldsymbol{\delta}^\top \boldsymbol{w}\}\big\{\widehat{f}(\boldsymbol{w} \mid \boldsymbol{x}) - f(\boldsymbol{w} \mid \boldsymbol{x})\big\}\,d\boldsymbol{w}.
$$
Taking absolute values and using the bound on the exponential tilt gives
$$
|\widehat{\nu}_{\boldsymbol{\delta}}(\boldsymbol{x}) - \nu_{\boldsymbol{\delta}}(\boldsymbol{x})|
\le e^{\tau_\Delta} \int_{\mathcal{W}} \big|\widehat{f}(\boldsymbol{w} \mid \boldsymbol{x}) - f(\boldsymbol{w} \mid \boldsymbol{x})\big|\,d\boldsymbol{w}.
$$
By Cauchy--Schwarz and finiteness of $|\mathcal{W}|$,
$$
\int_{\mathcal{W}} \big|\widehat{f} - f\big|\,d\boldsymbol{w}
\le |\mathcal{W}|^{1/2} \Big(\int_{\mathcal{W}} \big|\widehat{f} - f\big|^2\,d\boldsymbol{w}\Big)^{1/2}.
$$
Squaring both sides and integrating over $\boldsymbol{x}$ with respect to $P_{\boldsymbol{X}}$ gives
\begin{align*}
\int\{\widehat{\nu}_{\boldsymbol{\delta}}(\boldsymbol{x})-\nu_{\boldsymbol{\delta}}(\boldsymbol{x})\}^2\,dP_{\boldsymbol{X}}(\boldsymbol{x})
&\le e^{2\tau_\Delta}|\mathcal{W}|
   \int\!\!\int_{\mathcal{W}}\big|\widehat{f}(\boldsymbol{w}\mid \boldsymbol{x})-f(\boldsymbol{w}\mid \boldsymbol{x})\big|^2\,d\boldsymbol{w}\,dP_{\boldsymbol{X}}(\boldsymbol{x}).
\end{align*}
By definition,
$$
\|\widehat{f}-f\|_2^2
= \int\!\!\int_{\mathcal{W}}\big|\widehat{f}(\boldsymbol{w}\mid \boldsymbol{x})-f(\boldsymbol{w}\mid \boldsymbol{x})\big|^2
  f(\boldsymbol{w}\mid \boldsymbol{x})\,d\boldsymbol{w}\,dP_{\boldsymbol{X}}(\boldsymbol{x}).
$$
Since $f(\boldsymbol{w}\mid \boldsymbol{x})\ge f_{\min}$ for all $(\boldsymbol{x},\boldsymbol{w})$, any non-negative integrand $h(\boldsymbol{x},\boldsymbol{w})$ satisfies
\begin{align*}
\int\!\!\int_{\mathcal{W}} h(\boldsymbol{x},\boldsymbol{w})\,d\boldsymbol{w}\,dP_{\boldsymbol{X}}(\boldsymbol{x})
&= \int\!\!\int_{\mathcal{W}} \frac{h(\boldsymbol{x},\boldsymbol{w})}{f(\boldsymbol{w}\mid \boldsymbol{x})}\,
   f(\boldsymbol{w}\mid \boldsymbol{x})\,d\boldsymbol{w}\,dP_{\boldsymbol{X}}(\boldsymbol{x})\\
&\le f_{\min}^{-1}\int\!\!\int_{\mathcal{W}} h(\boldsymbol{x},\boldsymbol{w})\,f(\boldsymbol{w}\mid \boldsymbol{x})\,d\boldsymbol{w}\,dP_{\boldsymbol{X}}(\boldsymbol{x}).
\end{align*}
Applying this with $h(\boldsymbol{x},\boldsymbol{w})=\big|\widehat{f}(\boldsymbol{w}\mid \boldsymbol{x})-f(\boldsymbol{w}\mid \boldsymbol{x})\big|^2$ yields
$$
\int\!\!\int_{\mathcal{W}}\big|\widehat{f}(\boldsymbol{w}\mid \boldsymbol{x})-f(\boldsymbol{w}\mid \boldsymbol{x})\big|^2\,d\boldsymbol{w}\,dP_{\boldsymbol{X}}(\boldsymbol{x})
\le f_{\min}^{-1}\,\|\widehat{f}-f\|_2^2.
$$
Substituting this bound and using, for any function depending only on $\boldsymbol{X}$,
$$
\|\widehat{\nu}_{\boldsymbol{\delta}}-\nu_{\boldsymbol{\delta}}\|_2^2
= \int\{\widehat{\nu}_{\boldsymbol{\delta}}(\boldsymbol{x})-\nu_{\boldsymbol{\delta}}(\boldsymbol{x})\}^2\,dP_{\boldsymbol{X}}(\boldsymbol{x}),
$$
gives
$$
\|\widehat{\nu}_{\boldsymbol{\delta}}-\nu_{\boldsymbol{\delta}}\|_2^2
\le e^{2\tau_\Delta}|\mathcal{W}|\,f_{\min}^{-1}\,\|\widehat{f}-f\|_2^2,
$$
and hence
$$
\|\widehat{\nu}_{\boldsymbol{\delta}}-\nu_{\boldsymbol{\delta}}\|_2
\le e^{\tau_\Delta}|\mathcal{W}|^{1/2}f_{\min}^{-1/2}\,\|\widehat{f}-f\|_2.
$$
Moreover, since $\nu_{\boldsymbol{\delta}}(\boldsymbol{x})\in[e^{-\tau_\Delta},e^{\tau_\Delta}]\subset[e^{-\tau_\Delta}/2,2e^{\tau_\Delta}]$ for all $\boldsymbol{x}$ and the truncation map is $1$--Lipschitz,
$$
\|\widehat{\nu}_{\boldsymbol{\delta}}^{\dagger}-\nu_{\boldsymbol{\delta}}\|_2
\le \|\widehat{\nu}_{\boldsymbol{\delta}}-\nu_{\boldsymbol{\delta}}\|_2.
$$
\medskip\noindent
\textit{Control of $\widehat{r}_{\boldsymbol{\delta}} - r_{\boldsymbol{\delta}}$.}
The identity
\begin{align*}
\widehat{r}_{\boldsymbol{\delta}}(\boldsymbol{w}, \boldsymbol{x}) - r_{\boldsymbol{\delta}}(\boldsymbol{w}, \boldsymbol{x})
&= \exp\{\boldsymbol{\delta}^\top \boldsymbol{w}\}\Big\{\frac{1}{\widehat{\nu}_{\boldsymbol{\delta}}^{\dagger}(\boldsymbol{x})} - \frac{1}{\nu_{\boldsymbol{\delta}}(\boldsymbol{x})}\Big\}\\
&= \exp\{\boldsymbol{\delta}^\top \boldsymbol{w}\}\frac{\nu_{\boldsymbol{\delta}}(\boldsymbol{x}) - \widehat{\nu}_{\boldsymbol{\delta}}^{\dagger}(\boldsymbol{x})}{\widehat{\nu}_{\boldsymbol{\delta}}^{\dagger}(\boldsymbol{x})\nu_{\boldsymbol{\delta}}(\boldsymbol{x})}.
\end{align*}
Using the finite-$\Delta$ bounds on the numerator and denominators, there exists a constant $C_r(\Delta)$ such that
$$
|\widehat{r}_{\boldsymbol{\delta}} - r_{\boldsymbol{\delta}}|
\le C_r(\Delta)\,|\widehat{\nu}_{\boldsymbol{\delta}}^{\dagger} - \nu_{\boldsymbol{\delta}}|,
$$
so that
\begin{align*}
\|\widehat{r}_{\boldsymbol{\delta}} - r_{\boldsymbol{\delta}}\|_2
&\le C_r(\Delta)\,\|\widehat{\nu}_{\boldsymbol{\delta}}^{\dagger} - \nu_{\boldsymbol{\delta}}\|_2
\le C_r(\Delta)\,\|\widehat{\nu}_{\boldsymbol{\delta}} - \nu_{\boldsymbol{\delta}}\|_2\\
&\le C_1(\Delta)\,\|\widehat{f} - f\|_2,
\end{align*}
with $C_1(\Delta) := C_r(\Delta)e^{\tau_\Delta}|\mathcal{W}|^{1/2}f_{\min}^{-1/2}$, which proves \eqref{eq:r-lipschitz}.

\medskip\noindent
\textit{Control of $\widehat{\eta}_{\boldsymbol{\delta}} - \eta_{\boldsymbol{\delta}}$.}
The product expansion $\widehat{\mu}\widehat{f} - \mu f = (\widehat{\mu} - \mu)\widehat{f} + \mu(\widehat{f} - f)$ gives
\begin{align*}
\widehat{\eta}_{\boldsymbol{\delta}}(\boldsymbol{x}) - \eta_{\boldsymbol{\delta}}(\boldsymbol{x})
&= \int_{\mathcal{W}} \exp\{\boldsymbol{\delta}^\top \boldsymbol{w}\}
\big[(\widehat{\mu} - \mu)(\boldsymbol{x}, \boldsymbol{w})\,\widehat{f}(\boldsymbol{w} \mid \boldsymbol{x})
+ \mu(\boldsymbol{x}, \boldsymbol{w})\big\{\widehat{f}(\boldsymbol{w} \mid \boldsymbol{x}) - f(\boldsymbol{w} \mid \boldsymbol{x})\big\}\big]\,d\boldsymbol{w}.
\end{align*}
Assumption \textnormal{(C1)} implies $|\mu(\boldsymbol{x}, \boldsymbol{w})| \le M < \infty$,
and, with probability tending to one, $0\le \widehat f\le \widehat M_f$ for some
finite constant $\widehat M_f$. Hence
\begin{align*}
|\widehat{\eta}_{\boldsymbol{\delta}}(\boldsymbol{x}) - \eta_{\boldsymbol{\delta}}(\boldsymbol{x})|
&\le e^{\tau_\Delta}\Big[\int_{\mathcal{W}} |(\widehat{\mu}-\mu)(\boldsymbol{x},\boldsymbol{w})|\,\widehat{f}(\boldsymbol{w}\mid \boldsymbol{x})\,d\boldsymbol{w}
+ M\int_{\mathcal{W}}|\widehat{f}(\boldsymbol{w}\mid \boldsymbol{x})-f(\boldsymbol{w}\mid \boldsymbol{x})|\,d\boldsymbol{w}\Big].
\end{align*}
Applying Cauchy--Schwarz yields
$$
\|\widehat{\eta}_{\boldsymbol{\delta}} - \eta_{\boldsymbol{\delta}}\|_2
\le C_\eta(\Delta)\big(\|\widehat{\mu} - \mu\|_2 + \|\widehat{f} - f\|_2\big)
$$
for some finite constant $C_\eta(\Delta)$.

\medskip\noindent
\textit{Control of $\widehat{m}_{\boldsymbol{\delta}} - m_{\boldsymbol{\delta}}$.}
By the algebraic identity $a/b - c/d = (a - c)/b + c(1/b - 1/d)$,
\begin{align*}
\widehat{m}_{\boldsymbol{\delta}}(\boldsymbol{x}) - m_{\boldsymbol{\delta}}(\boldsymbol{x})
&= \frac{\widehat{\eta}_{\boldsymbol{\delta}}(\boldsymbol{x}) - \eta_{\boldsymbol{\delta}}(\boldsymbol{x})}{\widehat{\nu}_{\boldsymbol{\delta}}^{\dagger}(\boldsymbol{x})}
+ \eta_{\boldsymbol{\delta}}(\boldsymbol{x})\Big\{\frac{1}{\widehat{\nu}_{\boldsymbol{\delta}}^{\dagger}(\boldsymbol{x})} - \frac{1}{\nu_{\boldsymbol{\delta}}(\boldsymbol{x})}\Big\}.
\end{align*}
Using the uniform bounds on $\eta_{\boldsymbol{\delta}}, \nu_{\boldsymbol{\delta}}, \widehat{\nu}_{\boldsymbol{\delta}}^{\dagger}$ implied by \textnormal{(C1)}, \textnormal{(C3)}, and $\|\boldsymbol{\delta}\|\le\Delta$ gives
$$
|\widehat{m}_{\boldsymbol{\delta}} - m_{\boldsymbol{\delta}}|
\le C_m(\Delta)\Big(|\widehat{\eta}_{\boldsymbol{\delta}} - \eta_{\boldsymbol{\delta}}|
+ |\widehat{\nu}_{\boldsymbol{\delta}}^{\dagger} - \nu_{\boldsymbol{\delta}}|\Big)
$$
for some finite constant $C_m(\Delta)$. Taking $L_2(P)$-norms and combining the preceding bounds gives \eqref{eq:m-lipschitz} with
$$
C_2(\Delta) := C_m(\Delta)\big(C_\eta(\Delta) + e^{\tau_\Delta}|\mathcal{W}|^{1/2}\big).
$$
\end{proof}

\begin{corollary}[Rates for tilted nuisance functions]
\label{cor:rates-mu-pi}
Suppose the conditions of Lemma~\ref{lem:reduction-mu-pi} hold. If
$$
\|\widehat{\mu} - \mu\|_2 = o_P(n^{-1/4})
\qquad\text{and}\qquad
\|\widehat{f} - f\|_2 = o_P(n^{-1/4}),
$$
then, for every fixed $\boldsymbol{\delta}$ with $\|\boldsymbol{\delta}\| \le \Delta$,
$$
\|\widehat{r}_{\boldsymbol{\delta}} - r_{\boldsymbol{\delta}}\|_2 = o_P(n^{-1/4}),
\qquad
\|\widehat{m}_{\boldsymbol{\delta}} - m_{\boldsymbol{\delta}}\|_2 = o_P(n^{-1/4}),
$$
and hence the product condition
$\|\widehat{r}_{\boldsymbol{\delta}} - r_{\boldsymbol{\delta}}\|_2 \, \|\widehat{m}_{\boldsymbol{\delta}} - m_{\boldsymbol{\delta}}\|_2 = o_P(n^{-1/2})$
required in Theorem~\ref{thm:finite-delta-CLT} holds.
\end{corollary}

\medskip

The finite-$\boldsymbol{\delta}$ central limit theorem (Theorem~3) rests on a von Mises expansion and a second-order remainder bound for the tilted nuisances $(m_{\boldsymbol{\delta}}, r_{\boldsymbol{\delta}})$. Together with Corollary~\ref{cor:rates-mu-pi}, these results imply Theorem~3.

The estimator in \eqref{eq:psi-hat-finite} is implemented using $K$-fold cross-fitting. The data are randomly partitioned into $K$ disjoint folds of approximately equal size. For each fold $k\in\{1,\ldots,K\}$, nuisance functions are estimated using the remaining $K-1$ folds and evaluated only on observations in fold $k$. Thus, for observation $i$ in fold $k(i)$, the quantities $\widehat r_{\boldsymbol{\delta}}(\boldsymbol W_i,\boldsymbol X_i)$ and $\widehat m_{\boldsymbol{\delta}}(\boldsymbol X_i)$ in \eqref{eq:psi-hat-finite} are computed using nuisance estimates trained without observation $i$; the fold index is suppressed in the notation. For procedures that first estimate $\mu$ and $f$, those nuisance functions are trained on the corresponding training folds before constructing $\widehat r_{\boldsymbol{\delta}}$ and $\widehat m_{\boldsymbol{\delta}}$. Aggregating the fold-specific EIF contributions gives the final estimator and reduces overfitting bias from flexible nuisance estimation \citep{zheng2011cross, chernozhukov2018double}.

The cross-fitted one-step estimator has the density-ratio representation
\begin{equation}\label{eq:psi-hat-finite}
\widehat{\psi}(\boldsymbol{\delta})
= P_n\!\big[\widehat{r}_{\boldsymbol{\delta}}(\boldsymbol{W}, \boldsymbol{X})\{Y - \widehat{m}_{\boldsymbol{\delta}}(\boldsymbol{X})\}\big]
\;+\; P_n\!\big[\widehat{m}_{\boldsymbol{\delta}}(\boldsymbol{X})\big],
\qquad
\widehat{\theta}(\boldsymbol{\delta}) := \widehat{\psi}(\boldsymbol{\delta}) - \widehat{\psi}(\boldsymbol{0}),
\end{equation}
where, for a fixed tilt $\boldsymbol{\delta}$ with $\|\boldsymbol{\delta}\| \le \Delta$,
\begin{align*}
m_{\boldsymbol{\delta}}(\boldsymbol{x})
&:= \mathbb{E}_{g_{\boldsymbol{\delta}}}\!\left[\mu(\boldsymbol{X}, \boldsymbol{W})\mid \boldsymbol{X}=\boldsymbol{x}\right]\\
&= \frac{\mathbb{E}\!\left[\exp\{\boldsymbol{\delta}^\top \boldsymbol{W}\}\mu(\boldsymbol{X}, \boldsymbol{W})\mid \boldsymbol{X}=\boldsymbol{x}\right]}
       {\mathbb{E}\!\left[\exp\{\boldsymbol{\delta}^\top \boldsymbol{W}\}\mid \boldsymbol{X}=\boldsymbol{x}\right]}\\
&= \frac{\eta_{\boldsymbol{\delta}}(\boldsymbol{x})}{\nu_{\boldsymbol{\delta}}(\boldsymbol{x})},
\end{align*}
and
\begin{align*}
\nu_{\boldsymbol{\delta}}(\boldsymbol{x})
&:= \mathbb{E}\!\left[\exp\{\boldsymbol{\delta}^\top \boldsymbol{W}\}\mid \boldsymbol{X}=\boldsymbol{x}\right],\\
\eta_{\boldsymbol{\delta}}(\boldsymbol{x})
&:= \mathbb{E}\!\left[\exp\{\boldsymbol{\delta}^\top \boldsymbol{W}\}\mu(\boldsymbol{X}, \boldsymbol{W})\mid \boldsymbol{X}=\boldsymbol{x}\right],
\end{align*}
\begin{align*}
\widehat{\nu}_{\boldsymbol{\delta}}^{\dagger}(\boldsymbol{x})
&:=\min\Big\{\max\big(\widehat{\nu}_{\boldsymbol{\delta}}(\boldsymbol{x}),\ e^{-\tau_\Delta}/2\big),\ 2e^{\tau_\Delta}\Big\},\\
r_{\boldsymbol{\delta}}(\boldsymbol{W}, \boldsymbol{X})
&:= \frac{\exp\{\boldsymbol{\delta}^\top \boldsymbol{W}\}}{\nu_{\boldsymbol{\delta}}(\boldsymbol{X})},\\
\widehat{r}_{\boldsymbol{\delta}}(\boldsymbol{W}, \boldsymbol{X})
&:= \frac{\exp\{\boldsymbol{\delta}^\top \boldsymbol{W}\}}{\widehat{\nu}_{\boldsymbol{\delta}}^{\dagger}(\boldsymbol{X})},\\
\widehat{m}_{\boldsymbol{\delta}}(\boldsymbol{X})
&:= \frac{\widehat{\eta}_{\boldsymbol{\delta}}(\boldsymbol{X})}{\widehat{\nu}_{\boldsymbol{\delta}}^{\dagger}(\boldsymbol{X})}.
\end{align*}
The estimators $\widehat{\nu}_{\boldsymbol{\delta}}$ and $\widehat{\eta}_{\boldsymbol{\delta}}$ are obtained from cross-fitted regressions of the transformed outcomes $\exp\{\boldsymbol{\delta}^\top \boldsymbol{W}\}$ and $\exp\{\boldsymbol{\delta}^\top \boldsymbol{W}\}\mu(\boldsymbol{X}, \boldsymbol{W})$ on $\boldsymbol{X}$, respectively.

The von Mises decomposition (see, e.g., \citep{kennedy2024semiparametric}) gives
\begin{equation}\label{eq:VM-finite}
\widehat{\psi}(\boldsymbol{\delta}) - \psi(\boldsymbol{\delta})
= (P_n - P)\{\varphi_{\psi(\boldsymbol{\delta})}(\boldsymbol{Z})\}
+ (P_n - P)\!\big\{\widehat{\varphi}_{\psi(\boldsymbol{\delta})}(\boldsymbol{Z}) - \varphi_{\psi(\boldsymbol{\delta})}(\boldsymbol{Z})\big\}
+ R_2(\widehat{P}, P; \boldsymbol{\delta}),
\end{equation}
where
\begin{align*}
\varphi_{\psi(\boldsymbol{\delta})}(\boldsymbol{Z})
&:= r_{\boldsymbol{\delta}}(\boldsymbol{W}, \boldsymbol{X})\{Y - m_{\boldsymbol{\delta}}(\boldsymbol{X})\}
+ m_{\boldsymbol{\delta}}(\boldsymbol{X}) - \psi(\boldsymbol{\delta}),\\
\widehat{\varphi}_{\psi(\boldsymbol{\delta})}(\boldsymbol{Z})
&:= \widehat{r}_{\boldsymbol{\delta}}(\boldsymbol{W}, \boldsymbol{X})\{Y - \widehat{m}_{\boldsymbol{\delta}}(\boldsymbol{X})\}
+ \widehat{m}_{\boldsymbol{\delta}}(\boldsymbol{X}) - \widehat{\psi}(\boldsymbol{\delta}),
\end{align*}
and the second-order remainder is
$$
R_2(\widehat{P}, P; \boldsymbol{\delta})
:= \widehat{\psi}(\boldsymbol{\delta}) - \psi(\boldsymbol{\delta}) + P\!\left[\widehat{\varphi}_{\psi(\boldsymbol{\delta})}\right].
$$
\begin{lemma}[Second-order remainder]\label{lem:R2-finite}
For any fixed $\boldsymbol{\delta}$ with $\|\boldsymbol{\delta}\| \le \Delta$,
$$
\big|R_2(\widehat{P}, P; \boldsymbol{\delta})\big|
\ \le\ \|\widehat{r}_{\boldsymbol{\delta}} - r_{\boldsymbol{\delta}}\|_2\,\|\widehat{m}_{\boldsymbol{\delta}} - m_{\boldsymbol{\delta}}\|_2.
$$
\end{lemma}

\begin{proof}
The estimator representation \eqref{eq:psi-hat-finite} and the definition of $R_2$ give
\begin{align*}
R_2
&= \mathbb{E}\big[\widehat{r}_{\boldsymbol{\delta}}(\boldsymbol{W}, \boldsymbol{X})\{Y - \widehat{m}_{\boldsymbol{\delta}}(\boldsymbol{X})\} + \widehat{m}_{\boldsymbol{\delta}}(\boldsymbol{X})\big]\\
&\quad - \mathbb{E}\big[r_{\boldsymbol{\delta}}(\boldsymbol{W}, \boldsymbol{X})\{Y - m_{\boldsymbol{\delta}}(\boldsymbol{X})\} + m_{\boldsymbol{\delta}}(\boldsymbol{X})\big]\\
&= \mathbb{E}\big[(\widehat{r}_{\boldsymbol{\delta}} - r_{\boldsymbol{\delta}})\{Y - m_{\boldsymbol{\delta}}(\boldsymbol{X})\}\big]
   - \mathbb{E}\big[(\widehat{r}_{\boldsymbol{\delta}} - r_{\boldsymbol{\delta}})\{\widehat{m}_{\boldsymbol{\delta}}(\boldsymbol{X}) - m_{\boldsymbol{\delta}}(\boldsymbol{X})\}\big].
\end{align*}
The first term vanishes by iterated expectation:
$$
\mathbb{E}\!\big[(\widehat{r}_{\boldsymbol{\delta}} - r_{\boldsymbol{\delta}})\{Y - \mu(\boldsymbol{X}, \boldsymbol{W})\}\big]
= \mathbb{E}\!\Big[\,\mathbb{E}\big[(\widehat{r}_{\boldsymbol{\delta}} - r_{\boldsymbol{\delta}})\{Y - \mu(\boldsymbol{X}, \boldsymbol{W})\}\mid \boldsymbol{X}, \boldsymbol{W}\big]\Big]
= 0,
$$
and
\begin{align*}
\mathbb{E}\big[(\widehat{r}_{\boldsymbol{\delta}} - r_{\boldsymbol{\delta}})\{\mu(\boldsymbol{X}, \boldsymbol{W}) - m_{\boldsymbol{\delta}}(\boldsymbol{X})\}\mid \boldsymbol{X}\big]
&= \mathbb{E}[\widehat{r}_{\boldsymbol{\delta}}\mu \mid \boldsymbol{X}] - m_{\boldsymbol{\delta}}(\boldsymbol{X})\mathbb{E}[\widehat{r}_{\boldsymbol{\delta}} \mid \boldsymbol{X}]\\
&\quad - \mathbb{E}[r_{\boldsymbol{\delta}}\mu \mid \boldsymbol{X}] + m_{\boldsymbol{\delta}}(\boldsymbol{X})\mathbb{E}[r_{\boldsymbol{\delta}} \mid \boldsymbol{X}] \\
&= \frac{\eta_{\boldsymbol{\delta}}(\boldsymbol{X})}{\widehat{\nu}_{\boldsymbol{\delta}}^{\dagger}(\boldsymbol{X})}
   - \frac{\eta_{\boldsymbol{\delta}}(\boldsymbol{X})}{\nu_{\boldsymbol{\delta}}(\boldsymbol{X})}\\
&\quad - \frac{\eta_{\boldsymbol{\delta}}(\boldsymbol{X})}{\nu_{\boldsymbol{\delta}}(\boldsymbol{X})}
     \Big(\frac{\nu_{\boldsymbol{\delta}}(\boldsymbol{X})}{\widehat{\nu}_{\boldsymbol{\delta}}^{\dagger}(\boldsymbol{X})} - 1\Big) \\
&= 0.
\end{align*}
Hence $R_2 = -\mathbb{E}[(\widehat{r}_{\boldsymbol{\delta}} - r_{\boldsymbol{\delta}})(\widehat{m}_{\boldsymbol{\delta}} - m_{\boldsymbol{\delta}})]$, and the Cauchy--Schwarz inequality yields
$$
\big|R_2(\widehat{P}, P; \boldsymbol{\delta})\big|
\le \|\widehat{r}_{\boldsymbol{\delta}} - r_{\boldsymbol{\delta}}\|_2\,\|\widehat{m}_{\boldsymbol{\delta}} - m_{\boldsymbol{\delta}}\|_2.
$$
\end{proof}

\begin{theorem}[Asymptotic normality]\label{thm:finite-delta-CLT}
Suppose \textnormal{(C1)} and \textnormal{(C3)} hold. Let $\boldsymbol{Z} := (\boldsymbol{X}, \boldsymbol{W}, Y)$ and suppose $\{\boldsymbol{Z}_i\}_{i=1}^n$ are i.i.d.\ draws from $P$. Fix $\boldsymbol{\delta}$ with $\|\boldsymbol{\delta}\| \le \Delta < \infty$. Let $\widehat{\psi}(\boldsymbol{\delta})$ be the estimator in \eqref{eq:psi-hat-finite}, with nuisance estimators $(\widehat{r}_{\boldsymbol{\delta}}, \widehat{m}_{\boldsymbol{\delta}})$ satisfying
$$
\|\widehat{r}_{\boldsymbol{\delta}} - r_{\boldsymbol{\delta}}\|_2 = o_P(1),\qquad
\|\widehat{m}_{\boldsymbol{\delta}} - m_{\boldsymbol{\delta}}\|_2 = o_P(1),
$$
and the product condition
$$
\|\widehat{r}_{\boldsymbol{\delta}} - r_{\boldsymbol{\delta}}\|_2\,\|\widehat{m}_{\boldsymbol{\delta}} - m_{\boldsymbol{\delta}}\|_2 = o_P(n^{-1/2}).
$$
A sufficient rate condition is
$$
\|\widehat{r}_{\boldsymbol{\delta}} - r_{\boldsymbol{\delta}}\|_2 = o_P(n^{-1/4}),\qquad
\|\widehat{m}_{\boldsymbol{\delta}} - m_{\boldsymbol{\delta}}\|_2 = o_P(n^{-1/4}).
$$
Then
$$
\sqrt{n}\,\big\{\widehat{\psi}(\boldsymbol{\delta}) - \psi(\boldsymbol{\delta})\big\}
\ \rightsquigarrow\ \mathcal{N}\!\big(0,\ \mathrm{Var}\{\varphi_{\psi(\boldsymbol{\delta})}(\boldsymbol{Z})\}\big).
$$
Consequently, for the incremental effect $\widehat{\theta}(\boldsymbol{\delta}) = \widehat{\psi}(\boldsymbol{\delta}) - \widehat{\psi}(\boldsymbol{0})$,
$$
\sqrt{n}\,\big\{\widehat{\theta}(\boldsymbol{\delta}) - \theta(\boldsymbol{\delta})\big\}
\ \rightsquigarrow\ \mathcal{N}\!\big(0,\ \mathrm{Var}\{\varphi_{\theta(\boldsymbol{\delta})}(\boldsymbol{Z})\}\big),
\qquad
\varphi_{\theta(\boldsymbol{\delta})} := \varphi_{\psi(\boldsymbol{\delta})} - \varphi_{\psi(\boldsymbol{0})}.
$$
\end{theorem}

\begin{proof}
\textit{Control of the leading empirical process term.}
For any $\|\boldsymbol{\delta}\| \le \Delta$, boundedness of $\boldsymbol{W}$ implies there exists $\tau_\Delta < \infty$ such that
$$
e^{-\tau_\Delta} \le \exp\{\boldsymbol{\delta}^\top \boldsymbol{W}\} \le e^{\tau_\Delta},\qquad
e^{-\tau_\Delta} \le \nu_{\boldsymbol{\delta}}(\boldsymbol{X}) \le e^{\tau_\Delta},
$$
so that $e^{-2\tau_\Delta} \le r_{\boldsymbol{\delta}}(\boldsymbol{W}, \boldsymbol{X}) \le e^{2\tau_\Delta}$. Under \textnormal{(C1)}, $|Y| \le M$ almost surely, and hence $|m_{\boldsymbol{\delta}}(\boldsymbol{X})| \le M$ as well. Therefore
\begin{align*}
|\varphi_{\psi(\boldsymbol{\delta})}(\boldsymbol{Z})|
&\le e^{2\tau_\Delta}\,|Y - m_{\boldsymbol{\delta}}(\boldsymbol{X})|
   + |m_{\boldsymbol{\delta}}(\boldsymbol{X}) - \psi(\boldsymbol{\delta})|\\
&\le 2Me^{2\tau_\Delta} + 2M < \infty.
\end{align*}
Since $\{\boldsymbol{Z}_i\}_{i=1}^n$ are i.i.d., the classical Lindeberg--Feller CLT applies and yields
$$
\frac{1}{\sqrt n}\sum_{i=1}^n\varphi_{\psi(\boldsymbol{\delta})}(\boldsymbol{Z}_i)
\ \rightsquigarrow\ \mathcal{N}\!\big(0,\ \mathrm{Var}\{\varphi_{\psi(\boldsymbol{\delta})}(\boldsymbol{Z})\}\big).
$$
\medskip\noindent
\textit{Control of the estimated influence term.}
From the definitions,
\begin{align*}
\widehat{\varphi}_{\psi(\boldsymbol{\delta})}(\boldsymbol{Z})
&= \widehat{r}_{\boldsymbol{\delta}}\{Y - \widehat{m}_{\boldsymbol{\delta}}\}
+ \widehat{m}_{\boldsymbol{\delta}} - \widehat{\psi}(\boldsymbol{\delta}),\\
\varphi_{\psi(\boldsymbol{\delta})}(\boldsymbol{Z})
&= r_{\boldsymbol{\delta}}\{Y - m_{\boldsymbol{\delta}}\} + m_{\boldsymbol{\delta}} - \psi(\boldsymbol{\delta}),
\end{align*}
so that
\begin{align*}
\widehat{\varphi}_{\psi(\boldsymbol{\delta})} - \varphi_{\psi(\boldsymbol{\delta})}
&= (\widehat{r}_{\boldsymbol{\delta}} - r_{\boldsymbol{\delta}})\{Y - m_{\boldsymbol{\delta}}\}
+ \big(1 - \widehat{r}_{\boldsymbol{\delta}}\big)\{\widehat{m}_{\boldsymbol{\delta}} - m_{\boldsymbol{\delta}}\}
- \{\widehat{\psi}(\boldsymbol{\delta}) - \psi(\boldsymbol{\delta})\}.
\end{align*}
Applying $P_n - P$ to both sides and using $(P_n - P)\{\widehat{\psi}(\boldsymbol{\delta}) - \psi(\boldsymbol{\delta})\} = 0$ gives
\begin{align*}
(P_n - P)\big\{\widehat{\varphi}_{\psi(\boldsymbol{\delta})} - \varphi_{\psi(\boldsymbol{\delta})}\big\}
&= (P_n - P)\!\left[(\widehat{r}_{\boldsymbol{\delta}} - r_{\boldsymbol{\delta}})\{Y - m_{\boldsymbol{\delta}}\}
+ \big(1 - \widehat{r}_{\boldsymbol{\delta}}\big)\{\widehat{m}_{\boldsymbol{\delta}} - m_{\boldsymbol{\delta}}\}\right].
\end{align*}
The bound $|Y| \le M$, the finite-$\Delta$ bounds on $r_{\boldsymbol{\delta}}$ and $\widehat{r}_{\boldsymbol{\delta}}$, and the Cauchy--Schwarz inequality imply
\begin{align*}
(P_n - P)\big\{\widehat{\varphi}_{\psi(\boldsymbol{\delta})} - \varphi_{\psi(\boldsymbol{\delta})}\big\}
&= O_P\!\left(\frac{\|\widehat{r}_{\boldsymbol{\delta}} - r_{\boldsymbol{\delta}}\|_2
+ \|\widehat{m}_{\boldsymbol{\delta}} - m_{\boldsymbol{\delta}}\|_2}{\sqrt n}\right)
= o_P(1),
\end{align*}
under the assumed $L_2$ rates. In particular,
\begin{align*}
\sqrt{n}\,(P_n - P)\big\{\widehat{\varphi}_{\psi(\boldsymbol{\delta})} - \varphi_{\psi(\boldsymbol{\delta})}\big\}
&= O_P\!\Big(\|\widehat{r}_{\boldsymbol{\delta}} - r_{\boldsymbol{\delta}}\|_2+\|\widehat{m}_{\boldsymbol{\delta}} - m_{\boldsymbol{\delta}}\|_2\Big)
= o_P(1).
\end{align*}

\medskip\noindent
\textit{Control of the remainder.}
By Lemma~\ref{lem:R2-finite},
$$
\big|R_2(\widehat{P}, P; \boldsymbol{\delta})\big|
\le \|\widehat{r}_{\boldsymbol{\delta}} - r_{\boldsymbol{\delta}}\|_2\,\|\widehat{m}_{\boldsymbol{\delta}} - m_{\boldsymbol{\delta}}\|_2,
$$
so the product condition yields $\sqrt n\,R_2(\widehat{P}, P; \boldsymbol{\delta}) = o_P(1)$.

\medskip\noindent
\textit{Asymptotic conclusion.}
Combining \eqref{eq:VM-finite} with the empirical-process and remainder bounds, then applying Slutsky's theorem, yields
$$
\sqrt n\{\widehat{\psi}(\boldsymbol{\delta}) - \psi(\boldsymbol{\delta})\}
\ \rightsquigarrow\ \mathcal{N}\!\big(0,\ \mathrm{Var}\{\varphi_{\psi(\boldsymbol{\delta})}(\boldsymbol{Z})\}\big).
$$
Moreover, applying the same argument at $\boldsymbol{\delta}=\boldsymbol{0}$ and using the multivariate Lindeberg--Feller CLT gives the joint convergence
$$
\sqrt{n}\Big(
\widehat{\psi}(\boldsymbol{\delta})-\psi(\boldsymbol{\delta}),
\widehat{\psi}(\boldsymbol{0})-\psi(\boldsymbol{0})
\Big)
\rightsquigarrow
\mathcal{N}\!\Big(\boldsymbol{0},\ \Cov\big(\varphi_{\psi(\boldsymbol{\delta})}(\boldsymbol{Z}),\ \varphi_{\psi(\boldsymbol{0})}(\boldsymbol{Z})\big)\Big).
$$
Therefore, by the continuous mapping theorem,
\begin{align*}
\sqrt{n}\big\{\widehat{\theta}(\boldsymbol{\delta})-\theta(\boldsymbol{\delta})\big\}
&=\sqrt{n}\Big(\big\{\widehat{\psi}(\boldsymbol{\delta})-\psi(\boldsymbol{\delta})\big\}-\big\{\widehat{\psi}(\boldsymbol{0})-\psi(\boldsymbol{0})\big\}\Big)\\
&\rightsquigarrow
\mathcal{N}\!\big(0,\ \mathrm{Var}\{\varphi_{\theta(\boldsymbol{\delta})}(\boldsymbol{Z})\}\big).
\end{align*}
\end{proof}

Combining Theorem~\ref{thm:finite-delta-CLT} with Corollary~\ref{cor:rates-mu-pi} yields the finite-$\boldsymbol{\delta}$ CLT stated in Theorem~3.

\subsection{Influence-Function Variance Estimation}

Theorem~3 justifies influence function variance estimation. For $\boldsymbol{Z}_i=(\boldsymbol{X}_i,\boldsymbol{W}_i,Y_i)$, define
$$
\widehat{\varphi}_{\psi(\boldsymbol{\delta})}(\boldsymbol{Z}_i)
:= \widehat{r}_{\boldsymbol{\delta}}(\boldsymbol{W}_i, \boldsymbol{X}_i)\{Y_i - \widehat{m}_{\boldsymbol{\delta}}(\boldsymbol{X}_i)\}
+ \widehat{m}_{\boldsymbol{\delta}}(\boldsymbol{X}_i) - \widehat{\psi}(\boldsymbol{\delta}),
$$
where $\widehat{r}_{\boldsymbol{\delta}}$ and $\widehat{m}_{\boldsymbol{\delta}}$ are constructed from the fitted estimators $(\widehat{\mu}, \widehat{f})$ as above. A consistent estimator of the asymptotic variance is the sample second moment
$$
\widehat{\sigma}_{\psi(\boldsymbol{\delta})}^{\,2}
:= \frac{1}{n}\sum_{i=1}^n \widehat{\varphi}_{\psi(\boldsymbol{\delta})}(\boldsymbol{Z}_i)^2,
$$
For the incremental effect, define
$$
\widehat{\varphi}_{\theta(\boldsymbol{\delta})}(\boldsymbol{Z}_i)
:= \widehat{\varphi}_{\psi(\boldsymbol{\delta})}(\boldsymbol{Z}_i) - \widehat{\varphi}_{\psi(\boldsymbol{0})}(\boldsymbol{Z}_i),
\qquad
\widehat{\sigma}_{\theta(\boldsymbol{\delta})}^{\,2}
:= \frac{1}{n}\sum_{i=1}^n \widehat{\varphi}_{\theta(\boldsymbol{\delta})}(\boldsymbol{Z}_i)^2.
$$
The corresponding Wald interval is
$$
\widehat{\theta}(\boldsymbol{\delta})\ \pm\ z_{1-\alpha/2}\,
\frac{\widehat{\sigma}_{\theta(\boldsymbol{\delta})}}{\sqrt{n}},
$$
For simultaneous inference over tilts $\boldsymbol{\delta}_1, \ldots, \boldsymbol{\delta}_J$, the joint covariance estimator is the empirical covariance matrix of
$$
\big(\widehat{\varphi}_{\theta(\boldsymbol{\delta}_1)}(\boldsymbol{Z}_i), \ldots, \widehat{\varphi}_{\theta(\boldsymbol{\delta}_J)}(\boldsymbol{Z}_i)\big)^\top,\qquad i=1,\ldots,n,
$$
which enables simultaneous confidence intervals or Wald tests via a multivariate normal approximation.

\section{Simulation Details}
\label{sec:appendix-simulation-details}

\subsection{Data-Generating Mechanisms}
\label{sec:appendix-simulation-dgp}

All simulations use sample size $n=5,000$, covariate dimension $p=10$, exposure dimension $q=6$, and 500 independent Monte Carlo repetitions. The covariate vector $\boldsymbol{X}=(X_1,\ldots,X_p)^\top$ follows $\mathcal{N}(0,\boldsymbol{\Sigma}_X)$ with autoregressive covariance $\Sigma_{X,jk}=0.5^{|j-k|}$.

\begin{align*}
    \boldsymbol{W} = \boldsymbol{B}^\top \boldsymbol{X} + \boldsymbol{\beta}_0 + \boldsymbol{\mathcal{E}}.
\end{align*}

Here, $\boldsymbol{\beta}_0$ is the intercept vector, and $\boldsymbol{B}\in\mathbb{R}^{p\times q}$ is a sparse coefficient matrix with sparsity level $0.4$ whose nonzero entries are drawn from $\mathcal{N}(0,0.6^2)$. The error term $\boldsymbol{\mathcal{E}}$ determines the shape of the conditional exposure distribution, for which we consider three scenarios.

\begin{enumerate}
    \item Gaussian errors: $\boldsymbol{\mathcal{E}}\sim\mathcal{N}(0,\boldsymbol{\Sigma}_W)$, where $\boldsymbol{\Sigma}_W$ has an AR(1) structure with correlation parameter $\rho=0.6$.
    \item Skewed errors: $\boldsymbol{\mathcal{E}}\sim\mathcal{SN}(0,\boldsymbol{\Sigma}_W,\boldsymbol{\alpha})$, with $\boldsymbol{\alpha}=(4,\ldots,4)^\top$. This scenario introduces substantial asymmetry and non-Gaussianity.
    \item Truncated contaminated normal errors: $\boldsymbol{\mathcal{E}}$ follows $(1-\pi)\mathcal{N}(\boldsymbol{0},\boldsymbol{\Sigma}_W)+\pi\mathcal{N}(\boldsymbol{0},\omega^2\boldsymbol{\Sigma}_W)$ with $\pi=0.2$ and $\omega=1.5$, truncated at $6$ standard deviations of the baseline component to preserve the exponential tilting moments needed for the target estimands. This setting introduces heavier tails relative to the Gaussian design while maintaining the required integrability conditions.
\end{enumerate}

The outcome is generated under two complexity regimes. The linear regime provides a simple baseline for evaluating the effect of exposure-density misspecification, whereas the nonlinear regime adds terms motivated by air pollution mixture studies, including meteorological confounding, synergistic toxicity, and effect modification.
\begin{align*}
        Y = \alpha_0 + \boldsymbol{\alpha}^\top \boldsymbol{X} + \boldsymbol{\beta}^\top \boldsymbol{W} + \varepsilon, \quad \varepsilon \sim \mathcal{N}(0, 1).
\end{align*}
$$
        Y = \alpha_0 + \boldsymbol{\alpha}^\top \boldsymbol{X} + c_1 X_2^2 + \boldsymbol{\beta}^\top \boldsymbol{W} + c_2 W_1 W_2 + c_3 X_1 W_1 + \varepsilon, \quad \varepsilon \sim \mathcal{N}(0, 1).
$$
For the linear outcome, entries of $\boldsymbol{\alpha}$ are drawn from $\mathcal{N}(0.5,1)$ and entries of $\boldsymbol{\beta}$ are drawn from $\mathcal{N}(2,1)$, giving distinct covariate and exposure signals. The nonlinear outcome uses the same baseline coefficient distributions and adds $c_1X_2^2$ with $c_1=0.5$, $c_2W_1W_2$ with $c_2=1.0$, and $c_3X_1W_1$ with $c_3=0.8$. These terms represent a quadratic covariate effect, an interaction between two exposures, and a covariate-exposure interaction, respectively, reflecting the meteorological confounding, mixture interaction, and effect-modification features described in the main simulation design.

\subsection{Estimator Implementation}
\label{sec:appendix-simulation-estimators}

The simulation comparison targets $\psi(\boldsymbol{\delta})$ at the fixed tilt $\boldsymbol{\delta}=(0.3,\ldots,0.3)^\top\in\mathbb{R}^6$ and compares seven estimation procedures. In all procedures, the outcome regression $\mu(\boldsymbol{X},\boldsymbol{W})=\mathbb{E}[Y\mid\boldsymbol{X},\boldsymbol{W}]$ is estimated by XGBoost with five-fold cross-fitting.

The semiparametric exposure procedures use the location-shift representation $W_j=\mu_j(\boldsymbol{X})+\sigma_j\varepsilon_j$. The conditional mean $\mu_j(\boldsymbol{X})$ is estimated by XGBoost, and $\sigma_j$ is computed as the residual standard deviation within each training fold. The standardized residual vector is modeled with a Gaussian copula and one of three marginal residual families: Gaussian, scaled Student's $t$ with degrees of freedom estimated by maximum likelihood, or an empirical distribution obtained by log-spline density estimation. The first six procedures combine the three residual families with either plug-in density ratio estimation based on the fitted exposure density or SoftBART estimation of the conditional normalizing constant $\nu_{\boldsymbol{\delta}}(\boldsymbol{x})$ \citep{linero2018bayesian}. The seventh procedure estimates both tilted nuisance functions, $\nu_{\boldsymbol{\delta}}(\boldsymbol{x})$ and $\eta_{\boldsymbol{\delta}}(\boldsymbol{x})$, using SoftBART with five-fold cross-fitting, without separately estimating the exposure density.

\subsection{Additional Simulation Results}
\label{sec:appendix-simulation-results-table}

Table~\ref{tab:sim_summary} complements the boxplots presented in Section~\ref{subsec:sim_results} with numerical summaries aggregated over the six data-generating designs. We report the average signed bias, average absolute bias, and average root mean-squared error (RMSE).

\begin{table}[htbp]
\centering
\small
\caption{Average performance over the six simulation designs.}
\label{tab:sim_summary}
\begin{tabular}{p{8.6cm}ccc}
\hline
 & Mean bias & Mean absolute bias & Mean RMSE \\
\hline
Fully direct SoftBART regression & 0.17 & 1.09 & 1.66 \\
\hline
Gaussian residual model & 0.85 & 1.00 & 1.24 \\
\hline
Gaussian residual model with SoftBART estimation of $r_{\boldsymbol{\delta}}(\boldsymbol{w}, \boldsymbol{x})$ & 0.20 & 0.58 & 0.76 \\
\hline
$t$ residual model & 0.80 & 0.96 & 1.19 \\
\hline
$t$ residual model with SoftBART estimation of $r_{\boldsymbol{\delta}}(\boldsymbol{w}, \boldsymbol{x})$ & 0.20 & 0.58 & 0.76 \\
\hline
Empirical residual model & 0.29 & 0.60 & 0.78 \\
\hline
Empirical residual model with SoftBART estimation of $r_{\boldsymbol{\delta}}(\boldsymbol{w}, \boldsymbol{x})$ & 0.12 & 0.55 & 0.74 \\
\hline
\end{tabular}
\end{table}

The hybrid empirical-residual pipeline, which combines direct SoftBART estimation of the conditional normalizing constant with Monte Carlo evaluation of the tilted conditional mean, attains the smallest average absolute bias and the smallest average RMSE in Table~\ref{tab:sim_summary}. Its full-Monte-Carlo empirical-residual counterpart is also consistently competitive, indicating that robustness stems largely from utilizing a flexible residual family rather than imposing a tight parametric approximation on the exposure density.

The numerical summaries clarify the role of the denominator regression step. For both Gaussian and $t$ residual families, substituting the conditional normalizing constant with a direct SoftBART estimate reduces the average RMSE from over 1.00 to 0.76. This improvement is particularly visible in the skewed exposure and highly nonlinear outcome regimes. Conversely, the fully direct SoftBART estimator maintains moderate average bias but exhibits substantially larger dispersion than density-based alternatives, suggesting it serves better as a robustness check than a default implementation.

Together with the panelwise patterns in Figure~1, Table~\ref{tab:sim_summary} confirms that finite-sample variation is driven primarily by the estimation quality of the conditional normalizing constant. Once the semiparametric location-scale model is paired with a cross-fitted SoftBART estimate of $\nu_{\boldsymbol{\delta}}(\boldsymbol{x})$, the Gaussian, $t$, and empirical residual specifications perform remarkably similarly, remaining well-centered around the truth across all designs.

\section{Additional Data Analysis Results}
\label{sec:appendix-application-results}

Figure~\ref{fig:application_contour_appendix} provides an assessment of how large the sensitivity parameters would have to become for the upper bound to reach zero for each exposure (or group of exposures) examined. For this reason, we refer to this as the least favorable point, because it is the level of confounding required for the upper bound to reach zero at any Gelbrich distance, not all distances simultaneously. The NO$_3$+SO$_4$+NH$_4$, NO$_3$, and OM contours lie closest to the origin, indicating that comparatively small confounding would suffice to bring the upper bound to zero for at least one value of the Gelbrich constraint. At the other end, SO$_4$, BFGS, and the Efficient path lie farthest from the origin, so they require materially stronger confounding to overturn the estimated negative effects. BC and NH$_4$ fall between these extremes, still showing a modest degree of robustness to confounding.

\begin{figure}[htbp]
\centering
\includegraphics[width=0.82\textwidth]{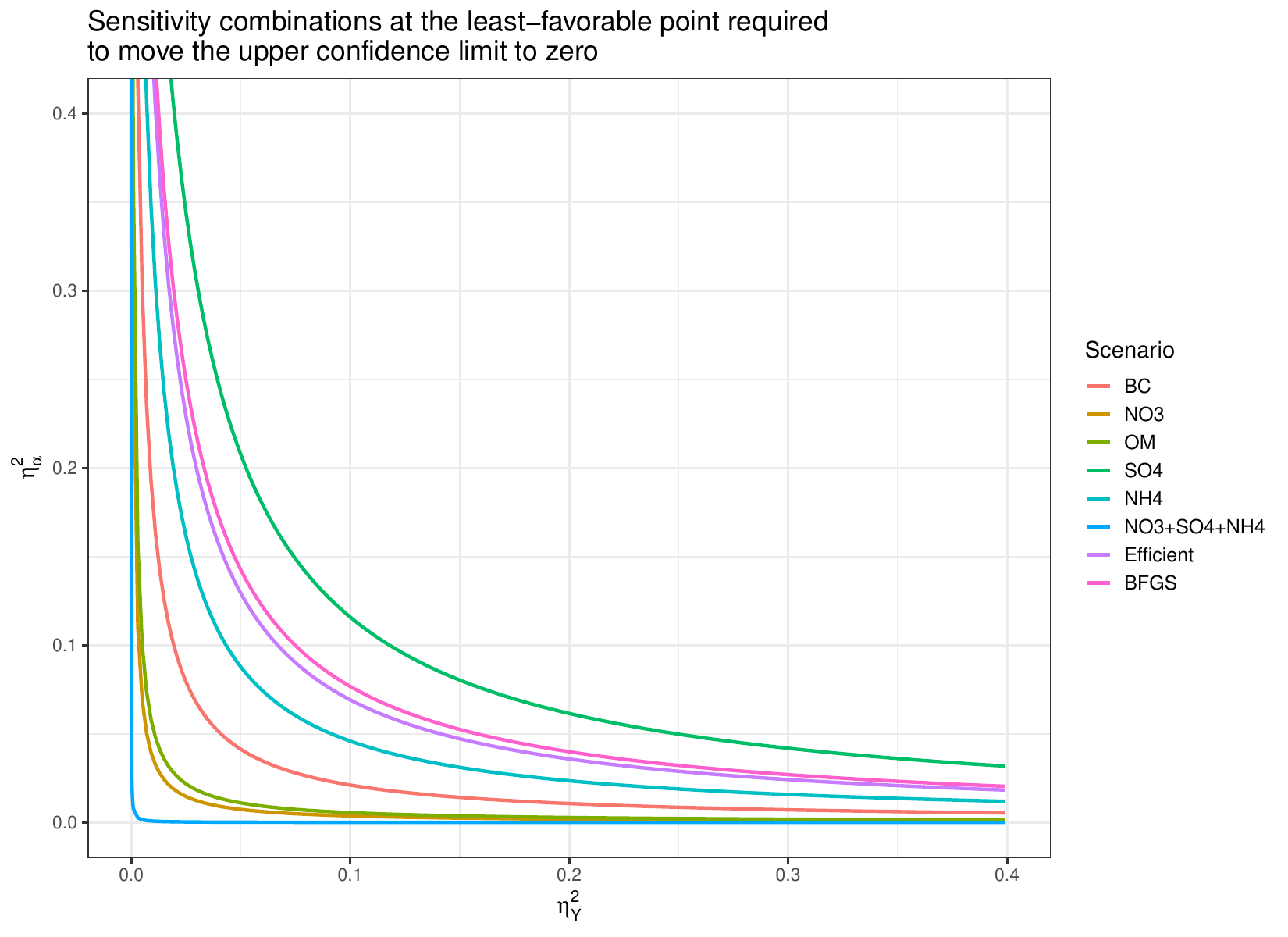}
\caption{Sensitivity contours at the least favorable Gelbrich target for each direction with a negative estimated effect. Each curve gives the combinations of $\eta_Y^2$ and $\eta_\alpha^2(\boldsymbol{\delta})$ that bring the sensitivity-adjusted upper bound to zero; contours farther from the origin indicate less sensitivity to unmeasured confounding under this diagnostic.}
\label{fig:application_contour_appendix}
\end{figure}

\section{Sensitivity analysis for unmeasured confounding}
\label{sec:sensitivity}

This section provides additional information and details about the proposed sensitivity analysis procedure, and builds on the framework used in \citep{chernozhukov2021long}. 

\subsection{Setup: long vs.\ short worlds}

Let the observed data be $\boldsymbol{Z}_i=(\boldsymbol{X}_i,\boldsymbol{W}_i,Y_i)$ drawn i.i.d.\ from $P_0$.
Assume there exists an unobserved confounder $U$ such that, with $\boldsymbol{V}:=(\boldsymbol{X},U)$,
the potential outcomes $\{Y(\boldsymbol{w}):\boldsymbol{w}\in\mathcal W\}$ satisfy
\begin{equation}
Y(\boldsymbol{w})\ \perp\!\!\!\perp\ \boldsymbol{W}\ \big|\ \boldsymbol{V},\qquad \forall \boldsymbol{w}\in\mathcal W.
\label{eq:latent_ignorability}
\end{equation}
Define the long and short outcome regressions
$$
\mu(\boldsymbol{v},\boldsymbol{w}):=\mathbb{E}[Y\mid \boldsymbol{V}=\boldsymbol{v},\boldsymbol{W}=\boldsymbol{w}],\qquad
\mu_s(\boldsymbol{x},\boldsymbol{w}):=\mathbb{E}[Y\mid \boldsymbol{X}=\boldsymbol{x},\boldsymbol{W}=\boldsymbol{w}].
$$
Let $f(\boldsymbol{w}\mid \boldsymbol{v})$ denote the long conditional density of $\boldsymbol{W}$ given $\boldsymbol{V}$,
and $f(\boldsymbol{w}\mid \boldsymbol{x})$ the short conditional density of $\boldsymbol{W}$ given $\boldsymbol{X}$
(i.e.\ the observed conditional law of the exposure mixture given $\boldsymbol{X}$).

\paragraph{Stochastic intervention fixed from observed data.}
We keep the intervention rule identical to Section~\ref{sec:Estimands}: for any fixed
$\boldsymbol\delta\in\mathbb R^q$,
\begin{equation}
g_{\boldsymbol\delta}(\boldsymbol{w}\mid \boldsymbol{x})
=\frac{\exp(\boldsymbol\delta^\top \boldsymbol{w})\,f(\boldsymbol{w}\mid \boldsymbol{x})}
{\int_{\mathcal W}\exp(\boldsymbol\delta^\top \boldsymbol{v})\,f(\boldsymbol{v}\mid \boldsymbol{x})\,d\boldsymbol{v}}.
\label{eq:policy_gdelta_short}
\end{equation}
Importantly, $g_{\boldsymbol\delta}(\cdot\mid \boldsymbol{x})$ depends only on $\boldsymbol{x}$ and is well-defined regardless of whether $U$ exists.

\paragraph{Target (long) estimand.}
Let $Y^{g_{\boldsymbol\delta}}$ denote the counterfactual outcome under the intervention that draws
$\boldsymbol{W}$ from $g_{\boldsymbol\delta}(\cdot\mid \boldsymbol{X})$ (independently of $U$ given $\boldsymbol{X}$).
Under \eqref{eq:latent_ignorability}, the causal estimand is
\begin{equation}
\psi(\boldsymbol\delta)
:=\mathbb{E}\big[Y^{g_{\boldsymbol\delta}}\big]
=\mathbb{E}\Big[\int_{\mathcal W}\mu(\boldsymbol{V},\boldsymbol{w})\,g_{\boldsymbol\delta}(\boldsymbol{w}\mid \boldsymbol{X})\,d\boldsymbol{w}\Big].
\label{eq:psi_long}
\end{equation}

\paragraph{Identified (short) estimand under ignorability given $\boldsymbol{X}$.}
If one incorrectly assumes ignorability given $\boldsymbol{X}$ only, the same formula yields the
identified estimand
\begin{equation}
\psi_s(\boldsymbol\delta)
:=\mathbb{E}\Big[\int_{\mathcal W}\mu_s(\boldsymbol{X},\boldsymbol{w})\,g_{\boldsymbol\delta}(\boldsymbol{w}\mid \boldsymbol{X})\,d\boldsymbol{w}\Big],
\label{eq:psi_short}
\end{equation}
which is the estimand targeted by our estimators. Our sensitivity analysis bounds the discrepancy $\psi_s(\boldsymbol\delta)-\psi(\boldsymbol\delta)$
as a function of interpretable sensitivity parameters.

\subsection{Linear functional form and Riesz representers}

For a fixed $\boldsymbol\delta$, define the functional on square-integrable functions
$h(\boldsymbol{v},\boldsymbol{w})\in L_2(P_{\boldsymbol{V},\boldsymbol{W}})$:
\begin{equation}
\Lambda_{\boldsymbol\delta}(h)
:=\mathbb{E}\Big[\int_{\mathcal W} h(\boldsymbol{V},\boldsymbol{w})\,g_{\boldsymbol\delta}(\boldsymbol{w}\mid \boldsymbol{X})\,d\boldsymbol{w}\Big].
\label{eq:Tdelta_operator}
\end{equation}
Then $\psi(\boldsymbol\delta)=\Lambda_{\boldsymbol\delta}(\mu)$ and
$\psi_s(\boldsymbol\delta)=\Lambda_{\boldsymbol\delta}(\mu_s)$.

\begin{proposition}
    The functional $\Lambda_{\boldsymbol\delta}(h)$ is a linear functional.
\end{proposition}

\begin{proof}
    To establish linearity, we must verify that $\Lambda_{\boldsymbol\delta}$ satisfies both additivity and homogeneity. Let $h_1, h_2 \in L_2(P_{\boldsymbol{V},\boldsymbol{W}})$ be arbitrary functions and let $c \in \mathbb{R}$ be a scalar constant.

    \paragraph{1. Additivity.}
    By the linearity of the Lebesgue integral and the linearity of the expectation operator, we have:
    \begin{align*}
        \Lambda_{\boldsymbol\delta}(h_1 + h_2) 
        &= \mathbb{E}\left[ \int_{\mathcal{W}} (h_1 + h_2)(\boldsymbol{V}, \boldsymbol{w}) \, g_{\boldsymbol{\delta}}(\boldsymbol{w} \mid \boldsymbol{X}) \, d\boldsymbol{w} \right] \\
        &= \mathbb{E}\left[ \int_{\mathcal{W}} \left( h_1(\boldsymbol{V}, \boldsymbol{w}) g_{\boldsymbol{\delta}}(\boldsymbol{w} \mid \boldsymbol{X}) + h_2(\boldsymbol{V}, \boldsymbol{w}) g_{\boldsymbol{\delta}}(\boldsymbol{w} \mid \boldsymbol{X}) \right) \, d\boldsymbol{w} \right] \\
        &= \mathbb{E}\left[ \int_{\mathcal{W}} h_1(\boldsymbol{V}, \boldsymbol{w}) g_{\boldsymbol{\delta}}(\boldsymbol{w} \mid \boldsymbol{X}) \, d\boldsymbol{w} + \int_{\mathcal{W}} h_2(\boldsymbol{V}, \boldsymbol{w}) g_{\boldsymbol{\delta}}(\boldsymbol{w} \mid \boldsymbol{X}) \, d\boldsymbol{w} \right] \\
        &= \mathbb{E}\left[ \int_{\mathcal{W}} h_1(\boldsymbol{V}, \boldsymbol{w}) g_{\boldsymbol{\delta}}(\boldsymbol{w} \mid \boldsymbol{X}) \, d\boldsymbol{w} \right] + \mathbb{E}\left[ \int_{\mathcal{W}} h_2(\boldsymbol{V}, \boldsymbol{w}) g_{\boldsymbol{\delta}}(\boldsymbol{w} \mid \boldsymbol{X}) \, d\boldsymbol{w} \right] \\
        &= \Lambda_{\boldsymbol\delta}(h_1) + \Lambda_{\boldsymbol\delta}(h_2).
    \end{align*}

    \paragraph{2. Homogeneity.}
    \begin{align*}
        \Lambda_{\boldsymbol\delta}(c \cdot h) 
        &= \mathbb{E}\left[ \int_{\mathcal{W}} c \cdot h(\boldsymbol{V}, \boldsymbol{w}) \, g_{\boldsymbol{\delta}}(\boldsymbol{w} \mid \boldsymbol{X}) \, d\boldsymbol{w} \right] \\
        &= \mathbb{E}\left[ c \int_{\mathcal{W}} h(\boldsymbol{V}, \boldsymbol{w}) \, g_{\boldsymbol{\delta}}(\boldsymbol{w} \mid \boldsymbol{X}) \, d\boldsymbol{w} \right] \\
        &= c \, \mathbb{E}\left[ \int_{\mathcal{W}} h(\boldsymbol{V}, \boldsymbol{w}) \, g_{\boldsymbol{\delta}}(\boldsymbol{w} \mid \boldsymbol{X}) \, d\boldsymbol{w} \right] \\
        &= c \cdot \Lambda_{\boldsymbol\delta}(h).
    \end{align*}
    
    Since both conditions are satisfied, $\Lambda_{\boldsymbol\delta}$ is a linear functional.
\end{proof}

\paragraph{Weak overlap / continuity condition.}
Assume $\Lambda_{\boldsymbol\delta}$ is continuous on $L_2(P_{\boldsymbol{V},\boldsymbol{W}})$; a sufficient
condition required here is the ``weak overlap'' requirement
\begin{equation}
\mathbb{E}\!\left[\left(\frac{g_{\boldsymbol\delta}(\boldsymbol{W}\mid \boldsymbol{X})}{f(\boldsymbol{W}\mid \boldsymbol{V})}\right)^2\right]<\infty.
\label{eq:weak_overlap_long}
\end{equation}
This condition is untestable without $U$, but it is the standard integrability requirement needed for
the Riesz representation.

\paragraph{Riesz representation.}
Under \eqref{eq:weak_overlap_long}, by the Riesz--Fr\'echet representation theorem there exists a unique
(long) Riesz representer $\alpha_{\boldsymbol\delta}(\boldsymbol{V},\boldsymbol{W})\in L_2(P_{\boldsymbol{V},\boldsymbol{W}})$ such that
$$
\Lambda_{\boldsymbol\delta}(h) = \mathbb{E}\big[h(\boldsymbol{V},\boldsymbol{W})\,\alpha_{\boldsymbol\delta}(\boldsymbol{V},\boldsymbol{W})\big],
\quad\forall h\in L_2(P_{\boldsymbol{V},\boldsymbol{W}}).
$$
Moreover, the representer has the Radon--Nikodym form
\begin{equation}
\alpha_{\boldsymbol\delta}(\boldsymbol{V},\boldsymbol{W})
=\frac{g_{\boldsymbol\delta}(\boldsymbol{W}\mid \boldsymbol{X})}{f(\boldsymbol{W}\mid \boldsymbol{V})}.
\label{eq:RR_long}
\end{equation}
Likewise, the short Riesz representer for $\Lambda_{\boldsymbol\delta}$ over $L_2(P_{\boldsymbol{X},\boldsymbol{W}})$ is
\begin{equation}
\alpha_{s,\boldsymbol\delta}(\boldsymbol{X},\boldsymbol{W})
=\frac{g_{\boldsymbol\delta}(\boldsymbol{W}\mid \boldsymbol{X})}{f(\boldsymbol{W}\mid \boldsymbol{X})}
=\mathbb{E}\big[\alpha_{\boldsymbol\delta}(\boldsymbol{V},\boldsymbol{W})\mid \boldsymbol{X},\boldsymbol{W}\big].
\label{eq:RR_short}
\end{equation}
Under the exponential tilt \eqref{eq:policy_gdelta_short}, $\alpha_{s,\boldsymbol\delta}$ simplifies to
\begin{equation}
\alpha_{s,\boldsymbol\delta}(\boldsymbol{X},\boldsymbol{W})
=\frac{\exp(\boldsymbol\delta^\top \boldsymbol{W})}
{\nu_{\boldsymbol\delta}(\boldsymbol{X})},\qquad
\nu_{\boldsymbol\delta}(\boldsymbol{x}):=\mathbb{E}\big[\exp(\boldsymbol\delta^\top \boldsymbol{W})\mid \boldsymbol{X}=\boldsymbol{x}\big].
\label{eq:RR_short_closed_form}
\end{equation}
Thus our policy estimand is a continuous linear functional of the long regression $\mu$ with a
well-defined RR.

\subsection{Exact OVB identity and sharp bound}

Define the outcome regression error and the RR error:
$$
\Delta_\mu(\boldsymbol{V},\boldsymbol{W}):=\mu(\boldsymbol{V},\boldsymbol{W})-\mu_s(\boldsymbol{X},\boldsymbol{W}),\qquad
\Delta_\alpha(\boldsymbol{V},\boldsymbol{W}):=\alpha_{\boldsymbol\delta}(\boldsymbol{V},\boldsymbol{W})-\alpha_{s,\boldsymbol\delta}(\boldsymbol{X},\boldsymbol{W}).
$$
Then the omitted-variable bias satisfies the exact identity
\begin{equation}
\psi(\boldsymbol\delta)-\psi_s(\boldsymbol\delta)
=\mathbb{E}\big[\Delta_\mu(\boldsymbol{V},\boldsymbol{W})\,\Delta_\alpha(\boldsymbol{V},\boldsymbol{W})\big].
\label{eq:ovb_identity}
\end{equation}
By Cauchy--Schwarz,
\begin{equation}
\big|\psi_s(\boldsymbol\delta)-\psi(\boldsymbol\delta)\big|
\le B(\boldsymbol\delta)
:=\sqrt{\mathbb{E}[\Delta_\mu^2]}\;\sqrt{\mathbb{E}[\Delta_\alpha^2]}.
\label{eq:ovb_bound_basic}
\end{equation}

It is often useful to isolate the ``degree of adversity''
$$
\varrho(\boldsymbol\delta)
:=\Corr\!\big(\Delta_\mu(\boldsymbol{V},\boldsymbol{W}),\Delta_\alpha(\boldsymbol{V},\boldsymbol{W})\big)\in[-1,1],
$$
so that \eqref{eq:ovb_identity} yields
$$
|\psi_s-\psi|^2 = \varrho(\boldsymbol\delta)^2\,B(\boldsymbol\delta)^2\le B(\boldsymbol\delta)^2.
$$
In our primary analysis and reported bounds, we focus on the worst-case scenario by considering adversarial confounding, which implicitly sets $|\varrho(\boldsymbol{\delta})| = 1$. This allows us to establish a conservative bound and subsequently omit the correlation term from our final formulas.

\subsection{Reparameterization by interpretable partial \texorpdfstring{$R^2$}{R2}}

The bound can be rewritten as a product of an identifiable scale term
and two sensitivity parameters with partial-$R^2$ interpretations. For the
incremental effect contrast
$\theta(\boldsymbol\delta):=\psi(\boldsymbol\delta)-\psi(\boldsymbol{0})$, use the RR contrasts
$$
\alpha_{\theta,\boldsymbol\delta}(\boldsymbol{V},\boldsymbol{W})
:=\alpha_{\boldsymbol\delta}(\boldsymbol{V},\boldsymbol{W})-\alpha_{\boldsymbol{0}}(\boldsymbol{V},\boldsymbol{W}),\qquad
\alpha_{s,\theta,\boldsymbol\delta}(\boldsymbol{X},\boldsymbol{W})
:=\alpha_{s,\boldsymbol\delta}(\boldsymbol{X},\boldsymbol{W})-\alpha_{s,\boldsymbol{0}}(\boldsymbol{X},\boldsymbol{W}).
$$
\paragraph{Identifiable scale.}
Let $\sigma_s^2:=\mathbb{E}\big[(Y-\mu_s(\boldsymbol{X},\boldsymbol{W}))^2\big]$ and
$A_{s,\theta}^2(\boldsymbol\delta):=\mathbb{E}\big[\alpha_{s,\theta,\boldsymbol\delta}(\boldsymbol{X},\boldsymbol{W})^2\big]$.
Define
\begin{equation}
S(\boldsymbol\delta)^2:=\sigma_s^2\;A_{s,\theta}^2(\boldsymbol\delta),
\label{eq:S_def}
\end{equation}
which depends only on the observed-data law of $(Y,\boldsymbol{W},\boldsymbol{X})$ and is therefore
estimable.

\paragraph{Outcome-side sensitivity (partial $R^2$).}
Define
\begin{equation}
C_Y^2
:=\frac{\mathbb{E}[\Delta_\mu(\boldsymbol{V},\boldsymbol{W})^2]}{\mathbb{E}[(Y-\mu_s(\boldsymbol{X},\boldsymbol{W}))^2]}
=\frac{\Var\!\big(\mathbb{E}[Y\mid \boldsymbol{X},\boldsymbol{W},U]\big)-\Var\!\big(\mathbb{E}[Y\mid \boldsymbol{X},\boldsymbol{W}]\big)}
{\Var(Y)-\Var\!\big(\mathbb{E}[Y\mid \boldsymbol{X},\boldsymbol{W}]\big)}\in[0,1],
\label{eq:CY_def}
\end{equation}
i.e.\ the nonparametric partial $R^2$ of $U$ with $Y$ given $(\boldsymbol{X},\boldsymbol{W})$.

\paragraph{Exposure/RR-side sensitivity.}
Because $\alpha_{s,\theta,\boldsymbol\delta}$ is the $L_2$ projection of
$\alpha_{\theta,\boldsymbol\delta}$ onto $(\boldsymbol{X},\boldsymbol{W})$, we have
$$
\mathbb{E}\big[(\alpha_{\theta,\boldsymbol\delta}-\alpha_{s,\theta,\boldsymbol\delta})^2\big]
=\mathbb{E}[\alpha_{\theta,\boldsymbol\delta}^2]-\mathbb{E}[\alpha_{s,\theta,\boldsymbol\delta}^2].
$$
Define the (RR-space) $R^2$:
$$
R_\alpha^2(\boldsymbol\delta)
:=\Corr^2\!\big(\alpha_{\theta,\boldsymbol\delta}(\boldsymbol{V},\boldsymbol{W}),\alpha_{s,\theta,\boldsymbol\delta}(\boldsymbol{X},\boldsymbol{W})\big)
=\frac{\mathbb{E}[\alpha_{s,\theta,\boldsymbol\delta}^2]}{\mathbb{E}[\alpha_{\theta,\boldsymbol\delta}^2]}\in(0,1],
$$
and let
\begin{equation}
C_D^2(\boldsymbol\delta)
:=\frac{\mathbb{E}[\alpha_{\theta,\boldsymbol\delta}^2]-\mathbb{E}[\alpha_{s,\theta,\boldsymbol\delta}^2]}{\mathbb{E}[\alpha_{s,\theta,\boldsymbol\delta}^2]}
=\frac{1-R_\alpha^2(\boldsymbol\delta)}{R_\alpha^2(\boldsymbol\delta)}.
\label{eq:CD_def}
\end{equation}
Thus $1-R_\alpha^2(\boldsymbol\delta)$ is the fraction of RR-contrast variation generated by the latent confounder.

\paragraph{Final bound in $R^2$ form.}
Combining \eqref{eq:ovb_bound_basic}--\eqref{eq:CD_def} yields
\begin{equation}
B(\boldsymbol\delta)^2
=S(\boldsymbol\delta)^2\,C_Y^2\,C_D^2(\boldsymbol\delta),
\qquad
\big|\theta_s(\boldsymbol\delta)-\theta(\boldsymbol\delta)\big|
\le \,S(\boldsymbol\delta)\,C_Y\,C_D(\boldsymbol\delta).
\label{eq:bias_bound_final}
\end{equation}
Accordingly, for any posited $(C_Y^2,R_\alpha^2(\boldsymbol\delta))$, we obtain the sensitivity interval for the incremental effect
\begin{equation}
\theta(\boldsymbol\delta)\in
\Big[\;\theta_s(\boldsymbol\delta)-S(\boldsymbol\delta)C_Y C_D(\boldsymbol\delta),\;
\theta_s(\boldsymbol\delta)+S(\boldsymbol\delta)C_Y C_D(\boldsymbol\delta)\;\Big],
\label{eq:sensitivity_interval}
\end{equation}
where $\theta_s(\boldsymbol\delta):=\psi_s(\boldsymbol\delta)-\psi_s(\boldsymbol{0})$.

\subsection{Implementation}

For each fixed $\boldsymbol\delta$ considered in the paper,
the sensitivity analysis requires three estimated quantities:
\begin{enumerate}
\item $\widehat\theta(\boldsymbol\delta)$: our EIF-based (cross-fitted) estimator of the incremental effect from Section~\ref{sec:estimation}.
\item $\widehat\sigma_s^2 := n^{-1}\sum_{i=1}^n \{Y_i-\widehat\mu_s(\boldsymbol{X}_i,\boldsymbol{W}_i)\}^2$ using the same
cross-fitted $\widehat\mu_s$ as in the main estimator.
\item $\widehat A_{s,\theta}^2(\boldsymbol\delta):=n^{-1}\sum_{i=1}^n \widehat\alpha_{s,\theta,\boldsymbol\delta}(\boldsymbol{X}_i,\boldsymbol{W}_i)^2$,
where $\widehat\alpha_{s,\theta,\boldsymbol\delta}=\widehat\alpha_{s,\boldsymbol\delta}-\widehat\alpha_{s,\boldsymbol{0}}
=\widehat\alpha_{s,\boldsymbol\delta}-1$, and $\widehat\alpha_{s,\boldsymbol\delta}$ is obtained either as the estimated density ratio
$\widehat g_{\boldsymbol\delta}/\widehat f$ or via the closed form \eqref{eq:RR_short_closed_form} by estimating
$\nu_{\boldsymbol\delta}(\boldsymbol{X})=\mathbb{E}[\exp(\boldsymbol\delta^\top\boldsymbol{W})\mid \boldsymbol{X}]$ with regression.
\end{enumerate}
Then $\widehat S(\boldsymbol\delta):=\sqrt{\widehat\sigma_s^2\,\widehat A_{s,\theta}^2(\boldsymbol\delta)}$.

Finally, for any posited sensitivity parameters
$$
\eta_Y^2:=C_Y^2\in[0,1],\qquad \eta_\alpha^2(\boldsymbol\delta):=1-R_\alpha^2(\boldsymbol\delta)\in[0,1),
$$
the bias bound becomes
$$
\widehat B(\boldsymbol\delta;\eta_Y^2,\eta_\alpha^2)
=\widehat S(\boldsymbol\delta)\,
\sqrt{\eta_Y^2}\,
\sqrt{\frac{\eta_\alpha^2(\boldsymbol\delta)}{1-\eta_\alpha^2(\boldsymbol\delta)}}.
$$
Accordingly, the estimated point identified set for the incremental effect is
$$
\theta(\boldsymbol\delta)\in
\Big[\,
\widehat\theta(\boldsymbol\delta)-\widehat B(\boldsymbol\delta;\eta_Y^2,\eta_\alpha^2),
\;
\widehat\theta(\boldsymbol\delta)+\widehat B(\boldsymbol\delta;\eta_Y^2,\eta_\alpha^2)
\Big].
$$
When $\widehat\theta(\boldsymbol\delta)\neq 0$, the ratio
$$
\frac{\widehat B(\boldsymbol\delta;\eta_Y^2,\eta_\alpha^2)}
{|\widehat\theta(\boldsymbol\delta)|}
$$
compares the benchmarked bias half-width with the magnitude of the estimated incremental effect. Table~\ref{tab:benchmark_bias_ratio} reports this ratio under the benchmark specification ($k_Y=1$ and $k_D=1$) established in the main manuscript over the 10 positive Gelbrich targets on each of the 9 intervention paths. Values below one indicate that the benchmarked point bound is narrower than the estimated effect in magnitude, whereas values above one indicate that the benchmarked bias half-width exceeds the estimated effect itself.

\begin{table}[p]
\centering
\caption{Ratio $\widehat B(\boldsymbol\delta)/|\widehat\theta(\boldsymbol\delta)|$ under the formal benchmark with $k_Y=1$ and $k_D=1$. Columns correspond to the positive Gelbrich targets used in the application.}
\label{tab:benchmark_bias_ratio}
\scriptsize
\setlength{\tabcolsep}{4pt}
\resizebox{\textwidth}{!}{%
\begin{tabular}{lrrrrrrrrrr}
\toprule
Shift size & 0.05 & 0.10 & 0.15 & 0.20 & 0.25 & 0.30 & 0.35 & 0.40 & 0.45 & 0.50 \\
\midrule
BC & 0.11 & 0.14 & 0.17 & 0.19 & 0.22 & 0.24 & 0.27 & 0.30 & 0.33 & 0.37 \\
NO$_3$ & 0.24 & 0.36 & 0.48 & 0.59 & 0.69 & 0.77 & 0.84 & 0.90 & 0.96 & 1.02 \\
OM & 0.13 & 0.19 & 0.25 & 0.31 & 0.37 & 0.42 & 0.47 & 0.52 & 0.57 & 0.62 \\
SO$_4$ & 0.07 & 0.08 & 0.08 & 0.09 & 0.09 & 0.10 & 0.11 & 0.11 & 0.12 & 0.12 \\
NH$_4$ & 0.07 & 0.09 & 0.10 & 0.12 & 0.13 & 0.14 & 0.15 & 0.16 & 0.17 & 0.17 \\
BC+OM & 1.64 & 2.35 & 0.83 & 1.56 & 0.73 & 3.19 & 2.68 & 1.79 & 18.25 & 7.85 \\
NO$_3$+SO$_4$+NH$_4$ & 0.16 & 0.16 & 0.27 & 0.22 & 0.27 & 0.35 & 0.31 & 0.35 & 0.41 & 0.36 \\
Efficient & 0.10 & 0.11 & 0.09 & 0.10 & 0.10 & 0.10 & 0.11 & 0.11 & 0.11 & 0.11 \\
BFGS & 0.03 & 0.32 & 0.24 & 0.21 & 0.20 & 0.20 & 0.17 & 0.16 & 0.15 & 0.15 \\
\bottomrule
\end{tabular}}%
\end{table}

\paragraph{Plug-in bounds for the estimated endpoints.}
The empirical results in Section~\ref{sec:application} report pointwise Wald intervals for $\theta(\boldsymbol\delta)$ after fixing the sensitivity parameters. Let $\widehat\varphi_{\theta(\boldsymbol\delta)}(\boldsymbol{Z}_i)$ denote the cross-fitted EIF contribution of $\widehat\theta(\boldsymbol\delta)$. Define the centered plug-in signals
\begin{align*}
\widehat\varphi_{\sigma,i}
&:=
\{Y_i-\widehat\mu_s(\boldsymbol{X}_i,\boldsymbol{W}_i)\}^2-\widehat\sigma_s^2, \\
\widehat\varphi_{A,i}(\boldsymbol\delta)
&:=
\widehat\alpha_{s,\theta,\boldsymbol\delta}(\boldsymbol{X}_i,\boldsymbol{W}_i)^2-\widehat A_{s,\theta}^2(\boldsymbol\delta),
\end{align*}
and write
$$
\lambda(\boldsymbol\delta)
:=
\,\sqrt{\eta_Y^2}\,
\sqrt{\frac{\eta_\alpha^2(\boldsymbol\delta)}{1-\eta_\alpha^2(\boldsymbol\delta)}}.
$$
The delta method gives
\begin{align*}
\widehat\varphi_{S,i}(\boldsymbol\delta)
&:=
\frac{
\widehat A_{s,\theta}^2(\boldsymbol\delta)\,\widehat\varphi_{\sigma,i}
+
\widehat\sigma_s^2\,\widehat\varphi_{A,i}(\boldsymbol\delta)
}{
2\,\widehat S(\boldsymbol\delta)
}, \\
\widehat\varphi_{-,i}(\boldsymbol\delta)
&:=
\widehat\varphi_{\theta(\boldsymbol\delta)}(\boldsymbol{Z}_i)
-\lambda(\boldsymbol\delta)\widehat\varphi_{S,i}(\boldsymbol\delta), \\
\widehat\varphi_{+,i}(\boldsymbol\delta)
&:=
\widehat\varphi_{\theta(\boldsymbol\delta)}(\boldsymbol{Z}_i)
+\lambda(\boldsymbol\delta)\widehat\varphi_{S,i}(\boldsymbol\delta).
\end{align*}
Writing
\begin{align*}
\widehat\theta_-(\boldsymbol\delta)
&:=
\widehat\theta(\boldsymbol\delta)-\widehat B(\boldsymbol\delta;\eta_Y^2,\eta_\alpha^2), \\
\widehat\theta_+(\boldsymbol\delta)
&:=
\widehat\theta(\boldsymbol\delta)+\widehat B(\boldsymbol\delta;\eta_Y^2,\eta_\alpha^2),
\end{align*}
the corresponding standard errors are
\begin{align*}
\widehat{\mathrm{se}}_-(\boldsymbol\delta)
&:=
\sqrt{\frac{1}{n^2}\sum_{i=1}^n \widehat\varphi_{-,i}(\boldsymbol\delta)^2}, \\
\widehat{\mathrm{se}}_+(\boldsymbol\delta)
&:=
\sqrt{\frac{1}{n^2}\sum_{i=1}^n \widehat\varphi_{+,i}(\boldsymbol\delta)^2}.
\end{align*}
The corresponding pointwise Wald interval is obtained by combining the lower and upper endpoint bounds:
$$
\left[
\widehat\theta_-(\boldsymbol\delta)-z_{0.95}\widehat{\mathrm{se}}_-(\boldsymbol\delta),
\;
\widehat\theta_+(\boldsymbol\delta)+z_{0.95}\widehat{\mathrm{se}}_+(\boldsymbol\delta)
\right].
$$
These intervals are used as pointwise plug-in uncertainty summaries; a formal coverage theorem would require a separate asymptotic analysis of the plug-in bias bound $\widehat B(\boldsymbol\delta;\eta_Y^2,\eta_\alpha^2)$, which we leave to future work.

\paragraph{Selection of sensitivity parameters.}
The parameters $\eta_Y^2$ and $\eta_\alpha^2(\boldsymbol\delta)$ have direct interpretations:
$\eta_Y^2$ is the maximal fraction of residual outcome variance explainable by $U$ given $(\boldsymbol{X},\boldsymbol{W})$,
and $\eta_\alpha^2(\boldsymbol\delta)$ is the maximal fraction of RR variation explainable by $U$ for the policy
indexed by $\boldsymbol\delta$. These can be calibrated by subject-matter knowledge and by benchmarking against
observed covariates, and then used in the bound and the endpoint bounds above.

\subsection{Formal benchmarking and calibration}
\label{sec:benchmark_formal}

Following the approach of \citep{chernozhukov2021long}, we calibrate the sensitivity parameters on the $f^2$ scale induced by nested linear projections. This construction translates the observed contribution of a single covariate $X_j$ into a benchmark that is commensurate with the omitted-variable calibration in the bias bound. We carry out this comparison for each of the 22 observed covariates.

\paragraph{Outcome-side benchmark.}
For each observed covariate $X_j$, let
\begin{align*}
\widehat\sigma_s^2
&:=
\min_{a,\boldsymbol{b},\boldsymbol{c}}
\frac{1}{n}\sum_{i=1}^n
\Big\{
Y_i-a-\boldsymbol{b}^\top \boldsymbol{W}_i-\boldsymbol{c}^\top \boldsymbol{X}_i
\Big\}^2, \\
\widehat\sigma_{s,-j}^2
&:=
\min_{a,\boldsymbol{b},\boldsymbol{c}}
\frac{1}{n}\sum_{i=1}^n
\Big\{
Y_i-a-\boldsymbol{b}^\top \boldsymbol{W}_i-\boldsymbol{c}^\top \boldsymbol{X}_{i,-j}
\Big\}^2,
\end{align*}
where $\boldsymbol{X}_{i,-j}$ denotes the observed covariate vector with $X_{ij}$ removed. The associated benchmark statistics are
\begin{align*}
\widehat\eta_{Y,j}^2
&:=
\frac{\widehat\sigma_{s,-j}^2-\widehat\sigma_s^2}{\widehat\sigma_{s,-j}^2}, \\
\widehat f_{Y,j}^2
&:=
\frac{\widehat\eta_{Y,j}^2}{1-\widehat\eta_{Y,j}^2}
=
\frac{\widehat\sigma_{s,-j}^2-\widehat\sigma_s^2}{\widehat\sigma_s^2}.
\end{align*}
Thus $\widehat\eta_{Y,j}^2$ is the observed partial $R^2$ for $X_j$ after adjusting for $(\boldsymbol{W},\boldsymbol{X}_{-j})$, and $\widehat f_{Y,j}^2$ expresses the same gain relative to the residual variation in the full projection. Table~\ref{tab:benchmark_outcome} reports these quantities for all 22 covariates. The largest outcome-side benchmark is \texttt{White}, with $\widehat\eta_{Y,j}^2=0.0642$ and $\widehat f_{Y,j}^2=0.0686$.

\paragraph{RR-side benchmark.}
For the RR side, we first evaluate the fitted short RR $\widehat\alpha_{s,\boldsymbol\delta}(\boldsymbol{X}_i,\boldsymbol{W}_i)$ at every intervention target retained in the application analysis. For each $X_j$ and each target $\boldsymbol\delta$, we then compute the nested linear projections
\begin{align*}
\widehat r_{\alpha,\mathrm{full},j}^2(\boldsymbol\delta)
&:=
\min_{a,\boldsymbol{b}}
\frac{1}{n}\sum_{i=1}^n
\Big\{
\widehat\alpha_{s,\boldsymbol\delta}(\boldsymbol{X}_i,\boldsymbol{W}_i)-a-\boldsymbol{b}^\top \boldsymbol{X}_i
\Big\}^2, \\
\widehat r_{\alpha,\mathrm{red},j}^2(\boldsymbol\delta)
&:=
\min_{a,\boldsymbol{b}}
\frac{1}{n}\sum_{i=1}^n
\Big\{
\widehat\alpha_{s,\boldsymbol\delta}(\boldsymbol{X}_i,\boldsymbol{W}_i)-a-\boldsymbol{b}^\top \boldsymbol{X}_{i,-j}
\Big\}^2.
\end{align*}
The corresponding pointwise benchmark statistics are
\begin{align*}
\widehat\eta_{\alpha,j}^2(\boldsymbol\delta)
&:=
\frac{
\widehat r_{\alpha,\mathrm{red},j}^2(\boldsymbol\delta)
-
\widehat r_{\alpha,\mathrm{full},j}^2(\boldsymbol\delta)
}{
\widehat r_{\alpha,\mathrm{red},j}^2(\boldsymbol\delta)
}, \\
\widehat f_{\alpha,j}^2(\boldsymbol\delta)
&:=
\frac{\widehat\eta_{\alpha,j}^2(\boldsymbol\delta)}{1-\widehat\eta_{\alpha,j}^2(\boldsymbol\delta)}
=
\frac{
\widehat r_{\alpha,\mathrm{red},j}^2(\boldsymbol\delta)
-
\widehat r_{\alpha,\mathrm{full},j}^2(\boldsymbol\delta)
}{
\widehat r_{\alpha,\mathrm{full},j}^2(\boldsymbol\delta)
}.
\end{align*}
This is the direct analogue of the outcome-side construction, with the fitted short RR taking the place of the observed outcome.

The benchmark reported in Section~\ref{sec:application} is attached to an intervention path rather than to a single target. For each scenario, we therefore rank the 22 covariates by the average of $\widehat f_{\alpha,j}^2(\boldsymbol\delta)$ over the positive Gelbrich targets on that path. This average is the RR-side score used in the formal benchmark. It summarizes how much adding $X_j$ improves the linear approximation to the fitted short RR along the displayed path, and it keeps the benchmark tied to an intervention curve that is actually reported in the application. Table~\ref{tab:benchmark_rr} reports these scenario-level mean $\widehat f_{\alpha,j}^2$ values for all 22 covariates.

With this aggregation, the selected RR-side benchmark covariates are \texttt{Housing Renter} for BC; \texttt{Housing More People Units} for NO$_3$ and OM; \texttt{Households Smartphone} for SO$_4$ and Efficient; \texttt{Unemployed} for NH$_4$; \texttt{Housing No Vehicle} for BC+OM and NO$_3$+SO$_4$+NH$_4$; and \texttt{Housing Vacant} for BFGS.

\paragraph{Benchmark calibration.}
For the selected outcome-side covariate,
$$
\eta_Y^2
=
\frac{k_Y\,\widehat f_{Y,j}^2}{1+k_Y\,\widehat f_{Y,j}^2}.
$$
For the selected RR-side covariate in a given scenario,
$$
\eta_\alpha^2(\boldsymbol\delta)
=
\frac{k_D\,\widehat f_{\alpha,j}^2(\boldsymbol\delta)}{1+k_D\,\widehat f_{\alpha,j}^2(\boldsymbol\delta)}.
$$
These calibrated values are inserted directly into the point bounds and the endpoint bounds above. The main application figure uses $k_Y=1$ and $k_D=1$, so the omitted confounder is calibrated to the observed benchmark strength on both the outcome side and the RR side.

\begin{table}[p]
\centering
\caption{Outcome-side benchmarking statistics for the 22 observed covariates. The table reports the observed partial $R^2$ and its $f^2$ transform from the nested linear projections of $Y$ on $(\boldsymbol{W},\boldsymbol{X})$ and $(\boldsymbol{W},\boldsymbol{X}_{-j})$.}
\label{tab:benchmark_outcome}
\begin{tabular}{lrr}
\toprule
Covariate & $\widehat\eta_{Y,j}^2$ & $\widehat f_{Y,j}^2$ \\
\midrule
White & 0.0642 & 0.0686 \\
Poverty & 0.0341 & 0.0353 \\
Physical Activity & 0.0322 & 0.0333 \\
Housing More People Units & 0.0151 & 0.0154 \\
Binge Drinking & 0.0132 & 0.0134 \\
Households No Internet & 0.0123 & 0.0124 \\
Housing No Vehicle & 0.0112 & 0.0113 \\
Housing 10 Units & 0.0087 & 0.0087 \\
Median Income & 0.0082 & 0.0083 \\
Cigarette Smoking & 0.0052 & 0.0052 \\
Low Education Computer No Internet & 0.0042 & 0.0042 \\
Percentage No Insurance & 0.0037 & 0.0037 \\
Male & 0.0032 & 0.0032 \\
Households Smartphone & 0.0023 & 0.0023 \\
Housing Renter & 0.0011 & 0.0011 \\
Housing Vacant & 0.0010 & 0.0010 \\
HS Higher & 0.0006 & 0.0006 \\
Housing Mobile & 0.0001 & 0.0001 \\
Obesity & 0.0001 & 0.0001 \\
Unemployed & 0.0001 & 0.0001 \\
Households Low Income No Internet & 0.0001 & 0.0001 \\
Households Only Smartphone & 0.0000 & 0.0000 \\
\bottomrule
\end{tabular}
\end{table}

\begin{table}[p]
\centering
\caption{RR-side benchmarking statistics for the 22 observed covariates. Entries are the scenario-level means of $\widehat f_{\alpha,j}^2(\boldsymbol\delta)$ over the positive Gelbrich grid within each scenario, computed from nested linear projections of $\widehat\alpha_{s,\boldsymbol\delta}(\boldsymbol{X}_i,\boldsymbol{W}_i)$ on $\boldsymbol{X}$ and $\boldsymbol{X}_{-j}$. The benchmark established in the main manuscript with $k_D=1$ selects the largest entry within each scenario.}
\label{tab:benchmark_rr}
\scriptsize
\setlength{\tabcolsep}{4pt}
\resizebox{\textwidth}{!}{%
\begin{tabular}{lrrrrrrrrr}
\toprule
Covariate & BC & NO$_3$ & OM & SO$_4$ & NH$_4$ & BC+OM & NO$_3$+SO$_4$+NH$_4$ & Efficient & BFGS \\
\midrule
White & 0.0000 & 0.0010 & 0.0046 & 0.0016 & 0.0001 & 0.0002 & 0.0001 & 0.0001 & 0.0027 \\
Poverty & 0.0010 & 0.0000 & 0.0006 & 0.0016 & 0.0005 & 0.0003 & 0.0007 & 0.0003 & 0.0047 \\
Physical Activity & 0.0002 & 0.0005 & 0.0014 & 0.0001 & 0.0008 & 0.0001 & 0.0003 & 0.0000 & 0.0004 \\
Housing More People Units & 0.0012 & 0.0083 & 0.0052 & 0.0001 & 0.0010 & 0.0019 & 0.0002 & 0.0007 & 0.0028 \\
Binge Drinking & 0.0002 & 0.0011 & 0.0001 & 0.0001 & 0.0011 & 0.0003 & 0.0001 & 0.0001 & 0.0033 \\
Households No Internet & 0.0000 & 0.0003 & 0.0002 & 0.0009 & 0.0003 & 0.0005 & 0.0002 & 0.0003 & 0.0002 \\
Housing No Vehicle & 0.0006 & 0.0012 & 0.0024 & 0.0000 & 0.0001 & 0.0020 & 0.0061 & 0.0006 & 0.0006 \\
Housing 10 Units & 0.0004 & 0.0002 & 0.0000 & 0.0001 & 0.0001 & 0.0005 & 0.0015 & 0.0002 & 0.0001 \\
Median Income & 0.0000 & 0.0003 & 0.0002 & 0.0001 & 0.0001 & 0.0002 & 0.0003 & 0.0002 & 0.0033 \\
Cigarette Smoking & 0.0015 & 0.0015 & 0.0031 & 0.0005 & 0.0009 & 0.0007 & 0.0000 & 0.0004 & 0.0036 \\
Low Education Computer No Internet & 0.0002 & 0.0004 & 0.0000 & 0.0015 & 0.0005 & 0.0001 & 0.0001 & 0.0004 & 0.0003 \\
Percentage No Insurance & 0.0015 & 0.0003 & 0.0000 & 0.0000 & 0.0003 & 0.0009 & 0.0001 & 0.0002 & 0.0002 \\
Male & 0.0001 & 0.0000 & 0.0024 & 0.0009 & 0.0004 & 0.0000 & 0.0022 & 0.0000 & 0.0015 \\
Households Smartphone & 0.0007 & 0.0001 & 0.0005 & 0.0028 & 0.0016 & 0.0002 & 0.0001 & 0.0023 & 0.0004 \\
Housing Renter & 0.0051 & 0.0017 & 0.0008 & 0.0005 & 0.0008 & 0.0019 & 0.0008 & 0.0012 & 0.0001 \\
Housing Vacant & 0.0002 & 0.0001 & 0.0001 & 0.0008 & 0.0013 & 0.0002 & 0.0011 & 0.0003 & 0.0110 \\
HS Higher & 0.0018 & 0.0001 & 0.0005 & 0.0003 & 0.0005 & 0.0004 & 0.0006 & 0.0001 & 0.0004 \\
Housing Mobile & 0.0002 & 0.0004 & 0.0005 & 0.0004 & 0.0006 & 0.0001 & 0.0000 & 0.0002 & 0.0032 \\
Obesity & 0.0039 & 0.0000 & 0.0012 & 0.0004 & 0.0000 & 0.0018 & 0.0001 & 0.0002 & 0.0003 \\
Unemployed & 0.0022 & 0.0007 & 0.0038 & 0.0015 & 0.0020 & 0.0006 & 0.0013 & 0.0015 & 0.0010 \\
Households Low Income No Internet & 0.0001 & 0.0006 & 0.0003 & 0.0005 & 0.0004 & 0.0001 & 0.0001 & 0.0001 & 0.0007 \\
Households Only Smartphone & 0.0003 & 0.0004 & 0.0000 & 0.0000 & 0.0000 & 0.0001 & 0.0001 & 0.0002 & 0.0003 \\
\bottomrule
\end{tabular}}%
\end{table}

\end{document}